\documentclass[preprint]{imsart} % ,preprint ejs,
\RequirePackage[OT1]{fontenc}
\usepackage{amsthm,amsmath,amstext,amsfonts,amsbsy,amssymb}
\usepackage{enumerate,multirow,booktabs}
\usepackage[figuresright]{rotating}
\RequirePackage[numbers]{natbib}
\RequirePackage[colorlinks,citecolor=blue,urlcolor=blue]{hyperref}

\startlocaldefs
\numberwithin{equation}{section}
\theoremstyle{plain}
\newtheorem{thm}{Theorem}[section]
\endlocaldefs

\newcommand{\ud}{\mathrm{d}}
\newcommand{\E}{{\mathbb{E}}}

\newcommand{\Cov}{{\mathbb{C}}\mathrm{ov}}
\newcommand{\Pro}{\mathrm{Pr}}

\newcommand{\TPR}{\mathrm{TPR}}
\newcommand{\FPR}{\mathrm{FPR}}

\newcommand{\TCF}{\mathrm{TCF}}
\newcommand{\I}{\mathrm{I}}

\begin{document}

\begin{frontmatter}

% "Title of the Paper"
\title{Bias--corrected methods for estimating the receiver operating characteristic surface of continuous diagnostic tests}
%\thanksref{t1}}
%\thankstext{t1}{This is an original survey paper}
\runtitle{Bias--corrected methods for ROC surfaces}

\begin{aug}
\author{\fnms{Khanh} \snm{To Duc}\thanksref{t1} 
\ead[label = e1]{toduc@stat.unipd.it},}
\author{\fnms{Monica} \snm{Chiogna}
%\thanksref{t3}
\ead[label = e2]{monica.chiogna@unipd.it}}
\and
\author{\fnms{Gianfranco} \snm{Adimari}
\ead[label = e3]{gianfranco.adimari@unipd.it}}
\address{Department of Statistical Sciences, University of Padua \\
Via C. Battisti, 241-243, 35121 Padua, Italy\\
\printead{e1,e2,e3}}

\thankstext{t1}{Corresponding author}
%\thankstext{t2}{First supporter of the project}
%\thankstext{t3}{Second supporter of the project}

\runauthor{K. To Duc et al.}

\affiliation{Department of Statistical Sciences, University of Padua}

\end{aug}

\begin{abstract}
Verification bias is a well-known problem that may occur in the evaluation of predictive ability of diagnostic tests. When a binary disease status is considered, various solutions can be found in the literature to correct inference based on usual measures of test accuracy, such as the receiver operating characteristic (ROC) curve or the area underneath. Evaluation of the predictive ability of continuous diagnostic tests in the presence of verification bias for  a three-class disease status is here discussed. In particular, several verification bias-corrected estimators of the ROC surface and of the volume underneath are proposed. Consistency and asymptotic normality of the proposed estimators are established and their finite sample behavior is investigated by means of  Monte Carlo simulation studies. Two illustrations are also given.
\end{abstract}

\begin{keyword}[class=MSC]
\kwd[Primary ]{62C99}
\kwd{}
\kwd[; secondary ]{62P10}
\end{keyword}

\begin{keyword}
\kwd{verification bias}
\kwd{missing at random}
\kwd{ROC surface analysis}
\kwd{true class fractions}
\end{keyword}

% history:
% \received{\smonth{1} \syear{0000}}

\tableofcontents

\end{frontmatter}

\section{Introduction}
\label{s:intro}
Before applying a diagnostic test in clinical settings, a rigorous statistical assessment of its performance in discriminating the disease status from the non-disease status is necessary. For a continuous-scale test $T$, the diagnosis is dependent upon whether the test result is above or below a specified cut  point $c.$ Assuming, without loss of generality, that higher test values indicate a higher likelihood of disease,  a result is called positive if its value exceeds the cut point, and negative otherwise.  A positive test indicates presence of the disease.

At a fixed cut point $c$, the accuracy of the test can be evaluated by its true positive rate (TPR) and its true negative rate (TNR), which are defined as the probabilities that the test correctly identifies the diseaded and non-diseaded subjects, respectively. 
The Receiver Operating Characteristic (ROC) curve
is the plot of TPR versus 1-TNR by varying the cut point $c$.
Usually, the ROC curve is monotone and lies in the upper triangle of the unit square, which consist of three vertices $(0,0), (0,1)$ and $(1,1)$. The shape of ROC curve allows to evaluate the ability of the test. For example, a ROC curve equal to a straight line joining points $(0,0)$ and $(1,1)$ represents a diagnostic tests which is the random guess. A commonly used summary measure that aggregates performance information across the range of possible cut points  is the area under ROC curve (AUC). Reasonable values of AUC range from 0.5, suggesting that the test is no better than chance alone, to 1.0, which indicates a perfect test. 

Clearly, the ROC curve and the AUC of a test under assessment are unknown and the statistical evaluation of the test requires suitable inferential procedures. See, for example, \citet{zho} and \citet{pep} as general references. In principle, knowledge of the true disease status of the subjects under study is obtained by the most accurate available test, called gold standard (GS) test. In practice, there may be many drawbacks to implement the GS test, which can be too expensive, or too invasive, or both for regular use. Thus, often only a subset of patients undergoes disease verification and the decision to send a patient to verification typically depends on the test result and other patient characteristic. Statistical evaluations based only on data from subjects with verified disease status are typically biased. This bias is known as verification bias. 

In the last fifteen years, various methods have been developed to deal with the verification bias problem, most of which assume that the true disease status, if missing, is missing at random (MAR, \citet{lit}). Among the others, we cite the papers by \cite{adi:15}, \cite{alo:05}, \cite{flu}, \cite{gu} and \cite{rot}. %\cite{zhe},
In particular, \citet{alo:05} proposed four types of partially parametric estimators of TPR and TNR, i.e., full imputation (FI), mean score imputation (MSI), inverse probability weighting (IPW) and semiparametric efficient (SPE) estimators.

In some medical studies, however, the disease status often involves more than two categories; for example, Alzheimer's dementia can be classified into three categories (see \cite{chi} for more details). 
In such situations, quantities used to evaluate the accuracy of a diagnostic test are the true class fractions (TCF's). These are well defined as a generalization of TPR and TNR. In a three--class diagnostic problem, given a pair of cut points $(c_1,c_2)$,  with $c_1 < c_2$, subjects are classified into class 1 if $T < c_1$; class 2 if $c_1 \le T < c_2$; and class 3 otherwise. The true class fractions of the test $T$ at $(c_1,c_2)$  are defined as
\begin{eqnarray}
\TCF_1(c_1) &=& \Pro(T < c_1|\mathrm{class}\, 1) = 1 - \Pro(T \ge c_1|\mathrm{class}\, 1), \nonumber \\
\TCF_2(c_1,c_2) &=& \Pro(c_1 < T < c_2|\mathrm{class}\, 2) \nonumber \\
&=& \Pro(T \ge c_1|\mathrm{class}\, 2) - \Pro(T \ge c_2|\mathrm{class}\, 2), \nonumber \\
\TCF_3(c_2) &=& \Pro(T > c_2|\mathrm{class}\, 3) = \Pro(T \ge c_2|\mathrm{class}\, 3) \nonumber.
\end{eqnarray}
The plot of (TCF$_1$, TCF$_2$, TCF$_3$) by varying the pair $(c_1,c_2)$ produces the ROC surface of $T$ in the unit cube. \citet{scu} and \citet{nak:04} mentioned that a ROC surface is well defined as a generalization of the ROC curve. Indeed, the projection of ROC surface to the plane defined by TCF$_2$ versus TCF$_1$ yields the ROC curve between classes 1 and 2. Similarly, on projecting the  ROC surface to the plane defined by the axes TCF$_2$ and TCF$_3$, the ROC curve between classes 2 and 3 is produced (see also \cite{nak:14}). The ROC surface will be the triangular plane with vertices $(0, 0, 1), (0, 1, 0)$, and $(1, 0, 0)$ if all of three TCF's are equal for every pair $(c_1,c_2)$. In this case, we say that the diagnostic test is the random guess, again. In practice, one can imagine that the graph of ROC surface lies in the unit cube and above the plane of the triangle with three vertices $(0, 0, 1), (0, 1, 0)$, and $(1, 0, 0)$. A summary of the overall diagnostic accuracy of the test under consideration is the volume under the ROC surface (VUS) which can be seen as a generalization of the AUC. Reasonable values of VUS vary from 1/6 to 1, ranging from bad to perfect diagnostic tests.

\citet{nak:04} and \citet{nak:14} gave some interesting results about ROC surface analysis in absence of verification bias. In particular, the authors formularized the ROC surface by a functional form and proposed a nonparametric approach for VUS estimation. Again without verification bias, parametric estimation of VUS is supplied in the work of \citet{xio}, where the assumption of normality distribution was used, whereas \citet{li} tackled the nonparametric and semi-parametric estimation of the ROC surface. \citet{china} proposed a regression approach to ROC surface, and in \citet{kang} a kernel smoothing based approach for estimation of VUS is employed. 

The issue of correcting for the verification bias in ROC surface analysis is very scarcely considered in the statistical literature. Until now, only \citet{chi} discussed about the issue. The authors proposed  maximum likelihood estimates for ROC surface and VUS. However, these results only concern ordinal diagnostic tests. This motivated us to develop bias-corrected methods for continuous diagnostic tests with three--class disease status.

In this paper, we propose several verification bias-corrected estimators of TCF$_1$, TCF$_2$ and TCF$_3$ for continuous diagnostic tests. The proposed estimators are the extension of FI, MSI, IPW and SPE estimators for the ROC curves in \cite{alo:05}. The new estimators allow to obtain bias-corrected ROC surfaces. Corresponding estimators of the VUS are also presented. Consistency and asymptotic normality of the proposed estimators are established under the MAR assumption.
 
The rest of paper is organized as follows. In Section 2, we review the estimators of ROC curves discussed in \cite{alo:05}.  The proposed extension, giving bias-corrected estimators of the ROC surface and of VUS, is  presented in Section 3, along with the relevant asymptotic results. In Section 4, some simulation results are produced and two applications of the methods are contained in Section 5. Finally, conclusions are drawn in Section 6.

\section{Background}
\label{s:existing}
In this section, we review  the approaches presented in \cite{alo:05} for bias--corrected ROC analysis in two--class problems.
\subsection{Notation and assumption}
Let us consider a study with $n$ subjects, for whom the result of a continuous test $T$ is available.  The patient's true condition (or disease status), $D$, is defined by a GS test. $D$ is a binary variable, that is $0$ if the subject is healthy and $1$ in case of disease. Further, let $V$ be a binary verification status of a patient, such that $V = 1$ if he/she is underwent the GS test, and $V = 0$ otherwise. In practice, some information, other than the test results, can be obtained for each patient. Let $A$ be a covariate vector for a patient, that may be associated with both $D$ and $V$. For the sake of reader's convenience, without loss of generality, in what follows we will consider $A$ to be univariate.
We assume that the verification status $V$ and the response $D$ are mutually independent given the test result $T$ and covariate $A$, i.e.,  $\Pro (V|T,A) = \Pro (V|D,T,A)$ or, equivalently, $\Pro (D|T,A) = \Pro (D|V,T,A)$. This assumption corresponds to the MAR assumption.
\subsection{Bias correction for ROC curve}\label{bc:curve}
Let FPR = 1-TNR. When all subjects are verified by GS, we have a full (or complete) data set. For a given cut point $c$, TPR and FPR are
\begin{eqnarray}
\TPR(c) &=& \Pro(T \ge c|D = 1) = \frac{\Pro(T \ge c ,D = 1)}{\Pro(D = 1)} = \frac{\beta_1}{\theta}, \label{beta:theta} \\
\FPR(c) &=& \Pro(T \ge c|D = 0) = \frac{\Pro(T \ge c ,D = 0)}{\Pro(D = 0)} = \frac{\beta_0}{1 - \theta} \nonumber.
\end{eqnarray}
Then, one can employ the empirical estimators $\hat{\beta}_0,\hat{\beta}_1$ and $\hat{\theta}$ to obtain the nonparametric estimators of TPR and FPR 
\begin{eqnarray}\label{npfull}
\widehat{\TPR}(c) = \frac{\hat{\beta}_1}{\hat{\theta}} = \frac{\sum\limits_{i=1}^{n}\mathrm{I}(T_i \ge c)D_i}{\sum\limits_{i=1}^{n} D_i}, \,
\widehat{\FPR}(c) = \frac{\hat{\beta}_0}{1 - \hat{\theta}} = \frac{\sum\limits_{i=1}^{n}\mathrm{I}(T_i \ge c)(1-D_i)}{\sum\limits_{i=1}^{n} (1-D_i)},
\end{eqnarray}
where $\mathrm{I}(\cdot)$ is the indicator function. 

If not all patients have their disease status verified, the nonparametric estimators (\ref{npfull}) can not be computed. If one computes the Na\"{i}ve estimators, i.e., estimators (\ref{npfull}) based only on verified subjects, typically gets estimates that are 
biased and inconsistent.

\citet{alo:05} proposed four partially parametric estimators to assess  continuous diagnostic (or screening) tests under the MAR assumption. 
In particular, FI estimators of $\TPR(c)$ and $\FPR(c)$ are
\[
\widehat{\TPR}_{\mathrm{FI}}(c) = \frac{\hat{\beta}_{1,\mathrm{FI}}}{\hat{\theta}_{\mathrm{FI}}} =  \frac{\sum\limits_{i=1}^{n}\mathrm{I}(T_i \ge c)\hat{\rho}_i}{\sum\limits_{i=1}^{n} \hat{\rho}_i}, \,
\widehat{\FPR}_{\mathrm{FI}}(c) = \frac{\hat{\beta}_{0,\mathrm{FI}}}{1 - \hat{\theta}_{\mathrm{FI}}} = \frac{\sum\limits_{i=1}^{n}\mathrm{I}(T_i \ge c)(1-\hat{\rho}_i)}{\sum\limits_{i=1}^{n} (1-\hat{\rho}_i)}.
\]
Here, the estimates $\hat{\rho}_i$ of $\rho_{i} = \Pro(D_i =1 | T_i,A_i)$ are obtained by using some suitable parametric model (e.g., logistic regression model) computed from verified subjects. MSI estimators only imputes the disease status for subjects who did not undergo the GS, resulting to be  
\begin{eqnarray}
\widehat{\TPR}_{\mathrm{MSI}}(c) &=& \frac{\hat{\beta}_{1,\mathrm{MSI}}}{\hat{\theta}_{\mathrm{MSI}}} =  \frac{\sum\limits_{i=1}^{n}\mathrm{I}(T_i \ge c)\left\{V_i D_i + (1-V_i)\hat{\rho}_i\right\}}{\sum\limits_{i=1}^{n} \left\{V_i D_i + (1-V_i)\hat{\rho}_i\right\}}, \label{est:msi:curve} \\
\widehat{\FPR}_{\mathrm{MSI}}(c) &=& \frac{\hat{\beta}_{0,\mathrm{MSI}}}{1 - \hat{\theta}_{\mathrm{MSI}}} =  \frac{\sum\limits_{i=1}^{n}\mathrm{I}(T_i \ge c)\left\{V_i (1-D_i) + (1-V_i)(1-\hat{\rho}_i)\right\}}{\sum\limits_{i=1}^{n} \left\{V_i (1-D_i) + (1-V_i)(1-\hat{\rho}_i)\right\}}. \nonumber
\end{eqnarray}
IPW method weights each verified subject by the inverse of the conditional verification probability $\pi_i = \Pro (V_i = 1|T_i,A_i)$ (i.e. the probability that the subject is selected for verification). Therefore, the estimators are
\begin{eqnarray}
\widehat{\TPR}_{\mathrm{IPW}}(c) &=& \frac{\hat{\beta}_{1,\mathrm{IPW}}}{\hat{\theta}_{\mathrm{IPW}}} =  \frac{\sum\limits_{i=1}^{n}\mathrm{I}(T_i \ge c)V_i D_i \hat{\pi}_i^{-1}}{\sum\limits_{i=1}^{n} V_i D_i \hat{\pi}_i^{-1}}, \label{est:ipw:curve} \\
\widehat{\FPR}_{\mathrm{IPW}}(c) &=& \frac{\hat{\beta}_{0,\mathrm{IPW}}}{1 - \hat{\theta}_{\mathrm{IPW}}} = \frac{\sum\limits_{i=1}^{n}\mathrm{I}(T_i \ge c)V_i (1-D_i) \hat{\pi}_i^{-1}}{\sum\limits_{i=1}^{n} V_i (1-D_i) \hat{\pi}_i^{-1}}.\nonumber
\end{eqnarray}
Again, the estimates $\hat{\pi}_i$ need to be obtained by using parametric regression models such as logistic or probit models. Finally, SPE estimators are defined as:
\begin{align}
\widehat{\TPR}_{\mathrm{SPE}}(c) &= \frac{\hat{\beta}_{1,\mathrm{SPE}}}{\hat{\theta}_{\mathrm{SPE}}} =  \frac{\sum\limits_{i=1}^{n}\mathrm{I}(T_i \ge c)\left\{V_i D_i\hat{\pi}_i^{-1} - (V_i \hat{\pi}_i^{-1} - 1)\hat{\rho}_i \right\}} {\sum\limits_{i=1}^{n} \left\{V_i D_i\hat{\pi}_i^{-1} - (V_i \hat{\pi}_i^{-1} - 1)\hat{\rho}_i \right\}}, \label{est:spe:curve} \\
\widehat{\FPR}_{\mathrm{SPE}}(c) &= \frac{\hat{\beta}_{0,\mathrm{SPE}}}{1 - \hat{\theta}_{\mathrm{SPE}}} = \frac{\sum\limits_{i=1}^{n}\mathrm{I}(T_i \ge c)\left\{V_i (1-D_i) \hat{\pi}_i^{-1} - (V_i \hat{\pi}_i^{-1} - 1) (1 - \hat{\rho}_i)\right\}}{\sum\limits_{i=1}^{n} \left\{V_i (1-D_i) \hat{\pi}_i^{-1} - (V_i \hat{\pi}_i^{-1} - 1)(1-\hat{\rho}_i)\right\}}. \nonumber
\end{align}
Alonzo and Pepe \cite{alo:05} showed that SPE estimators are doubly robust because they are consistent if either 
the $\pi_i$'s or the $\rho_i$'s are consistently estimated.   However, it is worth noting that SPE estimates may not be range-respecting, i.e., they could fall outside the interval $(0,1)$. This happens because the quantities $\left\{V_i D_i\hat{\pi}_i^{-1} - (V_i \hat{\pi}_i^{-1} - 1)\hat{\rho}_i\right\}$ or $\left\{V_i (1 - D_i) \hat{\pi}_i^{-1} - (V_i \hat{\pi}_i^{-1} - 1)(1 - \hat{\rho}_i)\right\}$ can be negative.

For each of the above methods, an estimated bias-corrected ROC curve can be obtained by plotting $\widehat{\TPR}(c)$ versus $\widehat{\FPR}(c)$ for all cut points $c$.

\section{Proposal}
\label{s:extension}
Consider now a three--class problem. We model the  disease status by a trinomial random vector $D = (D_1,D_2,D_3)$, such that $D_{k}$ is a Bernoulli random variable having  mean $\theta_k = \Pro(D_k = 1)$, with $\theta_1 + \theta_2 + \theta_3 = 1$. Let $\beta_{jk} = \Pro(T \ge c_j, D_k = 1)$ with $j = 1,2$ and $k = 1,2,3$. In this notation,
\begin{eqnarray}
\TCF_1(c_1) &=& 1 - \frac{\Pro(T \ge c_1,D_1 = 1)}{\Pro(D_1 = 1)} = 1 - \frac{\beta_{11}}{\theta_1}, \nonumber \\
\TCF_2(c_1,c_2) &=& \frac{\Pro(T \ge c_1,D_2 = 1) - \Pro(T \ge c_2,D_2 = 1)}{\Pro(D_2 = 1)} = \frac{\beta_{12} - \beta_{22}}{\theta_2}, \nonumber \\
\TCF_3(c_2) &=& \frac{\Pro(T \ge c_2,D_3 = 1)}{\Pro(D_3 = 1)} = \frac{\beta_{23}}{\theta_3} \label{est:original}.
\end{eqnarray}

When all subjects are verified, the nonparametric estimators of $\TCF_1, \TCF_2$ and $\TCF_3$ are given by
\begin{align}
%\begin{eqnarray}
\widehat{\TCF}_1(c_1) &= 1 - \frac{\sum\limits_{i=1}^{n}\mathrm{I}(T_i \ge c_1)D_{1i}}{\sum\limits_{i=1}^{n}D_{1i}} \nonumber \\
\widehat{\TCF}_2(c_1,c_2) &= \frac{\sum\limits_{i=1}^{n}\left\{\mathrm{I}(T_i \ge c_1) - \mathrm{I}(T_i \ge c_2)\right\}D_{2i}}{\sum\limits_{i=1}^{n}D_{2i}}  \nonumber \displaybreak[1] \\
\widehat{\TCF}_3(c_2) &= \frac{\sum\limits_{i=1}^{n}\mathrm{I}(T_i \ge c_2)D_{3i}}{\sum\limits_{i=1}^{n}D_{3i}} \nonumber.
%\end{eqnarray}
\end{align}

In the presence of verification bias, we propose  four estimators for $\TCF_1(c_1)$, $\TCF_2(c_1,c_2)$ and $\TCF_3(c_2).$  The proposed estimators work under the MAR assumption  
and are based on the estimation of the quantities $\theta_1, \theta_2, \beta_{11}, \beta_{12}, \beta_{22}$ and $\beta_{23}.$
They can be seen as an extension of estimators reviewed in Subsection \ref{bc:curve}. In expressions (\ref{beta:theta}) and (\ref{est:original}), we note that parameters $\theta$ and $\theta_k,$  so as  $\beta_1$ and $\beta_{jk},$   play, in essence, a similar role. Therefore, estimates of $\theta_k$  and $\beta_{jk}$  can be obtained by mimicking what was done in the two-class problem.
\subsection{Full imputation}\label{sec:FI}
For each $j = 1,2$ and $k = 1,2,3$, the FI estimators of $\theta_k$ and $\beta_{jk}$ are obtained as
\begin{eqnarray}
\hat{\theta}_{k,\mathrm{FI}} &=& \widehat{\Pro}(D_k = 1) = \frac{1}{n}\sum_{i=1}^{n}\hat{\rho}_{ki}, \label{est:fi1} \\
\hat{\beta}_{jk,\mathrm{FI}} &=& \widehat{\Pro}(T \ge c_j, D_k = 1) = \frac{1}{n}\sum_{i=1}^{n}\mathrm{I}(T_i \ge c_j)\hat{\rho}_{ki},\label{est:fi2}
\end{eqnarray}
where $\hat{\rho}_{ki}$ is an estimate of $\rho_{ki} = \Pro(D_{ki} = 1|T_i,A_i)$ given by some suitable model, such as the multinomial logistic or probit regression model, applied to
the verified sample units. Therefore, the FI estimator $\widehat{\TCF}_{1,\mathrm{FI}}(c_1)$, $\widehat{\TCF}_{2,\mathrm{FI}}(c_1,c_2)$ and $\widehat{\TCF}_{3,\mathrm{FI}}(c_2)$ are
\begin{align}
%\begin{eqnarray}
\widehat{\TCF}_{1,\mathrm{FI}}(c_1) &= 1 - \frac{\sum\limits_{i=1}^{n}\mathrm{I}(T_i \ge c_1)\hat{\rho}_{1i}}{\sum\limits_{i=1}^{n}\hat{\rho}_{1i}}  \nonumber, \displaybreak[1] \\
\widehat{\TCF}_{2,\mathrm{FI}}(c_1,c_2) &= \frac{\sum\limits_{i=1}^{n}\mathrm{I}(c_1 \le T_i < c_2)\hat{\rho}_{2i}}{\sum\limits_{i=1}^{n}\hat{\rho}_{2i}}, \nonumber \label{est:fi}  \\
\widehat{\TCF}_{3,\mathrm{FI}}(c_2) &= \frac{\sum\limits_{i=1}^{n}\mathrm{I}(T_i \ge c_2)\hat{\rho}_{3i}}{\sum\limits_{i=1}^{n}\hat{\rho}_{3i}}. \nonumber
%\end{eqnarray}
\end{align}
It is worth noting that  estimates $\hat{\theta}_{k,\mathrm{FI}}$ and $\hat{\beta}_{jk,\mathrm{FI}}$ in (\ref{est:fi1}) and (\ref{est:fi2}) are the solutions of the estimating equations
\begin{equation}
\sum_{i=1}^{n}(\hat{\rho}_{ki} - \theta_k) = 0 \quad \text{ and } \quad \sum_{i=1}^{n}\left\{\mathrm{I}(T_i \ge c_j)\hat{\rho}_{ki} - \beta_{jk}\right\} = 0 \label{est_eq:fi2}.
\end{equation}
\subsection{Mean score imputation}\label{sec:MSI}
By inspection of (\ref{est:msi:curve}), we get the MSI estimators of $\theta_k$, $k = 1,2,3$, as follows 
\[
\hat{\theta}_{k,\mathrm{MSI}} = \widehat{\Pro}(D_k = 1) = \frac{1}{n}\sum_{i=1}^{n}\left[V_i D_{ki} + (1 - V_i) \hat{\rho}_{ki}\right]. 
\]
The estimators of $\beta_{jk}$ are given by
\[
\hat{\beta}_{jk,\mathrm{MSI}} = \widehat{\Pro}(T \ge c_j, D_k = 1) = \frac{1}{n}\sum_{i=1}^{n}\mathrm{I}(T_i \ge c_j)\left[V_i D_{ki} + (1 - V_i) \hat{\rho}_{ki}\right].
\]
Then, the %estimators of $\theta_k$ and $\beta_{jk}$ lead to the following 
MSI estimators of $\TCF_1(c_1)$, $\TCF_2(c_1, c_2)$ and $\TCF_3(c_3)$ are :
\begin{eqnarray}
\widehat{\TCF}_{1,\mathrm{MSI}}(c_1) &=& 1 - \frac{\sum\limits_{i=1}^{n}\mathrm{I}(T_i \ge c_1)\left[V_i D_{1i} + (1 - V_i) \hat{\rho}_{1i}\right]}{\sum\limits_{i=1}^{n}\left[V_i D_{1i} + (1 - V_i) \hat{\rho}_{1i}\right]}, \nonumber \\
\widehat{\TCF}_{2,\mathrm{MSI}}(c_1,c_2) &=&  \frac{\sum\limits_{i=1}^{n}\mathrm{I}(c_1 \le T_i < c_2)\left[V_i D_{2i} + (1 - V_i) \hat{\rho}_{2i}\right]}{\sum\limits_{i=1}^{n}\left[V_i D_{2i} + (1 - V_i) \hat{\rho}_{2i}\right]}, \nonumber \\
\widehat{\TCF}_{3,\mathrm{MSI}}(c_2) &=&  \frac{\sum\limits_{i=1}^{n}\mathrm{I}(T_i \ge c_2)\left[V_i D_{3i} + (1 - V_i) \hat{\rho}_{3i}\right]}{\sum\limits_{i=1}^{n}\left[V_i D_{3i} + (1 - V_i) \hat{\rho}_{3i}\right]}. \nonumber
\end{eqnarray}
Again, we can obtain $\hat{\theta}_k$ and $\hat{\beta}_{jk}$ as solution of the estimating equations 
\begin{align}
\sum_{i=1}^{n}\left\{V_i (D_{ki} - \theta_k) + (1 - V_i)(\hat{\rho}_{ki} - \theta_k)\right\} &= 0, \label{est_eq:msi1} \\
\sum_{i=1}^{n}\left\{V_i (\mathrm{I}(T_i \ge c_j)D_{ki} - \beta_{jk}) + (1 - V_i)(\mathrm{I}(T_i \ge c_j)\hat{\rho}_{ki} - \beta_{jk})\right\} &= 0. \label{est_eq:msi2}
\end{align}
%Note that, the MAR assumption implies that data from the verification sample only can be used to obtain valid estimates of $\rho_{ki}$.
%%
\subsection{Inverse probability weighted}\label{sec:IPW}
From the IPW estimators of $\beta_1$ and $\theta$ in (\ref{est:ipw:curve}), we derive, by analogy,
\begin{eqnarray}
\hat{\theta}_{k,\mathrm{IPW}} &=& \widehat{\Pro}(D_{k} = 1) = \frac{\sum\limits_{i=1}^{n}V_i \hat{\pi}_i^{-1}D_{ki}}{\sum\limits_{i=1}^{n}V_i \hat{\pi}_i^{-1}},
\label{est:ipw2} \nonumber \\
\hat{\beta}_{jk,\mathrm{IPW}} &=& \widehat{\Pro}(T \ge c_j, D_k = 1) = \frac{\sum\limits_{i=1}^{n}\mathrm{I}(T_i \ge c_j) V_i \hat{\pi}_i^{-1}D_{ki}}{\sum\limits_{i=1}^{n}V_i \hat{\pi}_i^{-1}}. \nonumber
\end{eqnarray}
The estimates $\hat{\pi}_i$ are obtained in the same way as in the two-class case. 
%Note that, the IPW estimators only use verified subjects. 
Then, the IPW estimates $\widehat{\TCF}_{1,\mathrm{IPW}}(c_1), \widehat{\TCF}_{2,\mathrm{IPW}}(c_1,c_2)$ and $\widehat{\TCF}_{3,\mathrm{IPW}}(c_2)$ are
\begin{eqnarray}
\widehat{\TCF}_{1,\mathrm{IPW}}(c_1) &=&  1 - \frac{\sum\limits_{i=1}^{n}\mathrm{I}(T_i \ge c_1) V_i \hat{\pi}_i^{-1}D_{1i}}{\sum\limits_{i=1}^{n}V_i \hat{\pi}_i^{-1}D_{1i}} \nonumber, \\
\widehat{\TCF}_{2,\mathrm{IPW}}(c_1,c_2) &=&  \frac{\sum\limits_{i=1}^{n}\mathrm{I}(c_1 \le T_i < c_2)V_i \hat{\pi}_i^{-1}D_{2i}}{\sum\limits_{i=1}^{n}V_i \hat{\pi}_i^{-1}D_{2i}}, \nonumber \\
\widehat{\TCF}_{3,\mathrm{IPW}}(c_2) &=& \frac{\sum\limits_{i=1}^{n}\mathrm{I}(T_i \ge c_2) V_i \hat{\pi}_i^{-1}D_{3i}}{\sum\limits_{i=1}^{n}V_i \hat{\pi}_i^{-1}D_{3i}}, \nonumber
\end{eqnarray}
and the estimating equations  corresponding to $\hat{\theta}_{k,\mathrm{IPW}}$ and $\hat{\beta}_{jk,\mathrm{IPW}}$ are
\begin{eqnarray}
\sum_{i=1}^{n}V_i \hat{\pi}_i^{-1}\left(D_{ki} - \theta_k\right) &=& 0,
\label{est_eq:ipw1} \\
\sum_{i=1}^{n}V_i\hat{\pi}_i^{-1}\left(\mathrm{I}(T_i \ge c_j) D_{ki} - \beta_{jk}\right) &=& 0 \label{est_eq:ipw2}.
\end{eqnarray}
Note that the IPW estimators only use verified subjects.
\subsection{Semiparametric efficient}\label{sec:SPE}
Similarly to three previous cases, the SPE estimators of $\beta_{jk}$ and $\theta_k$ are derived in analogy to  $\hat{\beta}_{1,\mathrm{SPE}}$ and $\hat{\theta}_{\mathrm{SPE}}$ in (\ref{est:spe:curve}), i.e,
\begin{eqnarray}
\hat{\theta}_{k,\mathrm{SPE}} &=& \frac{1}{n}\sum_{i=1}^{n}\left\{ V_i D_{ki}\hat{\pi}_i^{-1} - \hat{\rho}_{ki}(V_i \hat{\pi}_i^{-1} - 1) \right\}, \label{est:spe1}\\
\hat{\beta}_{jk,\mathrm{SPE}} &=& \frac{1}{n}\sum_{i=1}^{n} \mathrm{I}(T_i \ge c_j) \left\{V_i D_{ki}\hat{\pi}_i^{-1} - \hat{\rho}_{ki}(V_i \hat{\pi}_i^{-1} - 1) \right\}. \label{est:spe2}
\end{eqnarray}
Therefore, we obtain
\begin{align}
%\begin{eqnarray}
\widehat{\TCF}_{1,\mathrm{SPE}}(c_1) &= 1 - \frac{\sum\limits_{i=1}^{n} \mathrm{I}(T_i \ge c_1) \left\{ V_i D_{1i}\hat{\pi}_i^{-1} - \hat{\rho}_{1i}(V_i \hat{\pi}_i^{-1} - 1) \right\}} {\sum\limits_{i=1}^{n} \left\{V_i D_{1i}\hat{\pi}_i^{-1} - \hat{\rho}_{1i}(V_i \hat{\pi}_i^{-1} - 1) \right\}} \nonumber, \\
\widehat{\TCF}_{2,\mathrm{SPE}}(c_1,c_2) &=  \frac{\sum\limits_{i=1}^{n} \mathrm{I}(c_1 \le T_i < c_2) \left\{ V_i D_{2i}\hat{\pi}_i^{-1} - \hat{\rho}_{2i}(V_i \hat{\pi}_i^{-1} - 1) \right\}}{\sum\limits_{i=1}^{n} \left\{V_i D_{2i}\hat{\pi}_i^{-1} - \hat{\rho}_{2i}(V_i \hat{\pi}_i^{-1} - 1) \right\}}, \nonumber \displaybreak[1] \\
\widehat{\TCF}_{3,\mathrm{SPE}}(c_2) &= \frac{\sum\limits_{i=1}^{n} \mathrm{I}(T_i \ge c_2) \left\{ V_i D_{3i}\hat{\pi}_i^{-1} - \hat{\rho}_{3i}(V_i \hat{\pi}_i^{-1} - 1) \right\}} {\sum\limits_{i=1}^{n} \left\{ V_i D_{3i}\hat{\pi}_i^{-1} - \hat{\rho}_{3i}(V_i \hat{\pi}_i^{-1} - 1) \right\}}. \nonumber
%\end{eqnarray}
\end{align}
The estimates $\hat{\theta}_{k,\mathrm{SPE}}$ and $\hat{\beta}_{jk,\mathrm{SPE}}$ solve the  estimating equations
\begin{align}
\sum_{i=1}^{n}\left\{\frac{V_i}{\hat\pi_i}\left[\mathrm{I}(T_i \ge c_j)D_{ki} - \beta_{jk}\right] - \frac{V_i - \hat\pi_i}{\hat\pi_i}\left[\mathrm{I}(T_i \ge c_j)\hat\rho_{ki} - \beta_{jk}\right]\right\} &= 0, \label{est_eq:spe1} \\
\sum_{i=1}^{n}\left\{\frac{V_i}{\hat\pi_i}\left(D_{ki} - \theta_k\right) - \frac{V_i - \hat\pi_i}{\hat\pi_i}\left(\hat\rho_{ki} - \theta_k\right)\right\} &= 0. \label{est_eq:spe2}
\end{align}

\subsection{Asymptotic distribution theory}
The parameters of interest $\TCF_1(c_1), \TCF_2(c_1,c_2)$ and $\TCF_3(c_2)$ are functions of $\theta_1,$ $\theta_2,$ $\beta_{11},$ $\beta_{12},$ $\beta_{22},$ $\beta_{23}$ and  $\tau = (\tau_\rho,\tau_\pi)$, where $\tau$ is the vector of  parameters of the models used to estimate $\rho = (\rho_1,\rho_2)$, or $\pi$, or both. Let us denote  $\alpha = (\theta_1,\theta_2,\beta_{11}, \beta_{12},\beta_{22},\beta_{23},\tau)$. The estimators (FI, MSI, IPW, SPE) of $\alpha$ are obtained by solving suitable estimating equations. Hence, we use results in \cite{alo:03} and \cite{alo:05} to give consistency and asymptotic normality of the
proposed bias--corrected estimators. 

According to equations (\ref{est_eq:fi2}), (\ref{est_eq:msi1}), (\ref{est_eq:msi2}), (\ref{est_eq:ipw1}), (\ref{est_eq:ipw2}), (\ref{est_eq:spe1}) and (\ref{est_eq:spe2}), let $G_{*}^{\theta_s}(\alpha) =  \sum\limits_{i=1}^{n}g_{i,*}^{\theta_s} (\alpha)$ and $G_{*}^{\beta_{jk}}(\alpha) =  \sum\limits_{i=1}^{n}g_{i,*}^{\beta_{jk}}(\alpha)$ be the estimating functions for $\theta_s$ and $\beta_{jk},$ with $k = 1,2,3,$  $s$ and $j = 1,2,$ for one  of the four previously introduced approaches (the star indicates FI, MSI, IPW, SPE). We assume that $\hat{\tau}$ is the solution to a classic set of estimating equations of the form $G^{\tau}(\alpha) =  \sum\limits_{i=1}^{n}g_i^{\tau}(\alpha) = 0$. For example, such estimating equations could be those derived from a  multinomial logistic regression model for estimation of the disease process and from a logistic regression model for estimation of the verification process. The estimate $\hat{\alpha}_*$ of $\alpha$ is then obtained by solving $G_{*}(\alpha) = \sum\limits_{i=1}^{n}g_{i,*}(\alpha) = 0$, where $g_{i,*}(\alpha) = \left(g_{i,*}^{\theta_1}(\alpha),g_{i,*}^{\theta_2}(\alpha),g_{i,*}^{\beta_{11}}(\alpha),g_{i,*}^{\beta_{12}}(\alpha),g_{i,*}^{\beta_{22}}(\alpha),g_{i,*}^{\beta_{23}}(\alpha),g_i^{\tau}(\alpha)\right)^\top$. 
 
Let $\alpha_0 = (\theta_{10},\theta_{10},\beta_{110}, \beta_{120},\beta_{220},\beta_{230},\tau_0)$ be  the true value of $\alpha.$
We assume that
\begin{enumerate}[({A}1)]
\item $D$ is missing at random (MAR);
\item the data $(D_i,T_i,A_i,V_i)$ are iid;
\item $(T,A)$ is a bounded random vector;
\item $\E \left[\frac{\partial}{\partial \alpha} g_{i,*}(\alpha_0)\right]$ is negative definite;
\item[(A5)] $\rho_{ki}$ and $\pi_i$ are bounded away from $0$.
\end{enumerate}
We consider also the following standard regularity conditions.
\begin{enumerate}[({C}1)]
\item $g_{i,*}(\alpha_0)$ are iid and $\E \left\{g_{i,*}(\alpha_0)\right\} = 0$.
\item Elements of $G_{*}(\alpha)$, $\frac{\partial}{\partial \alpha}G_{*}(\alpha)$, and $\frac{\partial^2}{\partial \alpha \partial \alpha^\top}G_{*}(\alpha)$ exist in a bounded $\delta$-neighborhood of $\alpha_0$, $N_{\delta}(\alpha_0)$.
\item $g_{i,*}(\alpha)$, $\frac{\partial}{\partial \alpha}g_{i,*}(\alpha)$, and $\frac{\partial^2}{\partial \alpha \partial \alpha^\top}g_{i,*}(\alpha)$ are uniformly bounded in $N_{\delta}(\alpha_0)$.
\end{enumerate}
Under the assumptions (A1)--(A5) and conditions (C1)--(C3), we obtain the asymptotic results summarized in the following theorem.

\begin{thm}\label{theo:cons}
Let $\TCF_{10}(c_1), \TCF_{20}(c_1,c_2), \TCF_{30}(c_2)$ be the true parameter values.
The FI, MSI, IPW or SPE bias-corrected estimators   $\widehat{\TCF}_{1,*}(c_1)$, $\widehat{\TCF}_{2,*}(c_1,c_2)$ and $\widehat{\TCF}_{3,*}(c_2)$  are consistent. Furthermore, 
\begin{equation}
\sqrt{n} \left[\left(\begin{array}{c}
\widehat{\TCF}_{1,*}(c_1) \\ \widehat{\TCF}_{2,*}(c_1,c_2) \\ \widehat{\TCF}_{3,*}(c_2)
\end{array}\right) - 
\left(\begin{array}{c}
\TCF_{10}(c_1) \\ \TCF_{20}(c_1,c_2) \\ \TCF_{30}(c_2)
\end{array}\right) \right]
\stackrel{d}{\to} {\cal N}_3 \left(\boldsymbol{0}, \Xi \right),
\label{asym:2}
\end{equation}
where 
\[
\Xi = \frac{\partial h(\alpha_0)}{\partial \alpha} {\Sigma} \frac{\partial h^\top(\alpha_0)}{\partial \alpha},
\]
with 
$h(\alpha) = \left(1 - \frac{\beta_{11}}{\theta_1}, \frac{\beta_{12} - \beta_{22}}{\theta_2}, \frac{\beta_{23}}{1 - (\theta_1 + \theta_2)}\right)^\top$ and
\[
\Sigma = \left[\E \left\{\frac{\partial}{\partial \alpha}g_{i,*}(\alpha_0)\right\}\right]^{-1} \Cov\{g_{i,*}(\alpha_0)\}\left[\E \left\{\frac{\partial}{\partial \alpha}g_{i,*}^\top(\alpha_0)\right\}\right]^{-1}.
\]
\end{thm}
\begin{proof}
We apply Theorem 1 and Theorem 2 of \cite{alo:03}. Under assumptions (A1)--(A5) and conditions (C1)--(C3), $\hat{\alpha}_*$ 
is consistent and $\sqrt{n}\left(\hat{\alpha}_* - \alpha_0\right) \stackrel{d}{\to} \mathcal{N}(\boldsymbol{0}, \Xi)$. Thus, the estimators $\widehat{\TCF}_{1,*}(c_1) = 1 - \hat{\beta}_{11}/\hat{\theta}_1, \widehat{\TCF}_{2,*}(c_1,c_2)$ $= (\hat{\beta}_{12} - \hat{\beta}_{22})/\hat{\theta}_2$ and $\widehat{\TCF}_{3,*}(c_2) = \hat{\beta}_{23}/(1 - (\hat{\theta}_1 + \hat{\theta}_2))$ are consistent for the true $\TCF_{10}(c_1)$, $ \TCF_{20}(c_1,c_2)$ and $\TCF_{30}(c_2)$ and, by and application of the multivariate delta method, result (\ref{asym:2}) follows. In Appendix \ref{s:asymptotic}, we check conditions (C1)--(C3) for each estimator, i.e., FI, MSI, IPW and SPE, under assumptions (A1)--(A5). This is done when a multinomial logistic regression model is used for the estimation of the disease process and a logistic regression model or a probit model is used for the estimation of the verification process.
\end{proof}

The above theorem gives a general result for all estimates, i.e., FI, MSI, IPW and SPE. In Appendix, the explicit form of the  asymptotic variance--covariance matrix is obtained.
In practice, the variance--covariance matrix $\Sigma$ is replaced by a consistent estimate $\hat{\Sigma}$  
\[
\hat{\Sigma} = n \left[\sum_{i=1}^{n}\frac{\partial}{\partial \alpha}g_{i,*}(\hat{\alpha})\right]^{-1} \left[\sum_{i=1}^{n}g_{i,*}(\hat{\alpha})g_{i,*}(\hat{\alpha})^\top\right]\left[\sum_{i=1}^{n}\frac{\partial}{\partial \alpha}g_{i,*}^\top(\hat{\alpha})\right]^{-1}.
\]

It is worth noting that SPE estimators of $\theta_k$ and $\beta_{jk}$ in (\ref{est:spe1}) and (\ref{est:spe2}), will inherit the double robustness property of $\hat{\theta}_{\mathrm{SPE}}$ and $\hat{\beta}_{1,\mathrm{SPE}}$ in (\ref{est:spe:curve}). That is, $\hat{\theta}_{k,\mathrm{SPE}}$ and $\hat{\beta}_{jk,\mathrm{SPE}}$ remain consistent if only one of  the disease model $P(D_k = 1|T,A)$ or the verification model $P(V = 1|T,A)$ is correctly specified in the estimation process; they are inconsistent if both models are misspecified. Clearly, this property holds also for the estimators $\widehat{\TCF}_{1,\mathrm{SPE}}(c_1), \widehat{\TCF}_{2,\mathrm{SPE}}(c_1,c_2)$ and $\widehat{\TCF}_{3,\mathrm{SPE}}(c_2)$.

\subsection{VUS estimation}\label{s:vus}

Let $\mu$ be the volume under the ROC surface (VUS) of $T$. A straightforward calculation (\citet{nak:04}) shows that 
\begin{eqnarray}
\mu &=& \Pro \left(T_i < T_\ell < T_r|D_{1i} = 1,D_{2\ell} = 1,D_{3r} = 1\right) \nonumber \\
&& + \: \frac{1}{2} \Pro\left(T_i < T_\ell = T_r|D_{1i} = 1,D_{2\ell} = 1,D_{3r} = 1\right) \nonumber\\
&& + \: \frac{1}{2} \Pro\left(T_i = T_\ell < T_r|D_{1i} = 1,D_{2\ell} = 1,D_{3r} = 1\right) \nonumber\\
&& + \: \frac{1}{6} \Pro\left(T_i = T_\ell = T_r|D_{1i} = 1,D_{2\ell} = 1,D_{3r} = 1\right) \nonumber
\end{eqnarray}
or, equivalently,
\begin{eqnarray}
\mu &=& \frac{\E \left(D_{1i}D_{2\ell}D_{3r}\I_{i\ell r}\right)}{\E \left(D_{1i}D_{2\ell}D_{3r}\right)},  \nonumber
\label{org:vus}
\end{eqnarray}
where $\I_{i\ell r} = \I(T_i < T_\ell < T_r) + \frac{1}{2}\I(T_i < T_\ell = T_r) + \frac{1}{2}\I(T_i = T_\ell < T_r) + \frac{1}{6}\I(T_i = T_\ell = T_r)$. 
Then, in the absence of missing data, a natural nonparametric estimator $\hat{\mu}$ of $\mu$ is given by
\begin{equation}
\hat{\mu} = \frac{\sum\limits_{i=1}^{n}\sum\limits_{\ell = 1, \ell \ne i}^{n} \sum\limits_{\stackrel{r=1}{r \ne \ell, r\ne i}}^{n}\I_{i\ell r}D_{1i}D_{2\ell}D_{3r}}{\sum\limits_{i=1}^{n}\sum\limits_{\ell = 1, \ell \ne i}^{n} \sum\limits_{\stackrel{r=1}{r \ne \ell, r\ne i}}^{n} D_{1i}D_{2\ell}D_{3r}}.
\label{nonp:vus}
\end{equation}

When the disease status is missing for some of the subjects, verification bias--corrected estimators of VUS,  can be obtained by using suitable estimates of quantities  $D_{1i}, D_{2i}$ and $D_{3i}$ in (\ref{nonp:vus}).
More precisely,  FI, MSI, IPW and SPE estimators of VUS take the form
\begin{eqnarray}
\hat{\mu}_{\mathrm{*}} &= &\frac{\sum\limits_{i=1}^{n}\sum\limits_{\ell = 1, \ell \ne i}^{n} \sum\limits_{\stackrel{r=1}{r \ne \ell, r\ne i}}^{n}\I_{i\ell r}\tilde{D}_{1i,\mathrm{*}}\tilde{D}_{2\ell,\mathrm{*}}\tilde{D}_{3r,\mathrm{*}}}{\sum\limits_{i=1}^{n}\sum\limits_{\ell = 1, \ell \ne i}^{n} \sum\limits_{\stackrel{r=1}{r \ne \ell, r\ne i}}^{n} \tilde{D}_{1i,\mathrm{*}}\tilde{D}_{2\ell,\mathrm{*}}\tilde{D}_{3r,\mathrm{*}}},\nonumber
\label{msi:vus}
\end{eqnarray}
where the star again stands for FI, MSI, IPW, SPE, and
\[
\tilde{D}_{ki,\mathrm{FI}}=\hat{\rho}_{ki}, \ \ \
\tilde{D}_{ki,\mathrm{MSI}}=V_i D_{ki} + (1-V_i)\hat{\rho}_{ki}, \ \ \
\tilde{D}_{ki,\mathrm{IPW}}=V_i D_{ki} \hat{\pi}_{i}^{-1},
\]
\[
\tilde{D}_{ki,\mathrm{SPE}} = V_i D_{ki} \hat{\pi}_i^{-1} - \hat{\rho}_{ki}(V_i\hat{\pi}_i^{-1} - 1)  \qquad   (k = 1,2,3).
\]

As for the estimators of the TCFs, under the MAR assumption and certain suitable regularity conditions, we can establish consistency and asymptotic normality of the above given bias-corrected VUS estimators (proof available from the authors). 
Moreover, the asymptotic variance of $\hat{\mu}_*$, i.e. the variance of $\sqrt{n}\left(\hat{\mu}_{*} - \mu_0 \right),$ 
can be consistently estimated by
\[
\frac{\frac{1}{n - 1} \sum\limits_{i=1}^{n}\hat{Q}_i^2(\hat{\mu}_{*},\hat{\tau})}{\hat{\theta}_{1,*}^2 \hat{\theta}_{2,*}^2 \hat{\theta}_{3,*}^2},
\]
where 
\begin{align}
\hat{Q}_i (\hat{\mu}_{*},\hat{\tau}) = & - \left\{\frac{1}{(n-1)(n-2)}\sum_{i=1}^{n}\sum_{\stackrel{\ell=i}{\ell \ne i}}^{n}\sum_{\stackrel{r = 1}{r \ne \ell, r\ne i}}^{n} \frac{\partial G_{i\ell r,*}(\hat{\mu}_{*},\tau_\rho,\hat{\tau}_\pi)}{\partial \tau_\rho^\top}\bigg|_{\tau_\rho = \hat{\tau}_\rho}\right\} \nonumber \\
&  \times \: \left\{\sum_{i=1}^{n}\frac{\partial g_i^{\tau_\rho}}{\partial \tau_\rho}\bigg|_{\tau_\rho = \hat{\tau}_\rho}\right\}^{-1} g_i^{{\tau}_\rho}\bigg|_{\tau_\rho = \hat{\tau}_\rho} \nonumber \\
&  - \: \left\{\frac{1}{(n-1)(n-2)}\sum_{i=1}^{n}\sum_{\stackrel{\ell=i}{\ell \ne i}}^{n}\sum_{\stackrel{r = 1}{r \ne \ell, r\ne i}}^{n} \frac{\partial G_{i\ell r,*}(\hat{\mu}_{*},\hat{\tau}_\rho,\tau_\pi)}{\partial \tau_\pi^\top}\bigg|_{\tau_\pi = \hat{\tau}_\pi}\right\} \nonumber \\
& \times \: \left\{\sum_{i=1}^{n}\frac{\partial g_i^{\tau_\pi}}{\partial \tau_\pi}\bigg|_{\tau_\pi = \hat{\tau}_\pi}\right\}^{-1} g_i^{{\tau}_\pi}\bigg|_{\tau_\pi = \hat{\tau}_\pi} \nonumber\\
& + \: \frac{1}{(n-1)(n-2)} \sum_{\stackrel{\ell=1}{\ell \ne i}}^{n} \sum_{\stackrel{r=1}{r \ne i, r \ne \ell}}^{n}\bigg\{ G_{i\ell r,*}(\hat{\mu}_{*},\hat{\tau}_\rho,\hat{\tau}_\pi) + G_{\ell ir,*}(\hat{\mu}_{*},\hat{\tau}_\rho,\hat{\tau}_\pi) \nonumber \\
& + \: G_{r\ell i,*}(\hat{\mu}_{*},\hat{\tau}_\rho,\hat{\tau}_\pi)\bigg\}, \label{asy:var:Qi}
\end{align}
with
\begin{align}
G_{i\ell r,\mathrm{FI}}(\mu,\tau_\rho,\tau_\pi) &= \rho_{1i}(\tau_\rho)\rho_{2\ell}(\tau_\rho)\rho_{3r}(\tau_\rho)\left(I_{i\ell r} - \mu\right), \nonumber \\
G_{i\ell r,\mathrm{MSI}}(\mu,\tau_\rho,\tau_\pi) &= D_{1i,\mathrm{MSI}}(\tau_\rho)D_{2\ell,\mathrm{MSI}}(\tau_\rho)D_{3r,\mathrm{MSI}}(\tau_\rho)\left(I_{i\ell r} - \mu\right), \nonumber \\
G_{i\ell r,\mathrm{IPW}}(\mu,\tau_\rho,\tau_\pi) &= \frac{V_iV_\ell V_r}{\pi_i(\tau_\pi)\pi_{\ell}(\tau_\pi)\pi_r(\tau_\pi)} D_{1i}D_{2\ell}D_{3r} \left(I_{i\ell r} - \mu\right), \nonumber \\
G_{i\ell r,\mathrm{SPE}}(\mu,\tau_\rho,\tau_\pi) &= D_{1i,\mathrm{SPE}}(\tau_\rho,\tau_\pi)D_{2\ell,\mathrm{SPE}}(\tau_\rho,\tau_\pi)D_{3r,\mathrm{SPE}}(\tau_\rho,\tau_\pi)\left(I_{i\ell r} - \mu\right), \nonumber
\end{align}
and
\begin{eqnarray}
D_{ki,\mathrm{MSI}}(\tau_\rho) &=& V_i D_{ki} + (1-V_i)\rho_{ki}(\tau_\rho), \nonumber \\
D_{ki,\mathrm{SPE}}(\tau_\rho,\tau_\pi) &=& V_i D_{ki} \pi^{-1}_i(\tau_\pi) - \rho_{ki}(\tau_\rho)(V_i\pi^{-1}_i(\tau_\pi) - 1), \nonumber
\end{eqnarray}
for $k = 1,2,3$. In (\ref{asy:var:Qi}), the functions $g^{\tau_\rho}_i(\cdot)$ and $g_i^{\tau_\pi}(\cdot)$ are the elements of the functions $g_i^\tau(\cdot)$ in the estimating function $G^\tau(\cdot)$ for the parameters of the models adopted for the disease and the verification processes. See the Appendix \ref{s:asymptotic} for their specification when the models chosen are, the multinomial logistic regression and the logistic (or probit) model, respectively.

\section{Simulation studies}
\label{s:simulation}
In this section, the ability of FI, MSI, IPW and SPE methods to estimate TCF$_1$, TCF$_2$ and TCF$_3$ are evaluated by using Monte Carlo experiments. Also, the square root of the estimates of the variances are compared with Monte Carlo and bootstrap standard deviations. Some simulation results concerning the behaviour of the estimators for the VUS are given in Appendix \ref{app:simu:vus}.

Note that, the bias-corrected estimators of TCF$_1$, TCF$_2$ and TCF$_3$ require a parametric regression model to estimate $\rho_{ki} = \Pro(D_{ki} = 1|T_i, A_i)$, or $\pi_i = \Pro(V_i = 1 | T_i,A_i)$, or both. A wrong specification of such models may affect the estimation. Therefore, in the simulation study we consider four scenarios:
\begin{enumerate}[(i)]
\item the disease model and the verification model are both correctly specified;
\item the verification model is misspecified;
\item the disease model is misspecified;
\item the disease model and the verification model are both misspecified.
\end{enumerate}
All scenarios allow to evaluate the behavior of the proposed estimators in finite samples. In practice, we consider $5000$ Monte Carlo replications, and three sample sizes, i.e., $250$, $500$ and $1000$ in scenario (i) and a sample size equal to $1000$ in scenarios (ii)--(iv). The choice of such sample size in scenarios (ii)--(iv) allows to dig up expected bad behaviors of the estimators under misspecification,  when a great amount of information is available, i.e., in large samples.

\subsection{Study 1}
The true disease $D$ is generated by a trinomial random vector $(D_{1},D_{2},D_{3})$, such that $D_{k}$ is a Bernoulli random variable with mean $\theta_k$, $k=1,2,3$. We set $\theta_1 = 0.4, \theta_2 = 0.35$ and $\theta_3 = 0.25$. The continuous test results $T$ and $A$ are generated from the following conditional models
\[
T,A |D_{k} \sim \mathcal{N}_2 \left(\mu_k, \Lambda\right), \qquad k = 1,2,3,
\]
where $\mu_k = (2 k, k)^\top.$ We consider three different values for $\Lambda$, specifically 
\[
\left(\begin{array}{c c}
1.75 & 0.1 \\
0.1 & 2.5
\end{array}\right) , \qquad 
\left(\begin{array}{c c}
2.5 & 1.5 \\
1.5 & 2.5
\end{array}\right) , \qquad 
\left(\begin{array}{c c}
5.5 & 3 \\
3 & 2.5
\end{array}\right),
\]
giving rise to a correlation between $T$ and $A$ equal to $0.36, 0.69$ and $0.84$, respectively.

In this scenario -and also in the next one- we consider six pairs for cut points $(c_1,c_2)$, i.e., $(2,4), (2,5)$, $(2,7), (4,5), (4,7)$ and $(5,7)$. Since the conditional distribution of $T$ given $D_{k}$ is the normal distribution, the true values of TCF's are obtained as
\begin{eqnarray}
{\TCF}_{1}(c_1) &=& \Phi \left(\frac{c_1 - 2}{\sigma_{T|D}}\right) \nonumber, \\
{\TCF}_{2}(c_1,c_2) &=& \Phi \left(\frac{c_2 - 4}{\sigma_{T|D}}\right) - \Phi \left(\frac{c_1 - 4}{\sigma_{T|D}}\right), \nonumber\\
{\TCF}_{3}(c_2) &=& 1 - \Phi \left(\frac{c_2 - 6}{\sigma_{T|D}}\right), \nonumber
\end{eqnarray}
where $\sigma_{T|D}$ denotes the entry in the 1-st row and 1-st column of $\Lambda$ and  $\phi(\cdot)$ and $\Phi(\cdot)$ are the density function and the cumulative distribution function of the standard normal random variable, respectively.

Under our data--generating process, the true conditional disease model is a multinomial logistic model
\[
\mathrm{Pr}(D_{k} = 1|T,A) = \frac{\exp \left(\tau_{\rho_{1k}} + \tau_{\rho_{2k}} T + \tau_{\rho_{3k}} A\right)}{1 + \exp \left(\tau_{\rho_{11}} + \tau_{\rho_{21}} T + \tau_{\rho_{31}} A\right) + \exp \left(\tau_{\rho_{12}} + \tau_{\rho_{22}} T + \tau_{\rho_{32}} A\right)},
\]
for suitable $\tau_{\rho_{1k}},\tau_{\rho_{2k}},\tau_{\rho_{3k}}$, where $k = 1,2$. The verification status $V$ is generated by the following model
\[
\mathrm{logit}\left\{\Pro(V = 1|T,A)\right\} = 0.5  - 0.3 T + 0.75 A.
\]
This choice corresponds to a verification rate of about $0.65$. In this study, the FI, MSI, IPW and SPE estimators are computed under correct working models for both the disease and the verification status. Therefore, in particular, the conditional verification probabilities $\pi_i$ are estimated from a logistic model for $V$ given $T$ and $A$. 

Tables \ref{scen1:1:250}--\ref{scen1:3:250} and Tables \ref{scen1:1:500}--\ref{scen1:3:1000} in Appendix \ref{app:simu:roc} show Monte Carlo means, Monte Carlo standard deviations (MC.sd), the square roots of the variance estimated via asymptotic results (asy.sd) and bootstrap standard deviations (boot.sd) of $\widehat{\TCF}_1$, $\widehat{\TCF}_2$ and $\widehat{\TCF}_3$. Here and in the following bootstrap estimates are obtained from 250 bootstrap replications. Overall, the estimators FI, MSI, IPW and SPE  behave similarly in this scenario, with  the IPW estimator showing a slightly bigger standard deviation. Simulation results, in this and in the following scenarios, also show that, excluding the SPE approach, bootstrap estimates of standard deviations are generally more accurate than estimates obtained via asymptotic theory. 

\begin{sidewaystable}
%\setstretch{1}
\begin{center}
\caption{Simulation results from 5000 replications when both models for $\rho_{k}$ and $\pi$ are correctly specified (Study 1) and the first value of $\Lambda$ is considered. ``True'' denotes the true parameter value. Sample size = 250.}
\label{scen1:1:250}
\begin{scriptsize}
\begin{tabular}{r c c c c c c c c c c c c}
  \hline
 & TCF$_1$ & TCF$_2$ & TCF$_3$ & MC.sd$_1$ & MC.sd$_2$ & MC.sd$_3$ & asy.sd$_1$ & asy.sd$_2$ & asy.sd$_3$ & boot.sd$_1$ & boot.sd$_2$ & boot.sd$_3$\\
 \hline 
 & \multicolumn{12}{c}{cut-point = (2,4)} \\
True  & 0.5000 & 0.4347 & 0.9347 &  &  &  &  &  &  &  &  &  \\ 
  FI & 0.5008 & 0.4357 & 0.9342 & 0.0529 & 0.0472 & 0.0272 & 0.0486 & 0.0440 & 0.0512 & 0.0530 & 0.0488 & 0.0276 \\ 
  MSI & 0.5007 & 0.4353 & 0.9341 & 0.0544 & 0.0536 & 0.0318 & 0.0500 & 0.0501 & 0.0538 & 0.0542 & 0.0542 & 0.0324 \\ 
  IPW & 0.5017 & 0.4352 & 0.9341 & 0.0714 & 0.0721 & 0.0371 & 0.0687 & 0.0697 & 0.0398 & 0.0704 & 0.0711 & 0.0373 \\ 
  SPE & 0.5008 & 0.4352 & 0.9343 & 0.0574 & 0.0648 & 0.0364 & 0.0562 & 0.0632 & 0.0340 & 0.0596 & 0.1497 & 0.0425 \\
  \hline
  & \multicolumn{12}{c}{cut-point = (2,5)} \\
  True  & 0.5000 & 0.7099 & 0.7752 &  &  &  &  &  &  &  &  &  \\ 
  FI & 0.5008 & 0.7122 & 0.7756 & 0.0529 & 0.0464 & 0.0537 & 0.0486 & 0.0454 & 0.0618 & 0.0530 & 0.0467 & 0.0533 \\ 
  MSI & 0.5007 & 0.7112 & 0.7747 & 0.0544 & 0.0511 & 0.0568 & 0.0500 & 0.0503 & 0.0644 & 0.0542 & 0.0514 & 0.0563 \\ 
  IPW & 0.5017 & 0.7123 & 0.7739 & 0.0714 & 0.0683 & 0.0666 & 0.0687 & 0.0658 & 0.0704 & 0.0704 & 0.0677 & 0.0655 \\ 
  SPE & 0.5008 & 0.7116 & 0.7751 & 0.0574 & 0.0619 & 0.0630 & 0.0562 & 0.0597 & 0.0603 & 0.0596 & 0.1219 & 0.1033 \\
   \hline
  & \multicolumn{12}{c}{cut-point = (2,7)} \\
  True  & 0.5000 & 0.9230 & 0.2248 &  &  &  &  &  &  &  &  &  \\ 
  FI & 0.5008 & 0.9231 & 0.2229 & 0.0529 & 0.0236 & 0.0520 & 0.0486 & 0.0327 & 0.0437 & 0.0530 & 0.0243 & 0.0525 \\ 
  MSI & 0.5007 & 0.9230 & 0.2230 & 0.0544 & 0.0285 & 0.0530 & 0.0500 & 0.0361 & 0.0447 & 0.0542 & 0.0287 & 0.0534 \\ 
  IPW & 0.5017 & 0.9234 & 0.2216 & 0.0714 & 0.0376 & 0.0748 & 0.0687 & 0.0341 & 0.0706 & 0.0704 & 0.0368 & 0.0727 \\ 
  SPE & 0.5008 & 0.9234 & 0.2236 & 0.0574 & 0.0361 & 0.0571 & 0.0562 & 0.0334 & 0.0559 & 0.0596 & 0.0474 & 0.4185 \\
  \hline
  & \multicolumn{12}{c}{cut-point = (4,5)} \\
  True & 0.9347 & 0.2752 & 0.7752 &  &  &  &  &  &  &  &  &  \\ 
  FI & 0.9350 & 0.2765 & 0.7756 & 0.0244 & 0.0408 & 0.0537 & 0.0224 & 0.0350 & 0.0618 & 0.0247 & 0.0415 & 0.0533 \\ 
  MSI & 0.9351 & 0.2759 & 0.7747 & 0.0270 & 0.0467 & 0.0568 & 0.0247 & 0.0411 & 0.0644 & 0.0271 & 0.0467 & 0.0563 \\ 
  IPW & 0.9356 & 0.2770 & 0.7739 & 0.0413 & 0.0690 & 0.0666 & 0.0343 & 0.0645 & 0.0704 & 0.0395 & 0.0663 & 0.0655 \\ 
  SPE & 0.9356 & 0.2764 & 0.7751 & 0.0378 & 0.0587 & 0.0630 & 0.0333 & 0.0560 & 0.0603 & 0.0471 & 0.1566 & 0.1033 \\
  \hline
  & \multicolumn{12}{c}{cut-point = (4,7)} \\
  True & 0.9347 & 0.4883 & 0.2248 &  &  &  &  &  &  &  &  &  \\ 
  FI & 0.9350 & 0.4874 & 0.2229 & 0.0244 & 0.0523 & 0.0520 & 0.0224 & 0.0494 & 0.0437 & 0.0247 & 0.0537 & 0.0525 \\ 
  MSI & 0.9351 & 0.4877 & 0.2230 & 0.0270 & 0.0559 & 0.0530 & 0.0247 & 0.0528 & 0.0447 & 0.0271 & 0.0567 & 0.0534 \\ 
  IPW & 0.9356 & 0.4881 & 0.2216 & 0.0413 & 0.0741 & 0.0748 & 0.0343 & 0.0708 & 0.0706 & 0.0395 & 0.0723 & 0.0727 \\ 
  SPE & 0.9356 & 0.4882 & 0.2236 & 0.0378 & 0.0661 & 0.0571 & 0.0333 & 0.0640 & 0.0559 & 0.0471 & 0.1745 & 0.4185 \\ 
  \hline
  & \multicolumn{12}{c}{cut-point = (5,7)} \\
  True & 0.9883 & 0.2132 & 0.2248 &  &  &  &  &  &  &  &  &  \\ 
  FI & 0.9880 & 0.2109 & 0.2229 & 0.0075 & 0.0432 & 0.0520 & 0.0066 & 0.0387 & 0.0437 & 0.0081 & 0.0436 & 0.0525 \\ 
  MSI & 0.9881 & 0.2118 & 0.2230 & 0.0098 & 0.0460 & 0.0530 & 0.0075 & 0.0423 & 0.0447 & 0.0101 & 0.0468 & 0.0534 \\ 
  IPW & 0.9885 & 0.2111 & 0.2216 & 0.0203 & 0.0634 & 0.0748 & 0.0097 & 0.0601 & 0.0706 & 0.0185 & 0.0625 & 0.0727 \\ 
  SPE & 0.9883 & 0.2117 & 0.2236 & 0.0191 & 0.0569 & 0.0571 & 0.0117 & 0.0542 & 0.0559 & 0.0180 & 0.1382 & 0.4185 \\
   \hline
\end{tabular}
\end{scriptsize}
\end{center}
\end{sidewaystable}

\begin{sidewaystable}
\begin{center}
\caption{Simulation results from 5000 replications when both models for $\rho_{k}$ and $\pi$ are correctly specified (Study 1) and the second value of $\Lambda$ is considered. ``True'' denotes the true parameter value. Sample size = 250.}
\label{scen1:2:250}
\begin{scriptsize}
\begin{tabular}{r c c c c c c c c c c c c}
  \hline
 & TCF$_1$ & TCF$_2$ & TCF$_3$ & MC.sd$_1$ & MC.sd$_2$ & MC.sd$_3$ & asy.sd$_1$ & asy.sd$_2$ & asy.sd$_3$ & boot.sd$_1$ & boot.sd$_2$ & boot.sd$_3$ \\ 
  \hline 
      & \multicolumn{12}{c}{cut-point = (2,4)} \\
True & 0.5000 & 0.3970 & 0.8970 &  &  &  &  &  &  &  &  &  \\ 
  FI & 0.4995 & 0.3966 & 0.8972 & 0.0502 & 0.0419 & 0.0357 & 0.0461 & 0.0375 & 0.0502 & 0.0506 & 0.0429 & 0.0362 \\ 
  MSI & 0.4996 & 0.3966 & 0.8970 & 0.0519 & 0.0498 & 0.0409 & 0.0479 & 0.0463 & 0.0536 & 0.0522 & 0.0506 & 0.0410 \\ 
  IPW & 0.5001 & 0.3972 & 0.8979 & 0.0659 & 0.0700 & 0.0523 & 0.0646 & 0.0677 & 0.0504 & 0.0658 & 0.0687 & 0.0510 \\ 
  SPE & 0.4996 & 0.3968 & 0.8980 & 0.0565 & 0.0623 & 0.0508 & 0.0559 & 0.0617 & 0.0469 & 0.0576 & 0.1619 & 0.0502 \\
  \hline 
      & \multicolumn{12}{c}{cut-point = (2,5)} \\
  True & 0.5000 & 0.6335 & 0.7365 &  &  &  &  &  &  &  &  &  \\ 
  FI & 0.4995 & 0.6340 & 0.7378 & 0.0502 & 0.0431 & 0.0580 & 0.0461 & 0.0410 & 0.0619 & 0.0506 & 0.0440 & 0.0580 \\ 
  MSI & 0.4996 & 0.6335 & 0.7370 & 0.0519 & 0.0502 & 0.0617 & 0.0479 & 0.0485 & 0.0653 & 0.0522 & 0.0510 & 0.0616 \\ 
  IPW & 0.5001 & 0.6330 & 0.7379 & 0.0659 & 0.0679 & 0.0733 & 0.0646 & 0.0660 & 0.0737 & 0.0658 & 0.0671 & 0.0721 \\ 
  SPE & 0.4996 & 0.6335 & 0.7377 & 0.0565 & 0.0616 & 0.0686 & 0.0559 & 0.0610 & 0.0665 & 0.0576 & 0.1438 & 0.0686 \\ 
  \hline 
      & \multicolumn{12}{c}{cut-point = (2,7)} \\
  True & 0.5000 & 0.8682 & 0.2635 &  &  &  &  &  &  &  &  &  \\ 
  FI & 0.4995 & 0.8679 & 0.2640 & 0.0502 & 0.0307 & 0.0559 & 0.0461 & 0.0333 & 0.0499 & 0.0506 & 0.0314 & 0.0558 \\ 
  MSI & 0.4996 & 0.8680 & 0.2644 & 0.0519 & 0.0362 & 0.0588 & 0.0479 & 0.0387 & 0.0523 & 0.0522 & 0.0372 & 0.0580 \\ 
  IPW & 0.5001 & 0.8678 & 0.2659 & 0.0659 & 0.0492 & 0.0695 & 0.0646 & 0.0472 & 0.0682 & 0.0658 & 0.0492 & 0.0690 \\ 
  SPE & 0.4996 & 0.8684 & 0.2649 & 0.0565 & 0.0467 & 0.0615 & 0.0559 & 0.0451 & 0.0591 & 0.0576 & 0.0593 & 0.0610 \\ 
  \hline 
      & \multicolumn{12}{c}{cut-point = (4,5)} \\  
  True & 0.8970 & 0.2365 & 0.7365 &  &  &  &  &  &  &  &  &  \\ 
  FI & 0.8974 & 0.2374 & 0.7378 & 0.0284 & 0.0368 & 0.0580 & 0.0274 & 0.0318 & 0.0619 & 0.0288 & 0.0371 & 0.0580 \\ 
  MSI & 0.8972 & 0.2369 & 0.7370 & 0.0320 & 0.0441 & 0.0617 & 0.0306 & 0.0395 & 0.0653 & 0.0320 & 0.0439 & 0.0616 \\ 
  IPW & 0.8978 & 0.2358 & 0.7379 & 0.0377 & 0.0603 & 0.0733 & 0.0361 & 0.0574 & 0.0737 & 0.0372 & 0.0586 & 0.0721 \\ 
  SPE & 0.8975 & 0.2368 & 0.7377 & 0.0364 & 0.0538 & 0.0686 & 0.0352 & 0.0519 & 0.0665 & 0.0363 & 0.2833 & 0.0686 \\ 
  \hline 
      & \multicolumn{12}{c}{cut-point = (4,7)} \\  
  True & 0.8970 & 0.4711 & 0.2635 &  &  &  &  &  &  &  &  &  \\ 
  FI & 0.8974 & 0.4713 & 0.2640 & 0.0284 & 0.0504 & 0.0559 & 0.0274 & 0.0467 & 0.0499 & 0.0288 & 0.0510 & 0.0558 \\ 
  MSI & 0.8972 & 0.4714 & 0.2644 & 0.0320 & 0.0554 & 0.0588 & 0.0306 & 0.0525 & 0.0523 & 0.0320 & 0.0562 & 0.0580 \\ 
  IPW & 0.8978 & 0.4706 & 0.2659 & 0.0377 & 0.0693 & 0.0695 & 0.0361 & 0.0677 & 0.0682 & 0.0372 & 0.0687 & 0.0690 \\ 
  SPE & 0.8975 & 0.4716 & 0.2649 & 0.0364 & 0.0635 & 0.0615 & 0.0352 & 0.0627 & 0.0591 & 0.0363 & 0.1949 & 0.0610 \\ 
  \hline 
      & \multicolumn{12}{c}{cut-point = (5,7)} \\  
  True & 0.9711 & 0.2347 & 0.2635 &  &  &  &  &  &  &  &  &  \\ 
  FI & 0.9710 & 0.2339 & 0.2640 & 0.0121 & 0.0404 & 0.0559 & 0.0118 & 0.0369 & 0.0499 & 0.0127 & 0.0409 & 0.0558 \\ 
  MSI & 0.9708 & 0.2345 & 0.2644 & 0.0165 & 0.0458 & 0.0588 & 0.0151 & 0.0431 & 0.0523 & 0.0167 & 0.0465 & 0.0580 \\ 
  IPW & 0.9710 & 0.2348 & 0.2659 & 0.0203 & 0.0569 & 0.0695 & 0.0178 & 0.0556 & 0.0682 & 0.0201 & 0.0569 & 0.0690 \\ 
  SPE & 0.9710 & 0.2348 & 0.2649 & 0.0201 & 0.0526 & 0.0615 & 0.0179 & 0.0521 & 0.0591 & 0.0199 & 0.1084 & 0.0610 \\ 
   \hline
\end{tabular}
\end{scriptsize}
\end{center}
\end{sidewaystable}

\begin{sidewaystable}
\begin{center}
\caption{Simulation results from 5000 replications when both models for $\rho_{k}$ and $\pi$ are correctly specified (Study 1) and the third value of $\Lambda$ is considered. ``True'' denotes the true parameter value. Sample size = 250.}
\label{scen1:3:250}
\begin{scriptsize}
\begin{tabular}{r c c c c c c c c c c c c}
  \hline
 & TCF$_1$ & TCF$_2$ & TCF$_3$ & MC.sd$_1$ & MC.sd$_2$ & MC.sd$_3$ & asy.sd$_1$ & asy.sd$_2$ & asy.sd$_3$ & boot.sd$_1$ & boot.sd$_2$ & boot.sd$_3$ \\ 
  \hline 
        & \multicolumn{12}{c}{cut-point = (2,4)} \\
True & 0.5000 & 0.3031 & 0.8031 &  &  &  &  &  &  &  &  &  \\ 
  FI & 0.4997 & 0.3037 & 0.8055 & 0.0498 & 0.0338 & 0.0498 & 0.0453 & 0.0294 & 0.0529 & 0.0498 & 0.0348 & 0.0489 \\ 
  MSI & 0.4999 & 0.3035 & 0.8046 & 0.0519 & 0.0453 & 0.0554 & 0.0480 & 0.0412 & 0.0578 & 0.0522 & 0.0453 & 0.0542 \\ 
  IPW & 0.5003 & 0.3033 & 0.8042 & 0.0617 & 0.0632 & 0.0655 & 0.0616 & 0.0617 & 0.0633 & 0.0624 & 0.0627 & 0.0639 \\ 
  SPE & 0.5001 & 0.3034 & 0.8044 & 0.0564 & 0.0588 & 0.0636 & 0.0564 & 0.0573 & 0.0615 & 0.0570 & 0.0579 & 0.0624 \\ 
  \hline 
        & \multicolumn{12}{c}{cut-point = (2,5)} \\  
  True & 0.5000 & 0.4682 & 0.6651 &  &  &  &  &  &  &  &  &  \\ 
  FI & 0.4997 & 0.4697 & 0.6684 & 0.0498 & 0.0381 & 0.0617 & 0.0453 & 0.0339 & 0.0608 & 0.0498 & 0.0390 & 0.0608 \\ 
  MSI & 0.4999 & 0.4691 & 0.6675 & 0.0519 & 0.0503 & 0.0670 & 0.0480 & 0.0460 & 0.0654 & 0.0522 & 0.0499 & 0.0653 \\ 
  IPW & 0.5003 & 0.4687 & 0.6675 & 0.0617 & 0.0688 & 0.0763 & 0.0616 & 0.0668 & 0.0741 & 0.0624 & 0.0676 & 0.0742 \\ 
  SPE & 0.5001 & 0.4690 & 0.6674 & 0.0564 & 0.0641 & 0.0735 & 0.0564 & 0.0621 & 0.0706 & 0.0570 & 0.0627 & 0.0715 \\ 
  \hline 
        & \multicolumn{12}{c}{cut-point = (2,7)} \\  
  True & 0.5000 & 0.7027 & 0.3349 &  &  &  &  &  &  &  &  &  \\ 
  FI & 0.4997 & 0.7037 & 0.3370 & 0.0498 & 0.0378 & 0.0591 & 0.0453 & 0.0353 & 0.0545 & 0.0498 & 0.0384 & 0.0592 \\ 
  MSI & 0.4999 & 0.7037 & 0.3367 & 0.0519 & 0.0482 & 0.0626 & 0.0480 & 0.0451 & 0.0588 & 0.0522 & 0.0476 & 0.0632 \\ 
  IPW & 0.5003 & 0.7033 & 0.3371 & 0.0617 & 0.0642 & 0.0709 & 0.0616 & 0.0614 & 0.0715 & 0.0624 & 0.0624 & 0.0721 \\ 
  SPE & 0.5001 & 0.7038 & 0.3367 & 0.0564 & 0.0603 & 0.0660 & 0.0564 & 0.0581 & 0.0661 & 0.0570 & 0.0587 & 0.0670 \\ 
  \hline 
        & \multicolumn{12}{c}{cut-point = (4,5)} \\  
  True & 0.8031 & 0.1651 & 0.6651 &  &  &  &  &  &  &  &  &  \\ 
  FI & 0.8037 & 0.1660 & 0.6684 & 0.0393 & 0.0277 & 0.0617 & 0.0366 & 0.0236 & 0.0608 & 0.0388 & 0.0282 & 0.0608 \\ 
  MSI & 0.8033 & 0.1656 & 0.6675 & 0.0425 & 0.0369 & 0.0670 & 0.0400 & 0.0333 & 0.0654 & 0.0420 & 0.0369 & 0.0653 \\ 
  IPW & 0.8033 & 0.1654 & 0.6675 & 0.0486 & 0.0497 & 0.0763 & 0.0469 & 0.0486 & 0.0741 & 0.0475 & 0.0496 & 0.0742 \\ 
  SPE & 0.8032 & 0.1656 & 0.6674 & 0.0469 & 0.0460 & 0.0735 & 0.0457 & 0.0454 & 0.0706 & 0.0460 & 0.0458 & 0.0715 \\ 
  \hline 
       & \multicolumn{12}{c}{cut-point = (4,7)} \\
  True & 0.8031 & 0.3996 & 0.3349 &  &  &  &  &  &  &  &  &  \\ 
  FI & 0.8037 & 0.4000 & 0.3370 & 0.0393 & 0.0419 & 0.0591 & 0.0366 & 0.0383 & 0.0545 & 0.0388 & 0.0430 & 0.0592 \\ 
  MSI & 0.8033 & 0.4002 & 0.3367 & 0.0425 & 0.0513 & 0.0626 & 0.0400 & 0.0480 & 0.0588 & 0.0420 & 0.0519 & 0.0632 \\ 
  IPW & 0.8033 & 0.4000 & 0.3371 & 0.0486 & 0.0639 & 0.0709 & 0.0469 & 0.0643 & 0.0715 & 0.0475 & 0.0652 & 0.0721 \\ 
  SPE & 0.8032 & 0.4004 & 0.3367 & 0.0469 & 0.0604 & 0.0660 & 0.0457 & 0.0605 & 0.0661 & 0.0460 & 0.0612 & 0.0670 \\ 
  \hline 
        & \multicolumn{12}{c}{cut-point = (5,7)} \\ 
  True & 0.8996 & 0.2345 & 0.3349 &  &  &  &  &  &  &  &  &  \\ 
  FI & 0.9000 & 0.2340 & 0.3370 & 0.0271 & 0.0348 & 0.0591 & 0.0255 & 0.0313 & 0.0545 & 0.0269 & 0.0356 & 0.0592 \\ 
  MSI & 0.8998 & 0.2346 & 0.3367 & 0.0312 & 0.0441 & 0.0626 & 0.0296 & 0.0407 & 0.0588 & 0.0310 & 0.0441 & 0.0632 \\ 
  IPW & 0.8998 & 0.2347 & 0.3371 & 0.0359 & 0.0553 & 0.0709 & 0.0347 & 0.0545 & 0.0715 & 0.0355 & 0.0555 & 0.0721 \\ 
  SPE & 0.8998 & 0.2348 & 0.3367 & 0.0351 & 0.0521 & 0.0660 & 0.0347 & 0.0516 & 0.0661 & 0.0348 & 0.0521 & 0.0670 \\ 
   \hline
\end{tabular}
\end{scriptsize}
\end{center}
\end{sidewaystable}

\subsection{Study 2}
In this study, the true disease status $D$ and the test results $T$ and $A$ are generated in the same way as in the first scenario. The true conditional verification process $\pi$, instead, is chosen to be the following function of $T$ and $A$
\[
\pi(T,A) = 0.35 + 0.3 \mathrm{I}\left(T > t^{(0.8)}\right) + 0.35 \mathrm{I}\left(A > a^{(0.8)}\right),
\]
where $t^{(0.8)}$ and $a^{(0.8)}$ correspond to the 80-th percentile of distribution of $T$ and $A$, respectively. In this case, the verification probabilities are 1 for subjects with $T > t^{(0.8)}$ and $A > a^{(0.8)}$; $0.7$ for subjects with $T \le t^{(0.8)}$ and $A > a^{(0.8)}$; $0.65$ for subjects with $T > t^{(0.8)}$ and $A \le a^{(0.8)}$; $0.35$ otherwise. In our setting, the verification rate is approximately $0.48$. 

The aim in this scenario is to evaluate the behavior of the estimators, in particular that of IPW and SPE, under misspecification of the verification process. Therefore, $\hat{\pi}_i$ is estimated from a logistic regression model with $V$ as the response and $T$ as predictor, while $\hat{\rho}_{ki}$ is still obtained from the multinomial logistic model (similarly to the first scenario). Clearly, the model used for verification status is misspecified.

Table \ref{scen2:2:1000} and Tables \ref{scen2:1:1000}--\ref{scen2:3:1000} in Appendix \ref{app:simu:roc}, show Monte Carlo means and standard deviations  for the estimators of the true class fractions TCF$_1$, TCF$_2$ and TCF$_3$. Moreover, estimated  standard deviations (via asymptotic theory) and bootstrap standard deviations are also presented. The results clearly show the effect of misspecification on IPW estimates, despite the high sample size. In particular, in terms of bias, the IPW method performs almost alway poorly, with high distortion in some cases (values highlighted in bold). On the other hand, the SPE estimator behaves well, due to its doubly robustness property.
\begin{sidewaystable}
\begin{center}
\caption{Simulation results from 5000 replications when the  model for the verification process is misspecified (Study 2) and the second value of $\Lambda$ is used. ``True'' indicates the true parameter value. Sample size = 1000.}
\label{scen2:2:1000}
\begin{scriptsize}
\begin{tabular}{r c c c c c c c c c c c c}
  \hline
  & TCF$_1$ & TCF$_2$ & TCF$_3$ & MC.sd$_1$ & MC.sd$_2$ & MC.sd$_3$ & asy.sd$_1$ & asy.sd$_2$ & asy.sd$_3$ & boot.sd$_1$ & boot.sd$_2$ & boot.sd$_3$ \\
   \hline 
        & \multicolumn{12}{c}{cut-point = (2,4)} \\
True & 0.5000 & 0.3970 & 0.8970 &  &  &  &  &  &  &  &  &  \\ 
  FI & 0.4998 & 0.3970 & 0.8977 & 0.0267 & 0.0211 & 0.0207 & 0.0231 & 0.0172 & 0.0268 & 0.0268 & 0.0213 & 0.0204 \\ 
  MSI & 0.4997 & 0.3970 & 0.8978 & 0.0272 & 0.0240 & 0.0220 & 0.0237 & 0.0202 & 0.0279 & 0.0274 & 0.0239 & 0.0219 \\ 
  IPW & \textbf{0.5983} & \textbf{0.3743} & \textbf{0.9150} & 0.0364 & 0.0407 & 0.0257 & 0.0368 & 0.0399 & 0.0284 & 0.0369 & 0.0399 & 0.0258 \\ 
  SPE & 0.4996 & 0.3971 & 0.8980 & 0.0314 & 0.0342 & 0.0268 & 0.0314 & 0.0336 & 0.0267 & 0.0315 & 0.0336 & 0.0268 \\ 
   \hline 
        & \multicolumn{12}{c}{cut-point = (2,5)} \\   
  True & 0.5000 & 0.6335 & 0.7365 &  &  &  &  &  &  &  &  &  \\ 
  FI & 0.4998 & 0.6340 & 0.7381 & 0.0267 & 0.0218 & 0.0332 & 0.0231 & 0.0198 & 0.0344 & 0.0268 & 0.0219 & 0.0326 \\ 
  MSI & 0.4997 & 0.6338 & 0.7381 & 0.0272 & 0.0246 & 0.0344 & 0.0237 & 0.0226 & 0.0356 & 0.0274 & 0.0243 & 0.0338 \\ 
  IPW & \textbf{0.5983} & \textbf{0.5749} & \textbf{0.7965} & 0.0364 & 0.0406 & 0.0348 & 0.0368 & 0.0401 & 0.0387 & 0.0369 & 0.0402 & 0.0346 \\ 
  SPE & 0.4996 & 0.6338 & 0.7383 & 0.0314 & 0.0341 & 0.0368 & 0.0314 & 0.0335 & 0.0364 & 0.0315 & 0.0336 & 0.0364 \\ 
   \hline 
        & \multicolumn{12}{c}{cut-point = (2,7)} \\   
  True & 0.5000 & 0.8682 & 0.2635 &  &  &  &  &  &  &  &  &  \\ 
  FI & 0.4998 & 0.8690 & 0.2639 & 0.0267 & 0.0175 & 0.0295 & 0.0231 & 0.0190 & 0.0280 & 0.0268 & 0.0176 & 0.0290 \\ 
  MSI & 0.4997 & 0.8689 & 0.2639 & 0.0272 & 0.0197 & 0.0308 & 0.0237 & 0.0211 & 0.0292 & 0.0274 & 0.0198 & 0.0302 \\ 
  IPW & \textbf{0.5983} & \textbf{0.8307} & \textbf{0.3054} & 0.0364 & 0.0342 & 0.0347 & 0.0368 & 0.0341 & 0.0358 & 0.0369 & 0.0343 & 0.0343 \\ 
  SPE & 0.4996 & 0.8688 & 0.2639 & 0.0314 & 0.0283 & 0.0316 & 0.0314 & 0.0282 & 0.0308 & 0.0315 & 0.0284 & 0.0310 \\ 
   \hline 
        & \multicolumn{12}{c}{cut-point = (4,5)} \\   
  True & 0.8970 & 0.2365 & 0.7365 &  &  &  &  &  &  &  &  &  \\ 
  FI & 0.8975 & 0.2370 & 0.7381 & 0.0159 & 0.0191 & 0.0332 & 0.0153 & 0.0147 & 0.0344 & 0.0162 & 0.0191 & 0.0326 \\ 
  MSI & 0.8974 & 0.2368 & 0.7381 & 0.0168 & 0.0215 & 0.0344 & 0.0162 & 0.0175 & 0.0356 & 0.0171 & 0.0213 & 0.0338 \\ 
  IPW & \textbf{0.9216} & \textbf{0.2006} & \textbf{0.7965} & 0.0165 & 0.0315 & 0.0348 & 0.0167 & 0.0308 & 0.0387 & 0.0168 & 0.0309 & 0.0346 \\ 
  SPE & 0.8974 & 0.2367 & 0.7383 & 0.0189 & 0.0266 & 0.0368 & 0.0191 & 0.0261 & 0.0364 & 0.0191 & 0.0262 & 0.0364 \\ 
   \hline 
        & \multicolumn{12}{c}{cut-point = (4,7)} \\   
  True & 0.8970 & 0.4711 & 0.2635 &  &  &  &  &  &  &  &  &  \\ 
  FI & 0.8975 & 0.4721 & 0.2639 & 0.0159 & 0.0276 & 0.0295 & 0.0153 & 0.0247 & 0.0280 & 0.0162 & 0.0276 & 0.0290 \\ 
  MSI & 0.8974 & 0.4719 & 0.2639 & 0.0168 & 0.0300 & 0.0308 & 0.0162 & 0.0269 & 0.0292 & 0.0171 & 0.0296 & 0.0302 \\ 
  IPW & \textbf{0.9216} & \textbf{0.4564} & \textbf{0.3054} & 0.0165 & 0.0395 & 0.0347 & 0.0167 & 0.0387 & 0.0358 & 0.0168 & 0.0388 & 0.0343 \\ 
  SPE & 0.8974 & 0.4717 & 0.2639 & 0.0189 & 0.0339 & 0.0316 & 0.0191 & 0.0331 & 0.0308 & 0.0191 & 0.0333 & 0.0310 \\ 
   \hline 
        & \multicolumn{12}{c}{cut-point = (5,7)} \\  
  True & 0.9711 & 0.2347 & 0.2635 &  &  &  &  &  &  &  &  &  \\ 
  FI & 0.9712 & 0.2351 & 0.2639 & 0.0069 & 0.0208 & 0.0295 & 0.0067 & 0.0185 & 0.0280 & 0.0070 & 0.0209 & 0.0290 \\ 
  MSI & 0.9712 & 0.2351 & 0.2639 & 0.0083 & 0.0237 & 0.0308 & 0.0080 & 0.0214 & 0.0292 & 0.0084 & 0.0234 & 0.0302 \\ 
  IPW & \textbf{0.9752} & \textbf{0.2558} & \textbf{0.3054} & 0.0092 & 0.0319 & 0.0347 & 0.0091 & 0.0315 & 0.0358 & 0.0092 & 0.0316 & 0.0343 \\ 
  SPE & 0.9713 & 0.2350 & 0.2639 & 0.0101 & 0.0273 & 0.0316 & 0.0100 & 0.0266 & 0.0308 & 0.0101 & 0.0267 & 0.0310 \\ 
   \hline
\end{tabular}
\end{scriptsize}
\end{center}
\end{sidewaystable}

\subsection{Study 3}
Starting from two independent random variables $Z_1 \sim \mathcal{N}(0,0.5)$ and $Z_2 \sim \mathcal{N}(0,0.5)$, the true conditional disease $D$ is generated by a trinomial random vector $(D_1,D_2,D_3)$ such that
\[
D_1 = \left\{\begin{array}{r l}
1 & \mathrm{ if } \, Z_1 + Z_2 \le h_1 \\
0 & \mathrm{ otherwise}
\end{array}\right. ,
\,
D_2 = \left\{\begin{array}{r l}
1 & \mathrm{ if } \, h_1 < Z_1 + Z_2 \le h_2 \\
0 & \mathrm{ otherwise}
\end{array}\right. ,
\]
\[
D_3 = \left\{\begin{array}{r l}
1 & \mathrm{ if } \, Z_1 + Z_2 > h_2 \\
0 & \mathrm{ otherwise}
\end{array}\right. .
\]
Here, $h_1$ and $h_2$ are two thresholds. We choose $h_1$ and $h_2$ to make $\theta_1 = 0.4$ and $\theta_3 = 0.25$. The continuous test results $T$ and the covariate $A$ are generated to be related to $D$ through $Z_1$ and $Z_2$. More precisely,
\[
T  = 0.5(Z_1 + Z_2) + \varepsilon_1, \qquad A = Z_1 + Z_2 + \varepsilon_2,
\]
where $\varepsilon_1$ and $\varepsilon_2$ are two independent normal random variables with mean $0$ and the common variance $0.25$, independent also from $Z_1$ and $Z_2$. %The true value of VUS is $0.5513$. 
The verification status $V$ is simulated by the following logistic model
\[
\mathrm{logit}\left\{\Pro(V = 1|T,A)\right\} = 0.1 - 1.53 T + A.
\]
Under this model, the verification rate is roughly $0.52$. We consider the cut points $c_1$ and $c_2$ as the pairs $(-1, -0.5), (-1,0.7), (-1, 1.3), (-0.5,0.7), (-0.5,1.3)$ and $(0.7, 1.3)$. In this set--up, we determine the true values of TCF's as
\begin{align}
{\TCF}_{1}(c_1) &= \frac{1}{\Phi(h_1)} \int_{-\infty}^{h_1}\Phi\left(\frac{c_1 - 0.5 z}{\sqrt{0.25}}\right)\phi(z)\ud z\nonumber, \\
{\TCF}_{2}(c_1,c_2) &= \frac{1}{\Phi(h_2) - \Phi(h_1)} \int_{h_1}^{h_2}\left[\Phi\left(\frac{c_2 - 0.5 z}{\sqrt{0.25}}\right) - \Phi\left(\frac{c_1 - 0.5 z}{\sqrt{0.25}}\right)\right]\phi(z)\ud z, \nonumber \\
{\TCF}_{3}(c_2) &= 1 - \frac{1}{1 - \Phi(h_2)} \int_{h_2}^{\infty}\Phi\left(\frac{c_2 - 0.5 z}{\sqrt{0.25}}\right)\phi(z)\ud z. \nonumber
\end{align}

The aim in this scenario is to evaluate the behavior of the estimators, in particular that of FI, MSI and SPE, when the estimators $\hat{\rho}_{ki}$ are inconsistent, whereas $\hat{\pi}_{i}$ are consistent. Therefore, $\hat{\rho}_{ki}$ are obtained from a multinomial logistic regression model with $D = (D_1,D_2,D_3)$ as the response and $T$ as predictor. As the correct process is a multinomial probit process, the chosen model is clearly misspecified. To estimate the conditional verification process $\pi$, we use a generalized linear model for $V$ given $T$ and $A$ with logit link. Clearly, this model is correctly specified.

Table \ref{scen3:1000} shows Monte Carlo means and standard deviations for the estimators of the true class fractions TCF$_1$, TCF$_2$ and TCF$_3$. Moreover, estimated  standard deviations (via asymptotic theory) and bootstrap standard deviations are also presented. The results clearly show the effect of misspecification on FI and MSI estimates, despite the high sample size. In particular, in terms of bias, the two methods performs almost alway poorly, with high distortion in some cases (values highlighted in bold). Again, the SPE  estimator behaves well due to its doubly robustness property.

\begin{sidewaystable}
\begin{center}
\caption{Simulation results from 5000 replications when only model for $\rho_{k}$ is misspecified (Study 3). ``True'' indicates the true parameter value. Sample size = 1000.}
\label{scen3:1000}
\begin{scriptsize}
\begin{tabular}{r c c c c c c c c c c c c}
  \hline
  & TCF$_1$ & TCF$_2$ & TCF$_3$ & MC.sd$_1$ & MC.sd$_2$ & MC.sd$_3$ & asy.sd$_1$ & asy.sd$_2$ & asy.sd$_3$ & boot.sd$_1$ & boot.sd$_2$ & boot.sd$_3$ \\
   \hline 
        & \multicolumn{12}{c}{cut-point = (-1,-0.5)} \\
True & 0.1812 & 0.1070 & 0.9817 &  &  &  &  &  &  &  &  &  \\ 
   FI & \textbf{0.2144} & \textbf{0.1318} & 0.9813 & 0.0230 & 0.0152 & 0.0051 & 0.0243 & 0.0135 & 0.0200 & 0.0230 & 0.0150 & 0.0052 \\ 
   MSI & \textbf{0.2172} & \textbf{0.1328} & 0.9800 & 0.0237 & 0.0182 & 0.0074 & 0.0250 & 0.0166 & 0.0207 & 0.0237 & 0.0179 & 0.0075 \\ 
   IPW & 0.1819 & 0.1072 & 0.9817 & 0.0258 & 0.0197 & 0.0091 & 0.0258 & 0.0194 & 0.0135 & 0.0260 & 0.0196 & 0.0092 \\ 
   SPE & 0.1818 & 0.1073 & 0.9816 & 0.0208 & 0.0206 & 0.0093 & 0.0207 & 0.0202 & 0.0090 & 0.0208 & 0.0204 & 0.0094 \\ 
   \hline 
        & \multicolumn{12}{c}{cut-point = (-1,0.7)} \\   
  True & 0.1812 & 0.8609 & 0.4469 &  &  &  &  &  &  &  &  &  \\ 
  FI & \textbf{0.2144} & \textbf{0.8879} & \textbf{0.4010} & 0.0230 & 0.0149 & 0.0284 & 0.0243 & 0.0153 & 0.0242 & 0.0230 & 0.0146 & 0.0284 \\ 
  MSI & \textbf{0.2172} & \textbf{0.8931} & \textbf{0.4035} & 0.0237 & 0.0165 & 0.0292 & 0.0250 & 0.0172 & 0.0251 & 0.0237 & 0.0165 & 0.0292 \\ 
  IPW & 0.1819 & 0.8606 & 0.4462 & 0.0258 & 0.0350 & 0.0437 & 0.0258 & 0.0342 & 0.0447 & 0.0260 & 0.0348 & 0.0437 \\ 
  SPE & 0.1818 & 0.8608 & 0.4462 & 0.0208 & 0.0311 & 0.0455 & 0.0207 & 0.0305 & 0.0449 & 0.0208 & 0.0310 & 0.0482 \\ 
   \hline 
        & \multicolumn{12}{c}{cut-point = (-1,1.3)} \\   
  True & 0.1812 & 0.9732 & 0.1171 &  &  &  &  &  &  &  &  &  \\ 
  FI & \textbf{0.2144} & 0.9672 & \textbf{0.0949} & 0.0230 & 0.0063 & 0.0161 & 0.0243 & 0.0099 & 0.0104 & 0.0230 & 0.0062 & 0.0161 \\ 
  MSI & \textbf{0.2172} & 0.9708 & \textbf{0.0960} & 0.0237 & 0.0079 & 0.0164 & 0.0250 & 0.0110 & 0.0109 & 0.0237 & 0.0078 & 0.0164 \\ 
  IPW & 0.1819 & 0.9734 & 0.1164 & 0.0258 & 0.0167 & 0.0358 & 0.0258 & 0.0130 & 0.0347 & 0.0260 & 0.0160 & 0.0354 \\ 
  SPE & 0.1818 & 0.9734 & 0.1169 & 0.0208 & 0.0158 & 0.0281 & 0.0207 & 0.0128 & 0.0263 & 0.0208 & 0.0151 & 0.0333 \\ 
   \hline 
        & \multicolumn{12}{c}{cut-point = (-0.5,0.7)} \\   
  True & 0.4796 & 0.7539 & 0.4469 &  &  &  &  &  &  &  &  &  \\ 
  FI & \textbf{0.5497} & 0.7561 & \textbf{0.4010} & 0.0302 & 0.0196 & 0.0284 & 0.0284 & 0.0183 & 0.0242 & 0.0301 & 0.0192 & 0.0284 \\ 
  MSI & \textbf{0.5502} & 0.7603 & \textbf{0.4035} & 0.0312 & 0.0220 & 0.0292 & 0.0295 & 0.0211 & 0.0251 & 0.0310 & 0.0219 & 0.0292 \\ 
  IPW & 0.4801 & 0.7534 & 0.4462 & 0.0390 & 0.0373 & 0.0437 & 0.0384 & 0.0371 & 0.0447 & 0.0387 & 0.0374 & 0.0437 \\ 
  SPE & 0.4801 & 0.7535 & 0.4462 & 0.0327 & 0.0344 & 0.0455 & 0.0322 & 0.0339 & 0.0449 & 0.0324 & 0.0343 & 0.0482 \\ 
   \hline 
        & \multicolumn{12}{c}{cut-point = (-0.5,1.3)} \\  
  True & 0.4796 & 0.8661 & 0.1171 &  &  &  &  &  &  &  &  &  \\ 
  FI & \textbf{0.5497} & \textbf{0.8354} & \textbf{0.0949} & 0.0302 & 0.0189 & 0.0161 & 0.0284 & 0.0185 & 0.0104 & 0.0301 & 0.0186 & 0.0161 \\ 
  MSI & \textbf{0.5502} & \textbf{0.8380} &\textbf{ 0.0960} & 0.0312 & 0.0207 & 0.0164 & 0.0295 & 0.0204 & 0.0109 & 0.0310 & 0.0204 & 0.0164 \\ 
  IPW & 0.4801 & 0.8661 & 0.1164 & 0.0390 & 0.0248 & 0.0358 & 0.0384 & 0.0238 & 0.0347 & 0.0387 & 0.0245 & 0.0354 \\ 
  SPE & 0.4801 & 0.8660 & 0.1169 & 0.0327 & 0.0250 & 0.0281 & 0.0322 & 0.0239 & 0.0263 & 0.0324 & 0.0245 & 0.0333 \\ 
   \hline 
        & \multicolumn{12}{c}{cut-point = (0.7,1.3)} \\   
  True & 0.9836 & 0.1122 & 0.1171 &  &  &  &  &  &  &  &  &  \\ 
  FI & \textbf{0.9933} & \textbf{0.0793} & \textbf{0.0949} & 0.0023 & 0.0133 & 0.0161 & 0.0021 & 0.0119 & 0.0104 & 0.0023 & 0.0131 & 0.0161 \\ 
  MSI & \textbf{0.9930} & \textbf{0.0777} & \textbf{0.0960} & 0.0038 & 0.0145 & 0.0164 & 0.0032 & 0.0135 & 0.0109 & 0.0038 & 0.0145 & 0.0164 \\ 
  IPW & 0.9839 & 0.1128 & 0.1164 & 0.0183 & 0.0324 & 0.0358 & 0.0122 & 0.0319 & 0.0347 & 0.0173 & 0.0325 & 0.0354 \\ 
  SPE & 0.9839 & 0.1125 & 0.1169 & 0.0180 & 0.0283 & 0.0281 & 0.0122 & 0.0280 & 0.0263 & 0.0170 & 0.0285 & 0.0333 \\ 
   \hline
\end{tabular}
\end{scriptsize}
\end{center}
\end{sidewaystable}

\subsection{Study 4}
We generate data exactly as in Study 3. The aim in this scenario is to evaluate the behavior of FI, MSI, IPW and SPE estimators when the estimates $\hat{\rho}_{ki}$ and $\hat{\pi}_{i}$ are inconsistent. Therefore, $\hat{\rho}_{ki}$ are obtained from a multinomial logistic regression model with $D = (D_1,D_2,D_3)$ as the response and $T$ as predictor. This model is misspecified. To estimate the conditional verification disease $\pi$, we use a generalized linear model for $V$ given $T$ and $A^{2/3}$ with logit link. Clearly, this model is misspecified.

Table \ref{scen4:1000} shows Monte Carlo means and standard deviations  for the estimators of the true class fractions TCF$_1$, TCF$_2$ and TCF$_3$. 
%\centerline{[Insert Table \ref{tab:(iv)} around here]}
Moreover,  estimated standard deviations (via asymptotic theory) and bootstrap standard deviations are also presented.
The results clearly show that when both the disease and verification models are misspecified, all estimators may behave poorly, with high distortion in some cases (values highlighted in bold).

\begin{sidewaystable}
\begin{center}
\caption{Simulation results from 5000 replications when both models for $\rho_{k}$ and $\pi$ are misspecified (Study 4). ``True'' indicates the true parameter value. Sample size = 1000.}
\label{scen4:1000}
\begin{scriptsize}
\begin{tabular}{r c c c c c c c c c c c c}
  \hline
  & TCF$_1$ & TCF$_2$ & TCF$_3$ & MC.sd$_1$ & MC.sd$_2$ & MC.sd$_3$ & asy.sd$_1$ & asy.sd$_2$ & asy.sd$_3$ & boot.sd$_1$ & boot.sd$_2$ & boot.sd$_3$ \\
   \hline 
        & \multicolumn{12}{c}{cut-point = (-1,-0.5)} \\
  	True & 0.1812 & 0.1070 & 0.9817 &  &  &  &  &  &  &  &  &  \\ 
  FI & \textbf{0.2143} & \textbf{0.1320} & 0.9814 & 0.0231 & 0.0149 & 0.0051 & 0.0243 & 0.0135 & 0.0200 & 0.0230 & 0.0150 & 0.0052 \\ 
  MSI & \textbf{0.2170} & \textbf{0.1330} & 0.9801 & 0.0238 & 0.0179 & 0.0074 & 0.0250 & 0.0166 & 0.0207 & 0.0237 & 0.0179 & 0.0075 \\ 
  IPW & \textbf{0.2185} & \textbf{0.1339} & 0.9792 & 0.0284 & 0.0234 & 0.0102 & 0.0282 & 0.0232 & 0.0105 & 0.0283 & 0.0233 & 0.0102 \\ 
  SPE & \textbf{0.2183} & \textbf{0.1339} & 0.9792 & 0.0247 & 0.0220 & 0.0101 & 0.0245 & 0.0219 & 0.0098 & 0.0246 & 0.0219 & 0.0102 \\  
   \hline 
        & \multicolumn{12}{c}{cut-point = (-1,0.7)} \\   
  True & 0.1812 & 0.8609 & 0.4469 &  &  &  &  &  &  &  &  &  \\ 
  FI & \textbf{0.2143} & \textbf{0.8887} & \textbf{0.4002} & 0.0231 & 0.0143 & 0.0285 & 0.0243 & 0.0153 & 0.0242 & 0.0230 & 0.0146 & 0.0285 \\ 
  MSI & \textbf{0.2170} & \textbf{0.8940} & \textbf{0.4029} & 0.0238 & 0.0164 & 0.0290 & 0.0250 & 0.0171 & 0.0251 & 0.0237 & 0.0165 & 0.0292 \\ 
  IPW & \textbf{0.2185} & \textbf{0.8994} & \textbf{0.4078} & 0.0284 & 0.0237 & 0.0397 & 0.0282 & 0.0232 & 0.0410 & 0.0283 & 0.0234 & 0.0397 \\ 
  SPE & \textbf{0.2183} & \textbf{0.8998} & \textbf{0.4071} & 0.0247 & 0.0223 & 0.0323 & 0.0245 & 0.0219 & 0.0325 & 0.0246 & 0.0220 & 0.0326 \\ 
   \hline 
        & \multicolumn{12}{c}{cut-point = (-1,1.3)} \\   
  True & 0.1812 & 0.9732 & 0.1171 &  &  &  &  &  &  &  &  &  \\ 
  FI & \textbf{0.2143} & 0.9675 & \textbf{0.0947} & 0.0231 & 0.0061 & 0.0160 & 0.0243 & 0.0099 & 0.0104 & 0.0230 & 0.0062 & 0.0161 \\ 
  MSI & \textbf{0.2170} & 0.9711 & \textbf{0.0958} & 0.0238 & 0.0078 & 0.0163 & 0.0250 & 0.0110 & 0.0108 & 0.0237 & 0.0078 & 0.0164 \\ 
  IPW & \textbf{0.2185} & 0.9742 & \textbf{0.0977} & 0.0284 & 0.0112 & 0.0269 & 0.0282 & 0.0107 & 0.0270 & 0.0283 & 0.0111 & 0.0273 \\ 
  SPE & \textbf{0.2183} & 0.9742 & \textbf{0.0978} & 0.0247 & 0.0110 & 0.0174 & 0.0245 & 0.0105 & 0.0175 & 0.0246 & 0.0108 & 0.0177 \\ 
   \hline 
        & \multicolumn{12}{c}{cut-point = (-0.5,0.7)} \\   
  True & 0.4796 & 0.7539 & 0.4469 &  &  &  &  &  &  &  &  &  \\ 
  FI & \textbf{0.5510} & 0.7567 & \textbf{0.4002} & 0.0306 & 0.0190 & 0.0285 & 0.0285 & 0.0183 & 0.0242 & 0.0301 & 0.0192 & 0.0285 \\ 
  MSI & \textbf{0.5514} & 0.7610 & \textbf{0.4029} & 0.0316 & 0.0219 & 0.0290 & 0.0295 & 0.0211 & 0.0251 & 0.0310 & 0.0219 & 0.0292 \\ 
  IPW & \textbf{0.5509} & 0.7655 & \textbf{0.4078} & 0.0360 & 0.0313 & 0.0397 & 0.0357 & 0.0310 & 0.0410 & 0.0358 & 0.0311 & 0.0397 \\ 
  SPE & \textbf{0.5509} & 0.7659 & \textbf{0.4071} & 0.0336 & 0.0286 & 0.0323 & 0.0329 & 0.0286 & 0.0325 & 0.0329 & 0.0287 & 0.0326 \\ 
   \hline 
        & \multicolumn{12}{c}{cut-point = (-0.5,1.3)} \\  
  True & 0.4796 & 0.8661 & 0.1171 &  &  &  &  &  &  &  &  &  \\ 
  FI & \textbf{0.5510} & \textbf{0.8355} & \textbf{0.0947} & 0.0306 & 0.0186 & 0.0160 & 0.0285 & 0.0186 & 0.0104 & 0.0301 & 0.0186 & 0.0161 \\ 
  MSI & \textbf{0.5514} & \textbf{0.8380} & \textbf{0.0958} & 0.0316 & 0.0205 & 0.0163 & 0.0295 & 0.0204 & 0.0108 & 0.0310 & 0.0204 & 0.0164 \\ 
  IPW & \textbf{0.5509} & \textbf{0.8403} & \textbf{0.0977} & 0.0360 & 0.0255 & 0.0269 & 0.0357 & 0.0251 & 0.0270 & 0.0358 & 0.0252 & 0.0273 \\ 
  SPE & \textbf{0.5509} & \textbf{0.8403} & \textbf{0.0978} & 0.0336 & 0.0240 & 0.0174 & 0.0329 & 0.0237 & 0.0175 & 0.0329 & 0.0238 & 0.0177 \\
   \hline 
        & \multicolumn{12}{c}{cut-point = (0.7,1.3)} \\   
  True & 0.9836 & 0.1122 & 0.1171 &  &  &  &  &  &  &  &  &  \\ 
  FI & \textbf{0.9934} & \textbf{0.0788} & \textbf{0.0947} & 0.0022 & 0.0129 & 0.0160 & 0.0021 & 0.0119 & 0.0104 & 0.0023 & 0.0131 & 0.0161 \\ 
  MSI & \textbf{0.9930} & \textbf{0.0771} & \textbf{0.0958} & 0.0038 & 0.0145 & 0.0163 & 0.0032 & 0.0134 & 0.0108 & 0.0037 & 0.0145 & 0.0164 \\ 
  IPW & \textbf{0.9925} & \textbf{0.0748} & \textbf{0.0977} & 0.0075 & 0.0213 & 0.0269 & 0.0057 & 0.0208 & 0.0270 & 0.0073 & 0.0211 & 0.0273 \\ 
  SPE & \textbf{0.9925} & \textbf{0.0744} & \textbf{0.0978} & 0.0074 & 0.0201 & 0.0174 & 0.0058 & 0.0196 & 0.0175 & 0.0073 & 0.0198 & 0.0177 \\ 
   \hline
\end{tabular}
\end{scriptsize}
\end{center}
\end{sidewaystable}

\section{Two illustrations}
To illustrate the application of the proposed methods, in this section we consider two quite distinct real data examples,
both dealing with epithelial ovarian cancer  (EOC).
In the first illustration, we consider diagnosis of EOC in one of three classes i.e.,  benign disease, early stage and late stage  
cancer on the basis of a well known tumor marker, i.e., CA125. We make use of a  publicly available dataset in which the disease status is known for all subjects. Then, we simulate a verification process and apply our estimators. This allows to compare results obtained in the complete data case with those obtained in the incomplete data case after correcting for verification bias.
In the second illustration, we focus on prediction of patients' response to chemotherapy, classified as sensitive, partially sensitive and resistant. Data are available for late stage  EOC patients. In this second example, the response is missing for 
about 25\% of the subjects involved in the study.

\subsection{Diagnosis of EOC}
We use data  from the Pre-PLCO Phase II Dataset from the SPORE/Early Detection Network/Prostate, Lung, Colon, and Ovarian Cancer Ovarian Validation Study. The study protocol and data are publicly available at the address\footnote{\url{http://edrn.nci.nih.gov/protocols/119-spore- edrn-pre-plco-ovarian-phase-ii-validation}}, along with descriptions of the study aims and analytic methods. In particular, we consider  the following three classes of EOC, i.e., benign disease, early stage (I and II) and late stage (III and IV) cancer, and 2 of the 59 available biomarkers, i.e.  CA125 and CA153, measured at Harvard laboratories. In detail, we use CA125 as the test $T$ s and CA153 as a covariate. Reasons for using CA153 as a covariate come from the medical literature that suggests that the concomitant measurement of CA153 with CA125 could be advantageous in the pre-operative discrimination of benign and malignant ovarian tumors. In addition, age of patients is also considered. Here, we have 134 patients with benign disease, 67 early stage samples and 77 late stage samples.

To mimic verification bias, a  subset of the complete dataset is constructed using the test $T$ and the vector $A=(A_1, A_2)$ of the two covariates, namely the marker CA153 ($A_1$) and age ($A_2$).
In this subset, $T$ and $A$ are known for all samples, but the true status (benign, early stage or late stage) is available only for some samples, that we select according to the following mechanism. We select all samples having a value for both $T$ and $A$ above their respective medians, i.e. 0.87 and (45,0.30); as for  the others, we apply the following selection process
\[
\Pro(V = 1) = 0.05 + \delta_1\mathrm{I}(T > 0.87) + \delta_2\mathrm{I}(A_1 > 0.30) + \delta_3\mathrm{I}(A_2 > 45),
\]
with $\delta_1 = 0.35$, $\delta_2 = 0.25$ and $\delta_3 = 0.35,$ leading to a marginal probability of selection equal to $0.634$. With such a choice, the verification probability is equal to about $0.65$ for subjects with $T > 0.87$, $A_1 > 0.30$ and $A_2 < 45$; $0.75$ for subjects with $T > 0.87$, $A_1 < 0.30$ and $A_2 > 45$; $0.65$ for subjects with $T < 0.87$, $A_1 > 0.30$ and $A_2 > 45$; $0.4$ for subjects with $T > 0.87$, $A_1 < 0.30$ and $A_2 < 45$; $0.3$ for subjects with $T < 0.87$, $A_1 > 0.30$ and $A_2 < 45$; $0.4$ for subjects with $T < 0.87$, $A_1 < 0.30$ and $A_2 > 45$; $0.05$ otherwise.

To apply FI, MSI and SPE estimators, we employ a  multinomial logistic model to estimate $\rho_{ki} = \Pro(D_{ki} = 1|T_i,A_{1i},A_{2i})$, where $D_k = 1$, $k = 1,2,3$ refers to  benign, early and late, respectively. On the other hand, SPE and IPW methods require estimates of $\pi_i = \Pro(V_{i} = 1|T_i,A_{1i},A_{2i})$. For estimating such quantities, we make use, firstly, of a correctly specified model, i.e., a linear threshold regression model and, then, of a misspecified model, i.e., a logistic model.

The estimated ROC surfaces for the test $T$ (CA125) obtained by applying the proposed methods are shown in Figure \ref{fg:com4}.

\begin{figure}
\begin{center}
\begin{tabular}{@{}c@{}c@{}}
  \multicolumn{2}{c}{ }\\[-0.5cm]
  \includegraphics[width=0.35\linewidth]{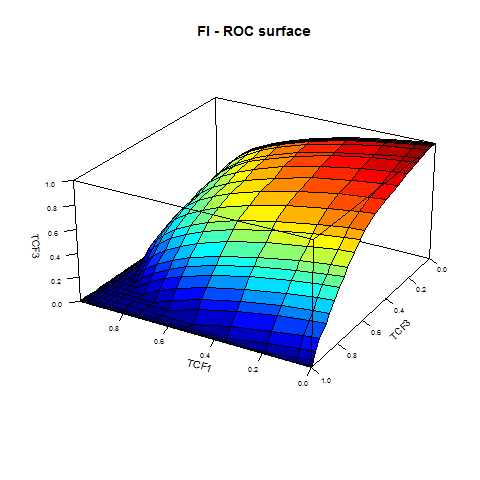}&
  \includegraphics[width=0.35\linewidth]{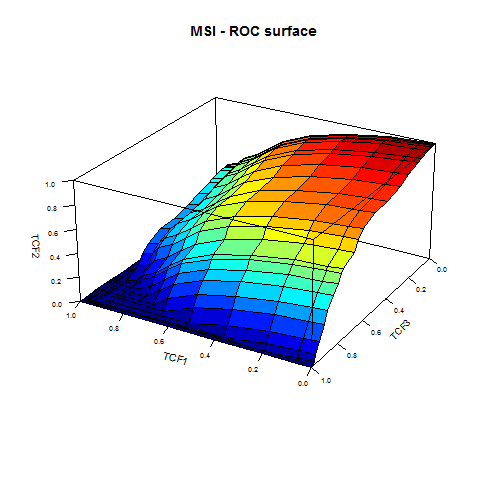}\\[-0.1em]
      \footnotesize (a) FI  & \footnotesize (b) MSI \\
  \multicolumn{2}{c}{\footnotesize }\\[-0.3cm]
  \includegraphics[width=0.35\linewidth]{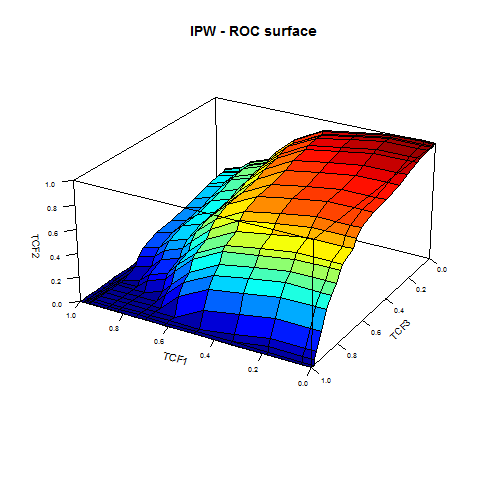}&
  \includegraphics[width=0.35\linewidth]{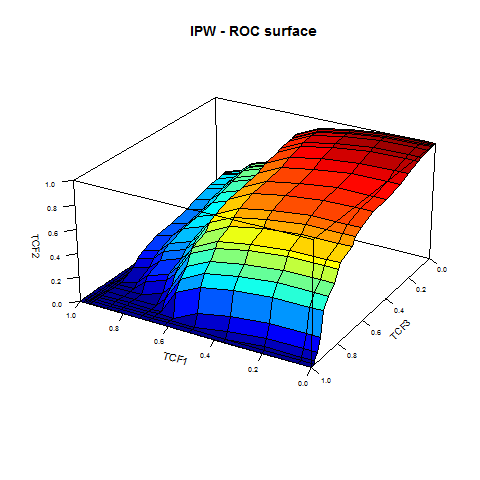}\\[-0.1em]
    \footnotesize (c) IPW--threshold model & \footnotesize (d) IPW--logit model\\
  \multicolumn{2}{c}{\footnotesize }\\[-0.3cm]  
  \includegraphics[width=0.35\linewidth]{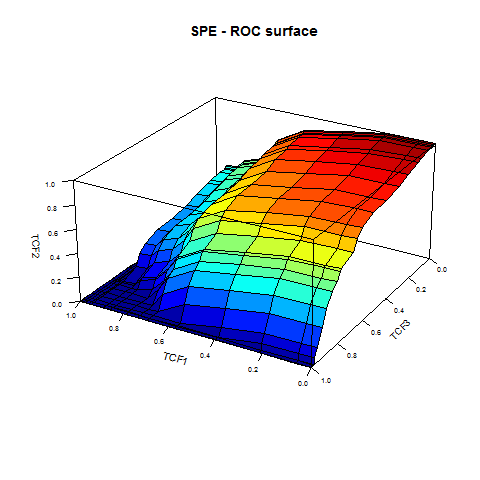}&
  \includegraphics[width=0.35\linewidth]{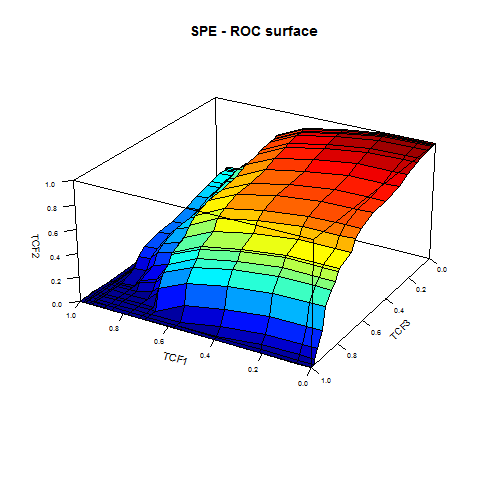}\\[-0.1em]
    \footnotesize (e) SPE--threshold model & \footnotesize (f) SPE--logit model\\
\end{tabular}
\caption{Bias--corrected estimated ROC surfaces for CA125, assessing the classification into three class of EOC: benign disease, early stage (I and II) and late stage (III and IV). For the SPE and IPW approaches, results for the correctly specified and misspecified model for the verification process are given.}
\label{fg:com4}
\end{center}
\end{figure}
For the sake of comparison, we also produced the estimate of the ROC surface  with full data (Full estimate), reported in Appendix \ref{app:fg:ex1}, Figure \ref{fg:full:ex1}. In Appendix \ref{app:fg:ex1}, Figure \ref{fg:roc_threshold_projection} and Figure \ref{fg:roc_logistic_projection}, we also give the projections of the bias--corrected estimated ROC surfaces to the planes defined by $\TCF_1$ versus $\TCF_2$, $\TCF_1$ versus $\TCF_3$ and $\TCF_2$ versus $\TCF_3$, i.e., the ROC curves between classes 1 and 2, classes 1 and 3, classes 2 and 3. Such plots are obtained by setting $\TCF_3 = 0$, $\TCF_2 = 0$ and $\TCF_1 = 0$, respectively. For example, the estimated ROC curves between classes 1 and 2 are defined as the set of points
\[
\left\{(\widehat{\TCF}_{1,*}(c_1), \widehat{\TCF}_{2,*}(c_1,+\infty)), c_1 \in \mathbb{R}\right\},
\]
that is, we ignore the cut point $c_2$. This is a construction equivalent to the most popular representation of an estimated ROC curve, which usually depicts $\widehat{\TCF}_{2}$ versus $1-\widehat{\TCF}_{1}.$   

Compared with the Full estimate, all the bias-corrected methods discussed in the paper seem to behave well, yielding reasonable estimates of the ROC surface and the ROC
curves. Moreover, Table~\ref{tab:result1} shows the VUS estimates obtained with the FI, MSI, IPW and SPE estimators (both for the correctly specified and misspecified model for the verification process), along with approximated 95\% confidence intervals.  Inspection of the table highlights that estimators with better performance are, overall, IPW and SPE. This might be an indication that the multinomial logistic model chosen for the disease process might not be fully adequate in this case.
\begin{table}[h]
\begin{center}
\caption{Bias--corrected (and Full) estimated VUS  for the marker CA125, assessing the classification into three classes of EOC: benign disease, early stage (I and II) and late stage (III and IV). For the SPE and IPW approaches, results for the correctly specified and misspecified model for the verification process are given.}
\label{tab:result1}
\begin{tabular}{l c c c c}
\hline
 & VUS Estimate  &  Asy.sd & Boot.sd &  95\% C.I. (with Asy.sd)  \\
\hline
Full & 0.5663 & & & \\
FI   & 0.5150 & 0.0404 & 0.0417 & (0.4357, 0.5942) \\
MSI  & 0.5183 & 0.0415 & 0.0431 & (0.4368, 0.5997) \\
IPW.logit & 0.5500 & 0.0416 & 0.0471 & (0.4685, 0.6314) \\
SPE.logit & 0.5581 & 0.0443 & 0.0463 & (0.4712, 0.6450) \\
IPW.thres & 0.5353 & 0.0393 & 0.0457 & (0.4583, 0.6123) \\
SPE.thres & 0.5470 & 0.0440 & 0.0438 & (0.4608, 0.6331) \\
\hline
\end{tabular}
\end{center}
\end{table}

\subsection{Prediction of response to chemotherapy}
A major challenge in advanced-stage EOC is prediction of response to platinum chemotherapy on the basis of markers measured at molecular level.  Indeed, several genomic profiling studies have shown that gene expressions relate with different aspects of ovarian cancer (tumor subtype, stage, grade, prognosis, and therapy resistance), although the measured association is usually very low. Here, we consider a cohort of 99 snap-frozen tumor biopsies taken from a frozen tissue bank, located at the Department of Oncology, IRCCS-Mario Negri Institute, Milano, Italy. Biopsies were collected from late stage (III and IV) cancer patients who underwent surgery at the Obstetrics and  Gynaecology Department, San Gerardo Hospital, Monza, Italy between September 1992 and March 2010. For 75 of the 99 subjects, the three-class response to platinum therapy is available, being 31 patients sensitive, 11 partially sensitive and 33 resistant. For all the subjects, we consider as test predictive of the response to therapy the marker ($T$) resulting as a given linear combination of the logarithm of the expression levels of six genes, i.e., Entrez Gene ID: 23513, 7284, 128408, 56996, 2969, 6170. As a covariate, we consider age at onset of patients.

The estimated ROC surfaces for $T$ obtained by applying the proposed methods are shown in Figure \ref{fg:com5}. FI, MSI, IPW and SPE estimators are based on the multinomial logistic model for the disease process and/or the logistic model for the verification process.  
\begin{figure}[h]
\begin{center}
\begin{tabular}{@{}c@{}c@{}}
  \multicolumn{2}{c}{ }\\[-0.5cm]
  \includegraphics[width=0.35\linewidth]{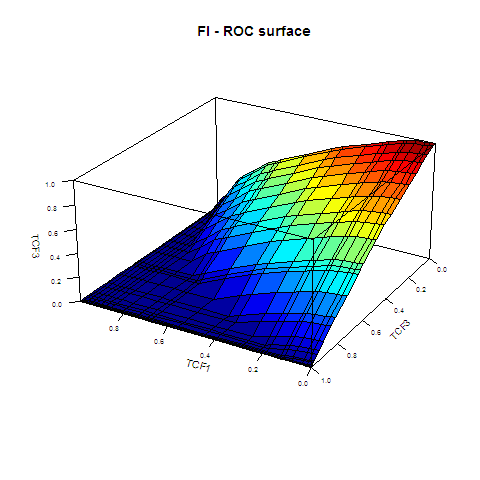}&
  \includegraphics[width=0.35\linewidth]{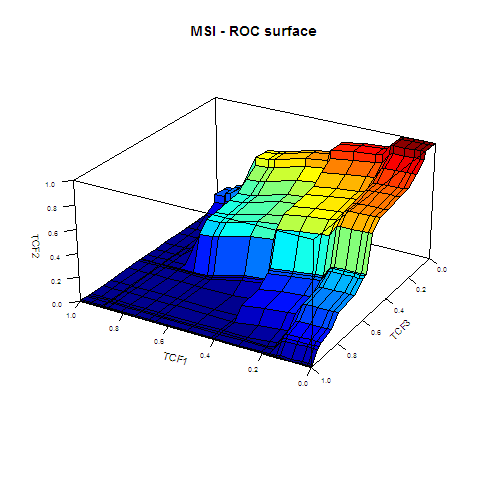}\\[-0.1em]
      \footnotesize (a) FI  & \footnotesize (b) MSI \\
  \multicolumn{2}{c}{\footnotesize }\\[-0.3cm]
  \includegraphics[width=0.35\linewidth]{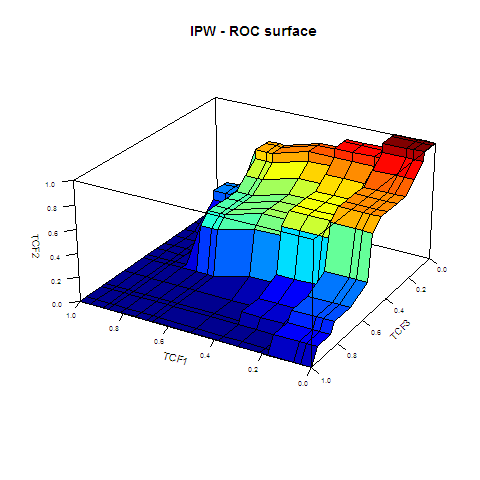}&
  \includegraphics[width=0.35\linewidth]{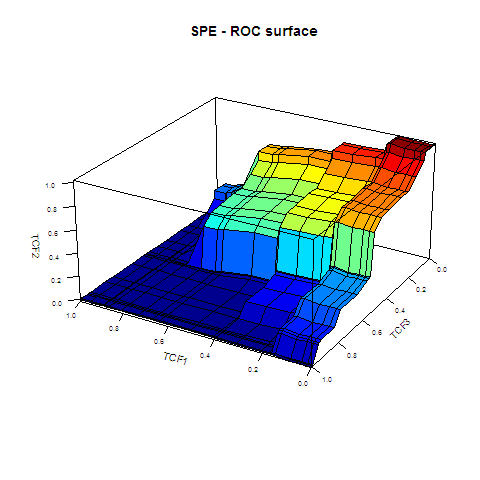}\\[-0.1em]
    \footnotesize (c) IPW  & \footnotesize (d) SPE \\
\end{tabular}
\caption{Bias--corrected estimated ROC surfaces for the test $T$ predicting the response to therapy of late stage EOC patients.}
\label{fg:com5}
\end{center}
\end{figure}
Table \ref{tab:result2} shows the corresponding VUS estimates, along with the na\"ive estimate. The table also gives the estimated standard deviations (via asymptotic theory), bootstrap standard deviations and approximated 95\% confidence intervals. Despite the limited sample size, the results show that $T$ has some ability to predict response to therapy for late stage EOC patients.
\begin{table}[h]
\begin{center}
\caption{Bias--corrected (and Na\"ive) estimated VUS  for the test $T$ predicting the response to therapy
of late stage EOC patients.}
\label{tab:result2}
\begin{tabular}{l c c c c}
\hline 
 & VUS Estimate  &  Asy.sd  & Boot.sd & 95\% C.I. (with Asy.sd) \\
\hline
Na\"ive & 0.3452 & & & \\
FI   & 0.3005 & 0.0512 & 0.0538 & (0.2002, 0.4009) \\
MSI  & 0.3197 & 0.0629 & 0.0656 & (0.1963, 0.4430) \\
IPW  & 0.3231 & 0.0654 & 0.0755 & (0.1949, 0.4512) \\
SPE  & 0.3110 & 0.0675 & 0.0704 & (0.1787, 0.4433) \\
\hline
\end{tabular}
\end{center}
\end{table}

\section{Conclusions}
\label{s:discussion}
This paper proposed several verification bias-corrected estimators of the ROC surface (and the VUS) of a continuous diagnostic test. These estimators, which can be considered an extension to the three-class case of estimators in \cite{alo:05}, are partially parametric in that they require the choice of a parametric model
for the estimation of the disease process, or of the verification process, or of both processes. In some cases, wrong specifications of such models can visibly affect the produced estimates, as highlighted also by our results in the simulation studies. To avoid misspecification problems, one possibility could be to resort on fully nonparametric estimators. This topic will be the focus of our future work.

\appendix
\section{Asymptotic distribution results}
\label{s:asymptotic}

In this section, we discuss validity of  conditions (C1), (C2) and (C3) for the proposed estimators. The discussion covers first the elements of the estimating functions corresponding to  the parameter $\tau$.  Then, we pass to the elements of the estimating functions corresponding to the parameters $\theta_1,\theta_2,\theta_{11},\beta_{12},\beta_{22},\beta_{23}$, specializing the discussion to the various methods. Finally, we give the explicit form of the variance-covariance matrix in Theorem 1. Recall that $\alpha_0$ denotes the true value of $\alpha.$

\noindent
{\bf Parameter $\tau$}. We noted in Section~3 that  estimators FI, MSI and SPE require a multinomial logistic or probit regression model to estimate the disease probabilities $\rho_{ki} = \Pro(D_{ki} = 1|T_i,A_i)$ with $k = 1,2,3$. In the following, we adopt the multinomial logistic model, but  arguments similar to those  given below hold also for the multinomial probit model, despite the rather more complex algebra (see \citet{dag}, Chapter 3, as a general reference).

The estimating function for the nuisance parameter $\tau \equiv \tau_\rho = (\tau_{\rho_1},\tau_{\rho_2})^\top,$ 
\[
G^{\tau_\rho}(\alpha) = \left(G^{\tau_{\rho_1}}(\alpha), G^{\tau_{\rho_2}}(\alpha)\right)^\top \equiv \left(\sum_{i=1}^{n} g_i^{\tau_{\rho_1}}(\alpha), \sum_{i=1}^{n} g_i^{\tau_{\rho_2}}(\alpha)\right)^\top,
\]
is obtained as the first derivative of the log likelihood function. With the multinomial logistic model, we get
\begin{eqnarray}
G^{\tau_\rho}(\alpha) = \left(\sum_{i=1}^{n}V_i U_i (D_{1i} - \rho_{1i}), \sum_{i=1}^{n}V_i  U_i (D_{2i} - \rho_{2i})\right)^\top,
\nonumber
%\label{est_equ:tau_logit}
\end{eqnarray}
where $U_i = (1,T_i,A_i)^\top$. Under assumption (A2), condition (C1) trivially holds. Moreover, 
we get
\begin{eqnarray}
\begin{array}{r l r l}
\frac{\partial}{\partial \tau_{\rho_1}} g_i^{\tau_{\rho_1}}(\alpha) &= - V_i U_i U_i^\top \rho_{1i} (1 - \rho_{1i}) , &
\frac{\partial}{\partial \tau_{\rho_2}} g_i^{\tau_{\rho_1}}(\alpha) &= V_i U_i U_i^\top \rho_{1i} \rho_{2i}, \\ [16pt]
\frac{\partial}{\partial \tau_{\rho_2}} g_i^{\tau_{\rho_2}}(\alpha) &= - V_i U_i U_i^\top \rho_{2i}(1 - \rho_{2i}), & 
\frac{\partial}{\partial \tau_{\rho_1}} g_i^{\tau_{\rho_2}}(\alpha) &= V_i U_i U_i^\top \rho_{1i} \rho_{2i},
\end{array}\nonumber
\label{deri_tau:1}
\end{eqnarray}
and $\frac{\partial}{\partial \theta_s} g^{\tau_\rho}_i(\alpha)=0,$  $\frac{\partial}{\partial \beta_{jk}} g^{\tau_\rho}_i(\alpha)=0$ for each $s,j,k.$ The second--order partial derivatives can be easily derived. Hence, for  $G^{\tau_\rho}(\alpha)$, condition (C2) holds and, by assumption (A3)--(A5) condition (C3) also holds.
 
The IPW and SPE estimators require estimates of $\pi_i = \Pro(V_i = 1| T_i,A_i)$. With  $T$ and $A$ as covariates, we can use the logistic or probit models to this end. In these cases, conditions (C1)--(C3) are  satisfied by the estimating functions
\begin{eqnarray}
G^{\tau_\pi}(\alpha) = \sum_{i=1}^{n}g_i^{\tau_\pi}(\alpha) = \sum_{i=1}^{n} U_i (V_i - \pi_i)
\nonumber
%\label{est_equ:tau_logit:pi}
\end{eqnarray}
or
\begin{eqnarray}
G^{\tau_\pi}(\alpha) = \sum_{i=1}^{n}g_i^{\tau_\pi}(\alpha) = \sum_{i=1}^{n}\left[\frac{V_i U_i \phi(U_i^\top \tau_\pi)}{\Phi(U_i^\top \tau_\pi)} - (1 - V_i)\frac{U_i \phi(U_i^\top \tau_\pi)}{1 - \Phi(U_i^\top \tau_\pi) }\right],
\nonumber
%\label{est_equ:tau_probit:pi}
\end{eqnarray}
where $\phi(\cdot)$ and $\Phi(\cdot)$ are the density function and the cumulative distribution function of the standard normal random variable, respectively.
Recall that $\tau_\pi$ is the component of nuisance parameter $\tau$ corresponding the model for estimating $\pi$.  The first-order derivatives are
\begin{eqnarray}
\frac{\partial}{\partial \tau_\pi} g_i^{\tau_\pi}(\alpha) = - U_i U_i^\top \pi_i (1 - \pi_i),
\nonumber
%\label{deri_tau_logit:pi}
\end{eqnarray}
or
\begin{eqnarray}
\frac{\partial}{\partial \tau_\pi} g_i^{\tau_\pi}(\alpha) &=& - \frac{V_i U_iU_i^\top \phi(U_i^\top \tau_\pi)\left[-U_i^\top \tau_\pi \Phi(U_i^\top \tau_\pi) - \phi(U_i^\top \tau_\pi)\right]}{\Phi^2(U_i^\top \tau_\pi)}\nonumber\\
&& - \: (1-V_i)\frac{U_i U_i^\top \phi(U_i^\top \tau_\pi)\left[U_i^\top \tau_\pi (\Phi(U_i^\top \tau_\pi)-1) + \phi(U_i^\top \tau_\pi)\right]}{\left[1 - \Phi(U_i^\top \tau_\pi)\right]^2}.\nonumber
%\label{deri_tau_probit:pi}
\end{eqnarray}

\noindent
{\bf FI and MSI estimators}.  
According to equations (\ref{est_eq:fi2}), (\ref{est_eq:msi1}) and (\ref{est_eq:msi2}), the estimating functions $G_{*}^{\theta_s}(\alpha)$ for FI and MSI estimators can be presented in the form
\begin{eqnarray}
G_{\mathrm{IE}}^{\theta_s}(\alpha) \equiv \sum_{i=1}^{n}g_{i,\mathrm{IE}}^{\theta_s}(\alpha) = \sum_{i=1}^{n}\left\{V_i \left[m D_{si} - \theta_s + (1-m)\rho_{si}\right] + (1 - V_i)(\rho_{si} - \theta_s)\right\},
\nonumber
%\label{est_equ:theta_IE}
\end{eqnarray}
with $s = 1,2$. Similarly,
\begin{eqnarray}
G_{\mathrm{IE}}^{\beta_{jk}}(\alpha) &\equiv& \sum_{i=1}^{n} g_{i,\mathrm{IE}}^{\beta_{jk}}(\alpha) \nonumber \\
&=& \sum_{i=1}^{n}\bigg\{V_i \left[m \mathrm{I}(T_i \ge c_j) D_{ki} - \beta_{jk} + (1-m)\mathrm{I}(T_i \ge c_j)\rho_{ki}\right] \nonumber\\
&& + \: (1 - V_i)\left[\mathrm{I}(T_i \ge c_j)\rho_{ki} - \beta_{jk} \right]\bigg\},
\nonumber
%\label{est_equ:beta_IE}
\end{eqnarray}
for $j = 1,2$ and $k = 1,2,3$. Here, the notation IE means  ``imputation estimator''. The estimating function corresponds to the FI estimator if $m = 0$, to the MSI estimator if $m = 1$. Using the conditional expectation and the assumption (A1), $\E\left[g_{i,\mathrm{IE}}^{\theta_s}(\alpha_0) \right]$ equals
\begin{align}
\lefteqn{\E_{D_s,T_i,A_i}\left[\E\left[g_{i,\mathrm{IE}}^{\theta_s}(\alpha_0)| T_i,A_i \right]\right]} \nonumber\\
&= \E_{D_{si},T_i,A_i} \left[\E\left[\left\{V_i \left[m D_{si} - \theta_{s0} + (1-m)\rho_{si}\right] + (1 - V_i)\left[\rho_{si} - \theta_{s0} \right]\right\}| T_i,A_i \right]\right] \nonumber\\
&= \E_{D_{si},T_i,A_i} \left[\pi_i \left[m \E \left[D_{si} | T_i,A_i\right] - \theta_{s0} + (1-m)\rho_{si}\right] + (1 - \pi_i)(\rho_{si} - \theta_{s0})\right] \nonumber \\
&= \E_{D_{si},T_i,A_i} \left[\pi_i \left[m \rho_{si} - \theta_{s0} + (1-m)\rho_{si}\right] + (1 - \pi_i)(\rho_{si} - \theta_{s0})\right] \nonumber \\
&= \E_{D_{si},T_i,A_i} \left[\pi_i (\rho_{si} - \theta_{s0}) + (1 - \pi_i)(\rho_{si} - \theta_{s0})\right] \nonumber \\
&= \E_{D_{si},T_i,A_i} \left[\rho_{si} - \theta_{s0}\right] = 0 \nonumber.
\end{align}
Similarly, we compute the expected value of the estimating function components $g_{i,\mathrm{IE}}^{\beta_{jk}}(\alpha_0)$ as follows
\begin{align}
\lefteqn{\E_{D_k,T_i,A_i}\left[\E\left[g_{i,\mathrm{IE}}^{\beta_{jk}}(\alpha_0)| T_i,A_i \right]\right]} \nonumber\\
&= \E_{D_{ki},T_i,A_i} \bigg[ \E \bigg[ \bigg\{V_i \left[m \mathrm{I}(T_i \ge c_j) D_{ki} - \beta_{jk0} + (1-m)\mathrm{I}(T_i \ge c_j)\rho_{ki}\right] \nonumber\\
& + \: (1 - V_i)\left[\mathrm{I}(T_i \ge c_j)\rho_{ki} - \beta_{jk0} \right] \bigg\} \bigg| T_i,A_i\bigg] \bigg] \nonumber \displaybreak[1] \\
&= \E_{D_{ki},T_i,A_i} \bigg[\pi_i \left[m \mathrm{I}(T_i \ge c_j)\rho_{ki} - \beta_{jk0} + (1 - m)\mathrm{I}(T_i \ge c_j)\rho_{ki}\right] \nonumber \\
& + \: (1 - \pi_i)(\mathrm{I}(T_i \ge c_j)\rho_{ki} - \beta_{jk0})\bigg] \nonumber \\
&= \E_{D_{ki},T_i,A_i} \left[\pi_i (\mathrm{I}(T_i \ge c_j)\rho_{ki} - \beta_{jk0}) + (1 - \pi_i)(\mathrm{I}(T_i \ge c_j)\rho_{ki} - \beta_{jk0})\right] \nonumber \\
&= \E_{D_{ki},T_i,A_i} \left[\mathrm{I}(T_i \ge c_j)\rho_{ki} - \beta_{jk0}\right] = 0 \nonumber.
\end{align}
Hence, under assumption (A2), condition (C1) holds for $G_{\mathrm{IE}}^{\theta_s}(\alpha)$ and
$G_{\mathrm{IE}}^{\beta_{jk}}(\alpha)$.  

We now verify  conditions (C2) and (C3). The partial derivative of $G_{\mathrm{IE}}^{\theta_s}(\alpha)$ with respect to $\beta_{jk}$ equals $0$ for all $j,k$. Moreover,
\begin{eqnarray}
\frac{\partial}{\partial \theta_{s'}}G_{\mathrm{IE}}^{\theta_s}(\alpha) &=& \sum_{i=1}^{n} \frac{\partial}{\partial \theta_{s'}} \left\{V_i \left[m D_{si} - \theta_s + (1-m)\rho_{si}\right] + (1 - V_i)\left[\rho_{si} - \theta_s \right]\right\} \nonumber \\
&=& \sum_{i=1}^{n} \mathrm{I}(s' = s)\left\{- V_i - (1 - V_i)\right\} = -n\mathrm{I}(s' = s)\nonumber
\end{eqnarray}
and
\[
\frac{\partial}{\partial \tau_\rho}G_{\mathrm{IE}}^{\theta_s}(\alpha) = \left(\frac{\partial}{\partial \tau_{\rho_1}}G_{\mathrm{IE}}^{\theta_s}(\alpha),\frac{\partial}{\partial \tau_{\rho_2}}G_{\mathrm{IE}}^{\theta_s}(\alpha)\right)^\top.
\]
For each $l = 1,2$ and $s = 1,2$, we have
\[
\frac{\partial}{\partial \tau_{\rho_l}}G_{\mathrm{IE}}^{\theta_s}(\alpha) = \sum_{i=1}^{n}(1-m V_i)\frac{\partial}{\partial \tau_{\rho_l}}\rho_{si}.
\]
Recall that, under the multinomial logistic model,
\begin{equation}
\rho_{si} = \frac{e^{U_i^\top \tau_s}}{1 + e^{U_i^\top \tau_{\rho_1}} + e^{U_i^\top \tau_{\rho_2}}}, \qquad s = 1,2.
\label{form_rho:1}
\end{equation}
Thus, we obtain
\begin{equation}
\begin{array}{r l r l}
\frac{\partial}{\partial \tau_{\rho_1}} \rho_{1i} &= U_i \rho_{1i} (1 - \rho_{1i}), &  \frac{\partial}{\partial \tau_{\rho_2}} \rho_{1i} &= - U_i \rho_{1i}\rho_{2i}, \\ [16pt]
\frac{\partial}{\partial \tau_{\rho_2}} \rho_{2i} &= U_i \rho_{2i} (1 - \rho_{2i}), &  \frac{\partial}{\partial \tau_{\rho_1}} \rho_{2i} &= - U_i \rho_{1i}\rho_{2i}.
\end{array}
\label{deri_rho:1}
\end{equation}
The derivatives of $G_{\mathrm{IE}}^{\beta_{jk}}(\alpha)$ are 
\[
\frac{\partial}{\partial \theta_{s}} G_{\mathrm{IE}}^{\beta_{jk}}(\alpha) = 0, \qquad \frac{\partial}{\partial \beta_{j'k'}} G_{\mathrm{IE}}^{\beta_{jk}}(\alpha) = -n\mathrm{I}(j'k' = jk)
\]
and
\[
\frac{\partial}{\partial \tau_{\rho_l}}G_{\mathrm{IE}}^{\beta_{jk}}(\alpha) = \sum_{i=1}^{n}(1-m V_i)\mathrm{I}(T_i \ge c_j)\frac{\partial}{\partial \tau_{\rho_l}}\rho_{ki},
\]
where $\frac{\partial}{\partial \tau_{\rho_l}}\rho_{si}$ is in (\ref{deri_rho:1}). Hence, we  have the explicit form of the partial derivatives of both $G_{\mathrm{IE}}^{\theta_s}(\alpha)$ and $G_{\mathrm{IE}}^{\beta_{jk}}(\alpha).$
The only not null elements of the second--order partial derivative of $G_{\mathrm{IE}}^{\theta_s}(\alpha)$ and
$G_{\mathrm{IE}}^{\beta_{jk}}(\alpha)$ are those corresponding to the matrices $\frac{\partial^2}{\partial \tau \partial \tau^\top} G_{\mathrm{IE}}^{\theta_s}(\alpha)$ and  $\frac{\partial^2}{\partial \tau \partial \tau^\top} G_{\mathrm{IE}}^{\beta_{jk}}(\alpha)$. These elements involve the derivatives with respect to $\tau$ of quantities in (\ref{deri_rho:1}). It follows that conditions (C2) and (C3) hold for  $G_{\mathrm{IE}}^{\theta_s}(\alpha)$ and  $G_{\mathrm{IE}}^{\beta_{jk}}(\alpha)$ for each $s,j,k.$

%%%
\noindent
{\bf IPW estimator}.
Recall that the estimating function for $\theta_s$ is
\[
G_{\mathrm{IPW}}^{\theta_s}(\alpha) = \sum_{i=1}^{n}g_{i,\mathrm{IPW}}^{\theta_s}(\alpha) = \sum_{i=1}^{n}\frac{V_i}{\pi_i}\left( D_{si} - \theta_s\right)\quad  s = 1,2,
\]
%\label{est_equ:theta_IPW}
and for the parameter $\beta_{jk}$ is
\[
G_{\mathrm{IPW}}^{\beta_{jk}}(\alpha) = \sum_{i=1}^{n}g_{i,\mathrm{IPW}}^{\beta_{jk}}(\alpha) = \sum_{i=1}^{n}\frac{V_i}{\pi_i}\left(\mathrm{I}(T_i \ge c_j) D_{ki} - \beta_{jk}\right) \quad j = 1,2;\,\, k = 1,2,3.
\]
%\label{est_equ:beta_IPW}
We show that these estimating functions are unbiased under assumptions (A1) and (A2). In fact, we get
\begin{eqnarray}
\E \left[V_i \pi_i^{-1}\left(D_{si} - \theta_{s0}\right)\right] &=& \E_{D_s,T,A}\left[\E \left(V_i \pi_i^{-1}\left(D_{si} - \theta_{s0}\right) |,T_i,A_i\right)\right]\nonumber \\
&=& \E_{D_s,T,A}\left[\pi_i^{-1} \E \left(V_i|T_i,A_i\right)  \left(\rho_{si} - \theta_{s0}\right)\right]\nonumber \\
&=& \E_{D_s,T,A}\left[\rho_{si} - \theta_{s0}\right] = 0,\nonumber
\end{eqnarray} 
and
\begin{eqnarray}
\lefteqn{\E_{D_k,T_i,A_i}\left[\E\left[g_{i,\mathrm{IPW}}^{\beta_{jk}}(\alpha_0)| T_i,A_i \right]\right]} \nonumber \\
&=& \E_{D_{ki},T_i,A_i} \bigg[ \E \bigg[ \bigg\{\frac{V_i}{\pi_i}\left(\mathrm{I}(T_i \ge c_j) D_{ki} - \beta_{jk0}\right) \bigg\} \bigg| T_i,A_i\bigg] \bigg] \nonumber\\
&=& \E_{D_{ki},T_i,A_i} \left[\mathrm{I}(T_i \ge c_j)\rho_{ik} - \beta_{jk0}\right] = 0 \nonumber.
\end{eqnarray}
Therefore, condition (C1) holds for $G_{\mathrm{IPW}}^{\theta_s}(\alpha)$ and
$G_{\mathrm{IPW}}^{\beta_{jk}}(\alpha),$ all $s,j,k.$ 

Next, we obtain the partial derivatives
\[
\begin{array}{r l r l}
\frac{\partial}{\partial \theta_{s'}}G_{\mathrm{IPW}}^{\theta_s}(\alpha) &= -\displaystyle\sum_{i=1}^{n}\frac{V_i}{\pi_i}\mathrm{I}(s' = s), & \frac{\partial}{\partial \beta_{jk}}G_{\mathrm{IPW}}^{\theta_s}(\alpha) &= 0, \\
\frac{\partial}{\partial \theta_s}G_{\mathrm{IPW}}^{\beta_{jk}}(\alpha) &= 0, & \frac{\partial}{\partial \beta_{j'k'}}G_{\mathrm{IPW}}^{\beta_{jk}}(\alpha) &= -\displaystyle\sum_{i=1}^{n}\frac{V_i}{\pi_i}\mathrm{I}(j'k' = jk),
\end{array}
\]
and, for the logistic model (used to estimate the verification process)
\begin{eqnarray}
\frac{\partial}{\partial \tau_\pi}G_{\mathrm{IPW}}^{\theta_s}(\alpha) &=& - \sum_{i=1}^{n}\frac{V_i(D_{si} - \theta_s)U_i}{e^{U_i^\top \tau_\pi}}, \nonumber \\
\frac{\partial}{\partial \tau_\pi}G_{\mathrm{IPW}}^{\beta_{jk}}(\alpha) &=& - \sum_{i=1}^{n}\frac{V_i(\mathrm{I}(T_i \ge c_j)D_{ki} - \beta_{jk})U_i}{e^{U_i^\top \tau_\pi}}, \nonumber
\end{eqnarray}
or the probit model
\begin{eqnarray}
\frac{\partial}{\partial \tau_\pi}G_{\mathrm{IPW}}^{\theta_s}(\alpha) &=& -\displaystyle\sum_{i=1}^{n}\frac{V_i(D_{si} - \theta_s) U_i \phi(U_i^\top \tau_\pi)}{\Phi^2(U_i^\top \tau_\pi)}, \nonumber\\
\frac{\partial}{\partial \tau_\pi}G_{\mathrm{IPW}}^{\beta_{jk}}(\alpha) &=& -\displaystyle\sum_{i=1}^{n}\frac{V_i(\mathrm{I}(T_i \ge c_j)D_{ki} - \beta_{jk}) U_i \phi(U_i^\top \tau_\pi)}{\Phi^2(U_i^\top \tau_\pi)}. \nonumber
\end{eqnarray}
The computation of the second-order derivatives is similar and the results imply that the conditions (C2) and (C3) hold.

\noindent
{\bf SPE estimator}.
Recall that
\begin{eqnarray}
G_{\mathrm{SPE}}^{\theta_s}(\alpha) &=& \sum_{i=1}^{n}\left\{\frac{V_i}{\pi_i}\left(D_{si} - \theta_s\right) - \frac{V_i - \pi_i}{\pi_i}\left(\rho_{si} - \theta_s\right)\right\}, \quad  s = 1,2, \nonumber \\
G_{\mathrm{SPE}}^{\beta_{jk}}(\alpha) &=& \sum_{i=1}^{n}\left\{\frac{V_i}{\pi_i}\left[\mathrm{I}(T_i \ge c_j)D_{ki} - \beta_{jk}\right] - \frac{V_i - \pi_i}{\pi_i}\left[\mathrm{I}(T_i \ge c_j)\rho_{ki} - \beta_{jk}\right]\right\}, \nonumber
\end{eqnarray}
for $j = 1,2$ and $k = 1,2,3$. Under assumption (A1), $\E\left[g_{i,\mathrm{SPE}}^{\theta_k}(\alpha_0)\right]$ equals
\begin{eqnarray}
\lefteqn{\E_{D_s,T_i,A_i}\left[\E\left[g_{i,\mathrm{SPE}}^{\theta_s}(\alpha_0)| T_i,A_i \right]\right]} \nonumber\\
&=& \E_{D_{si},T_i,A_i} \left[\E\left[\left\{\frac{V_i}{\pi_i}\left(D_{si} - \theta_{s0}\right) - \frac{V_i - \pi_i}{\pi_i}\left(\rho_{si} - \theta_{s0}\right)\right\}| T_i,A_i \right]\right] \nonumber\\
&=& \E_{D_{si},T_i,A_i} \left[\pi^{-1}_i \left[ \E \left[D_{si} | T_i,A_i\right] - \theta_{s0}\right]\pi_i\right] \nonumber \\
&& - \: \pi_i^{-1}\E_{D_{si},T_i,A_i} \left[ \E \left[(V_i - \pi_i)(\rho_{si} - \theta_{s0})| T_i,A_i\right]\right] \nonumber \\
&=& \E_{D_{si},T_i,A_i}\left[ \rho_{si} - \theta_{s0}\right] - \pi_i^{-1}\E_{D_{si},T_i,A_i} \left[(\rho_{si} - \theta_{s0}) \E \left[(V_i - \pi_i)| T_i,A_i\right]\right] \nonumber \\
&=& \E_{D_{si},T_i,A_i} \left[\rho_{si} - \theta_{s0}\right] = 0 \nonumber.
\end{eqnarray}
and
\begin{eqnarray}
\lefteqn{\E_{D_k,T_i,A_i}\left[\E\left[g_{i,\mathrm{SPE}}^{\beta_{jk}}(\alpha_0)| T_i,A_i \right]\right]} \nonumber\\
&=& \E_{D_{ki},T_i,A_i} \bigg[ \E \bigg[ \bigg\{ \frac{V_i}{\pi_i}\left[\mathrm{I}(T_i \ge c_j)D_{ki} - \beta_{jk0}\right] \nonumber \\
&& - \: \frac{V_i - \pi_i}{\pi_i}\left[\mathrm{I}(T_i \ge c_j)\rho_{ki} - \beta_{jk0}\right] \bigg\} \bigg| T_i,A_i\bigg] \bigg] \nonumber\\
&=& \E_{D_{ki},T_i,A_i} \left[\left[\mathrm{I}(T_i \ge c_j) \E \left[D_{ki} | T_i,A_i\right] - \beta_{jk0}\right]\right] \nonumber \\
&& - \: \pi_i^{-1}\E_{D_{ki},T_i,A_i} \left[ \E \left[(V_i - \pi_i)(\mathrm{I}(T_i \ge c_j)\rho_{ki} - \beta_{jk0})| T_i,A_i\right]\right] \nonumber \\
&=& \E_{D_{ki},T_i,A_i} \left[\mathrm{I}(T_i \ge c_j)\rho_{ki} - \beta_{jk0}\right] = 0 \nonumber.
\end{eqnarray}
Therefore, condition (C1) holds for $G_{\mathrm{SPE}}^{\theta_s}(\alpha)$ and
$G_{\mathrm{SPE}}^{\beta_{jk}}(\alpha),$ all $s,j,k.$ 

Next, we obtain the partial derivatives
\[
\begin{array}{r l r l}
\frac{\partial}{\partial \theta_{s'}}G_{\mathrm{SPE}}^{\theta_s}(\alpha) &= -n\mathrm{I}(s' = s) & \frac{\partial}{\partial \beta_{jk}}G_{\mathrm{SPE}}^{\theta_s}(\alpha) &= 0 \\
\frac{\partial}{\partial \theta_s}G_{\mathrm{SPE}}^{\beta_{jk}}(\alpha) &= 0 & \frac{\partial}{\partial \beta_{j'k'}}g_{\mathrm{SPE}}^{\beta_{jk}}(\alpha) &= -n\mathrm{I}(j'k' = jk)
\end{array}
\]
and the partial derivative with respect to $\tau_{\rho} \equiv (\tau_{\rho_1},\tau_{\rho_2})$
\begin{eqnarray}
\frac{\partial}{\partial \tau_{\rho_l}}G_{\mathrm{SPE}}^{\theta_s}(\alpha) &=& \sum_{i=1}^{n}-\frac{V_i - \pi_i}{\pi_i}\frac{\partial}{\partial \tau_{\rho_l}}\rho_{si}, \nonumber \\
\frac{\partial}{\partial \tau_{\rho_l}}G_{\mathrm{SPE}}^{\beta_{jk}}(\alpha) &=& \sum_{i=1}^{n}-\frac{V_i - \pi_i}{\pi_i}\mathrm{I}(T_i \ge c_j)\frac{\partial}{\partial \tau_{\rho_l}}\rho_{si}, \nonumber
\end{eqnarray}
where $\frac{\partial}{\partial \tau_{\rho_l}}\rho_{si}$ is given in (\ref{deri_rho:1}). The partial derivative with respect to $\tau_\pi$, are
\begin{eqnarray}
\frac{\partial}{\partial \tau_\pi}G_{\mathrm{SPE}}^{\theta_s}(\alpha) &=& \sum_{i=1}^{n}\frac{V_iU_i(\rho_{si} - D_{si})}{e^{U_i^\top \tau_\pi}}, \nonumber \\
\frac{\partial}{\partial \tau_\pi}G_{\mathrm{SPE}}^{\beta_{jk}}(\alpha) &=& \sum_{i=1}^{n}\frac{V_iU_i\mathrm{I}(T_i \ge c_j)(\rho_{ki} - D_{ki})}{e^{U_i^\top \tau}}, \nonumber
\end{eqnarray}
when the logistic model is used for the verification process. If the
probit model is used, we have
\begin{eqnarray}
\frac{\partial}{\partial \tau_\pi}G_{\mathrm{SPE}}^{\theta_s}(\alpha) &=& \sum_{i=1}^{n}\frac{V_iU_i(D_{si} - \rho_{si})\phi(U_i^\top \tau_\pi)}{\Phi^2(U_i^\top \tau_\pi)}, \nonumber\\
\frac{\partial}{\partial \tau_\pi}G_{\mathrm{SPE}}^{\beta_{jk}}(\alpha) &=& \sum_{i=1}^{n}\frac{V_iU_i\mathrm{I}(T_i \ge c_j)(D_{si} - \rho_{ki})\phi(U_i^\top \tau_\pi)}{\Phi^2(U_i^\top \tau_\pi)}. \nonumber
\end{eqnarray}
Also in this case, computation of the second--order partial derivatives develops similarly and  the results imply that the conditions (C2) and (C3) hold.

\noindent
{\bf Asymptotic covariance matrix}.
Recall that the  asymptotic covariance matrix of TCF estimators is obtained as
\[
\frac{\partial h({\alpha_0})}{\partial \alpha} {\Sigma} \frac{\partial h^\top(\alpha_0)}{\partial \alpha},
\]
where $h(\alpha) = \left(1 - \frac{\beta_{11}}{\theta_1}, \frac{\beta_{12} - \beta_{22}}{\theta_2}, \frac{\beta_{23}}{1 - (\theta_1 + \theta_2)}\right)^\top$ and 
\[
\Sigma = \left[\E \left\{\frac{\partial}{\partial \alpha}g_{i,*}(\alpha_0)\right\}\right]^{-1} \E\{g_{i,*}(\alpha_0)g_{i,*}(\alpha_0)^\top\}\left[\E \left\{\frac{\partial}{\partial \alpha}g_{i,*}^\top(\alpha_0)\right\}\right]^{-1}.
\]
It is easy to derive that
\[
\frac{\partial h({\alpha})}{\partial \alpha} = \left(\begin{array}{r r r r r r r}
\frac{ {\beta}_{11}}{{\theta}_1^2} & 0 & -\frac{1}{{\theta}_1} & 0 & 0 & 0 & 0\\
0 & -\frac{{\beta}_{12} -  {\beta}_{22}}{{\theta}_2^2} & 0 & \frac{1}{{\theta}_2} & -\frac{1}{{\theta}_2} & 0 & 0 \\
\frac{{\beta}_{23}}{(1 -  {\theta}_1 -  {\theta}_2)^2} & \frac{{\beta}_{23}}{(1 -  {\theta}_1 -  {\theta}_2)^2} & 0 & 0 & 0  & \frac{1}{1 -  {\theta}_1 -  {\theta}_2} & 0
\end{array}
\right).
\]
The elements $g_{i,*}({\alpha})$ of the estimating functions $G_{*}({\alpha})$ are given in the previous paragraphs. 
Now, we derive the explicit form for $\frac{\partial}{\partial \alpha}g_{i,*}({\alpha})$.

First, we consider the class of imputation estimators. We get
\begin{eqnarray}
\frac{\partial}{\partial \alpha}g_{i,\mathrm{IE}}^{\theta_1}(\alpha) &=& \left(-1,0,0,0,0,0,A_{11i},A_{21i}\right), \nonumber\\
\frac{\partial}{\partial \alpha}g_{i,\mathrm{IE}}^{\theta_2}(\alpha) &=& \left(0,-1,0,0,0,0,A_{12i},A_{22i}\right), \nonumber\\
\frac{\partial}{\partial \alpha}g_{i,\mathrm{IE}}^{\beta_{11}}(\alpha) &=& \left(0,0,-1,0,0,0,B_{111i},B_{121i}\right), \nonumber\\
\frac{\partial}{\partial \alpha}g_{i,\mathrm{IE}}^{\beta_{12}}(\alpha) &=& \left(0,0,0,-1,0,0,B_{112i},B_{122i}\right), \nonumber\\
\frac{\partial}{\partial \alpha}g_{i,\mathrm{IE}}^{\beta_{22}}(\alpha) &=& \left(0,0,0,0,-1,0,B_{212i},B_{222i}\right), \nonumber\\
\frac{\partial}{\partial \alpha}g_{i,\mathrm{IE}}^{\beta_{23}}(\alpha) &=& \left(0,0,0,0,0,-1, B_{213i}, B_{223i}\right), \nonumber\\
\frac{\partial}{\partial \alpha}g_{i,\mathrm{IE}}^{\tau_{\rho_1}}(\alpha) &=& \left(0,0,0,0,0,0,C_{11i},C_{21i}\right), \nonumber\\
\frac{\partial}{\partial \alpha}g_{i,\mathrm{IE}}^{\tau_{\rho_2}}(\alpha) &=& \left(0,0,0,0,0,0,C_{12i},C_{22i}\right), \nonumber
\end{eqnarray}
where 
\[
A_{lsi} = (1-mV_i)\frac{\partial}{\partial \tau_{\rho_l}}\rho_{si}, \quad B_{jlki} = (1 - mV_i)\mathrm{I}(T_i \ge c_j)\frac{\partial}{\partial \tau_{\rho_l}}\rho_{ki}, \quad C_{lsi} = \frac{\partial}{\partial \tau_{\rho_l}}g_{i}^{\tau_{\rho_s}}(\alpha),
\]
with $j,l,s = 1,2$, $k = 1,2,3$ and $i = 1,\ldots,n$ (see (\ref{deri_tau:1}) and (\ref{deri_rho:1}) for the multinomial logistic modeling of 
the disease process). 

Thus,
\begin{eqnarray}
\frac{\partial}{\partial \alpha}g_{i,\mathrm{IE}}(\alpha) = \left(\begin{array}{r r r r r r r r}
-1 & 0 & 0 & 0 & 0 & 0 & A_{11i} & A_{21i} \\
0 & -1 & 0 & 0 & 0 & 0 & A_{12i} & A_{22i} \\
0 & 0 & -1 & 0 & 0 & 0 & B_{111i} & B_{121i} \\
0 & 0 & 0 & -1 & 0 & 0 & B_{112i} & B_{122i} \\
0 & 0 & 0 & 0 & -1 & 0 & B_{212i} & B_{222i} \\
0 & 0 & 0 & 0 & 0 & -1 & B_{213i} & B_{223i} \\
0 & 0 & 0 & 0 & 0 & 0 & C_{11i} & C_{21i} \\
0 & 0 & 0 & 0 & 0 & 0 & C_{12i} & C_{22i}
\end{array}\right)\nonumber.
\end{eqnarray}
Then, we consider the inverse probability weighted estimators. Let
\[
A_{ki} = \frac{\partial}{\partial \tau_\pi}g_{i,\mathrm{IPW}}^{\theta_k}(\alpha), \qquad B_{jki} = \frac{\partial}{\partial \tau_\pi}g_{i,\mathrm{IPW}}^{\beta_{jk}}(\alpha), \qquad C_i = \frac{\partial}{\partial \tau_\pi} g_i^{\tau_\pi}(\alpha).
\]
Note that these quantities change according to the model, logit or probit, chosen for the
verification process. We obtain
\begin{align}
\frac{\partial}{\partial \alpha}g_{i,\mathrm{IPW}}^{\theta_1}(\alpha) &= \left(-\frac{V_i}{\pi_i},0,0,0,0,0,A_{1i}\right), \nonumber  \\
\frac{\partial}{\partial \alpha}g_{i,\mathrm{IPW}}^{\theta_2}(\alpha) &= \left(0,-\frac{V_i}{\pi_i},0,0,0,0,A_{2i}\right), \nonumber \\
\frac{\partial}{\partial \alpha}g_{i,\mathrm{IPW}}^{\beta_{11}}(\alpha) &= \left(0,0,-\frac{V_i}{\pi_i},0,0,0,B_{11i}\right), \nonumber \displaybreak[1]\\
\frac{\partial}{\partial \alpha}g_{i,\mathrm{IPW}}^{\beta_{12}}(\alpha) &= \left(0,0,0,-\frac{V_i}{\pi_i},0,0,B_{12i}\right), \nonumber \\
\frac{\partial}{\partial \alpha}g_{i,\mathrm{IPW}}^{\beta_{22}}(\alpha) &= \left(0,0,0,0,-\frac{V_i}{\pi_i},0,B_{22i}\right), \nonumber \\
\frac{\partial}{\partial \alpha}g_{i,\mathrm{IPW}}^{\beta_{23}}(\alpha) &= \left(0,0,0,0,0,-\frac{V_i}{\pi_i},B_{23i}\right), \nonumber\\
\frac{\partial}{\partial \alpha}g_{i}^{\tau_\pi}(\alpha) &= \left(0,0,0,0,0,0,C_i \right), \nonumber
\end{align}
Summarizing
\begin{eqnarray}
\frac{\partial}{\partial \alpha}g_{i,\mathrm{IPW}}(\alpha) = \left(\begin{array}{r r r r r r r}
-\frac{V_i}{\pi_i} & 0 & 0 & 0 & 0 & 0 & A_{1i} \\
0 & -\frac{V_i}{\pi_i} & 0 & 0 & 0 & 0 & A_{2i} \\
0 & 0 & -\frac{V_i}{\pi_i} & 0 & 0 & 0 & B_{11i} \\
0 & 0 & 0 & -\frac{V_i}{\pi_i} & 0 & 0 & B_{12i} \\
0 & 0 & 0 & 0 & -\frac{V_i}{\pi_i} & 0 & B_{22i} \\
0 & 0 & 0 & 0 & 0 & -\frac{V_i}{\pi_i} & B_{23i} \\
0 & 0 & 0 & 0 & 0 & 0 & C_i
\end{array}\right)\nonumber.
\end{eqnarray}
Finally, we consider the SPE estimators. We have
\begin{eqnarray}
\frac{\partial}{\partial \alpha}g_{i,\mathrm{SPE}}^{\theta_1}(\alpha) &=& \left(-1,0,0,0,0,0,H_{11i},H_{21i},D_{1i}\right), \nonumber\\
\frac{\partial}{\partial \alpha}g_{i,\mathrm{SPE}}^{\theta_2}(\alpha) &=& \left(0,-1,0,0,0,0,H_{12i},H_{22i},D_{2i}\right), \nonumber\\
\frac{\partial}{\partial \alpha}g_{i,\mathrm{SPE}}^{\beta_{11}}(\alpha) &=& \left(0,0,-1,0,0,0,G_{111i},G_{121i},E_{11i}\right), \nonumber\\
\frac{\partial}{\partial \alpha}g_{i,\mathrm{SPE}}^{\beta_{12}}(\alpha) &=& \left(0,0,0,-1,0,0,G_{112i},G_{122i},E_{12i}\right), \nonumber\\
\frac{\partial}{\partial \alpha}g_{i,\mathrm{SPE}}^{\beta_{22}}(\alpha) &=& \left(0,0,0,0,-1,0,G_{212i},G_{222i},E_{22i}\right), \nonumber\\
\frac{\partial}{\partial \alpha}g_{i,\mathrm{SPE}}^{\beta_{23}}(\alpha) &=& \left(0,0,0,0,0,-1, G_{213i}, G_{223i}, E_{23i}\right), \nonumber\\
\frac{\partial}{\partial \alpha}g_{i,\mathrm{SPE}}^{\tau_{\rho_1}}(\alpha) &=& \left(0,0,0,0,0,0,C_{11i},C_{21i},0\right), \nonumber\\
\frac{\partial}{\partial \alpha}g_{i,\mathrm{SPE}}^{\tau_{\rho_2}}(\alpha) &=& \left(0,0,0,0,0,0,C_{12i},C_{22i},0\right), \nonumber \\
\frac{\partial}{\partial \alpha}g_{i,\mathrm{SPE}}^{\tau_\pi}(\alpha) &=& \left(0,0,0,0,0,0,0,0,C_i\right), \nonumber
\end{eqnarray}
where
\[
\begin{array}{r l r l}
H_{lki} &= -\dfrac{V_i - \pi_i}{\pi_i}\dfrac{\partial}{\partial \tau_{\rho_l}}\rho_{ki}, & G_{jlki} &= -\frac{V_i - \pi_i}{\pi_i}\mathrm{I}(T_i \ge c_j)\frac{\partial}{\partial \tau_{\rho_l}}\rho_{ki}, \\
D_{si} &= \dfrac{\partial}{\partial \tau_\pi}g_{i,\mathrm{SPE}}^{\theta_s}(\alpha), & E_{jki} &= \dfrac{\partial}{\partial \tau_\pi}g_{i,\mathrm{SPE}}^{\beta_{jk}}(\alpha),
\end{array}
\]
and $C_{lsi}$ and $C_i$ are defined above. Therefore
\begin{eqnarray}
\frac{\partial}{\partial \alpha}g_{i,\mathrm{SPE}}(\alpha) = \left(\begin{array}{r r r r r r r r r}
-1 & 0 & 0 & 0 & 0 & 0 & H_{11i} & H_{21i} & D_{1i} \\
0 & -1 & 0 & 0 & 0 & 0 & H_{12i} & H_{22i} & D_{2i}\\
0 & 0 & -1 & 0 & 0 & 0 & G_{111i} & G_{121i} & E_{11i}\\
0 & 0 & 0 & -1 & 0 & 0 & G_{112i} & G_{122i} & E_{12i}\\
0 & 0 & 0 & 0 & -1 & 0 & G_{212i} & G_{222i} & E_{22i}\\
0 & 0 & 0 & 0 & 0 & -1 & G_{213i} & G_{223i} & E_{23i}\\
0 & 0 & 0 & 0 & 0 & 0 & C_{11i} & C_{21i} & 0 \\
0 & 0 & 0 & 0 & 0 & 0 & C_{12i} & C_{22i} & 0\\
0 & 0 & 0 & 0 & 0 & 0 & 0 & 0 & C_i
\end{array}\right)\nonumber.
\end{eqnarray}

\section{Simulation results of Study 1 and Study 2}\label{app:simu:roc}
In this section, we present  results of simulations in Study 1 and Study 2. Tables \ref{scen1:1:500}--\ref{scen1:3:1000} show simulation results of Study 1 for sample sizes equal to $500$ and $1000$, respectively. The results of Study 2 are presented in Tables \ref{scen2:1:1000} and \ref{scen2:3:1000}, corresponding to the first and third value of $\Lambda,$ respectively.

%% first covariance matrix, n = 500
\begin{sidewaystable}
\begin{center}
\caption{Simulation results from 5000 replications when both models for $\rho_{k}$ and $\pi$ are correctly specified (Study 1) and the first value of $\Lambda$ is considered. ``True'' denotes the true parameter value. Sample size = 500.}
\label{scen1:1:500}
\begin{scriptsize}
\begin{tabular}{r c c c c c c c c c c c c}
  \hline
  & TCF$_1$ & TCF$_2$ & TCF$_3$ & MC.sd$_1$ & MC.sd$_2$ & MC.sd$_3$ & asy.sd$_1$ & asy.sd$_2$ & asy.sd$_3$ & boot.sd$_1$ & boot.sd$_2$ & boot.sd$_3$ \\ 
   \hline 
   & \multicolumn{12}{c}{cut-point = (2,4)} \\
True & 0.5000 & 0.4347 & 0.9347 &  &  &  &  &  &  &  &  &  \\ 
  FI & 0.5007 & 0.4357 & 0.9343 & 0.0373 & 0.0339 & 0.0190 & 0.0343 & 0.0309 & 0.0360 & 0.0372 & 0.0340 & 0.0193 \\ 
  MSI & 0.5008 & 0.4357 & 0.9343 & 0.0382 & 0.0378 & 0.0225 & 0.0353 & 0.0354 & 0.0380 & 0.0381 & 0.0380 & 0.0227 \\ 
  IPW & 0.5020 & 0.4357 & 0.9345 & 0.0506 & 0.0508 & 0.0260 & 0.0493 & 0.0502 & 0.0280 & 0.0499 & 0.0507 & 0.0262 \\ 
  SPE & 0.5012 & 0.4357 & 0.9345 & 0.0401 & 0.0456 & 0.0256 & 0.0399 & 0.0450 & 0.0249 & 0.0409 & 0.0457 & 0.0259 \\ 
  \hline 
     & \multicolumn{12}{c}{cut-point = (2,5)} \\
  True  & 0.5000 & 0.7099 & 0.7752 &  &  &  &  &  &  &  &  &  \\ 
  FI & 0.5007 & 0.7115 & 0.7747 & 0.0373 & 0.0329 & 0.0372 & 0.0343 & 0.0321 & 0.0436 & 0.0372 & 0.0329 & 0.0374 \\ 
  MSI & 0.5008 & 0.7111 & 0.7743 & 0.0382 & 0.0361 & 0.0395 & 0.0353 & 0.0355 & 0.0455 & 0.0381 & 0.0362 & 0.0396 \\ 
  IPW & 0.5020 & 0.7111 & 0.7743 & 0.0506 & 0.0496 & 0.0464 & 0.0493 & 0.0479 & 0.0500 & 0.0499 & 0.0487 & 0.0463 \\ 
  SPE & 0.5012 & 0.7112 & 0.7744 & 0.0401 & 0.0434 & 0.0442 & 0.0399 & 0.0425 & 0.0433 & 0.0409 & 0.0436 & 0.0448 \\ 
  \hline 
     & \multicolumn{12}{c}{cut-point = (2,7)} \\
  True & 0.5000 & 0.9230 & 0.2248 &  &  &  &  &  &  &  &  &  \\ 
  FI & 0.5007 & 0.9228 & 0.2241 & 0.0373 & 0.0167 & 0.0377 & 0.0343 & 0.0229 & 0.0310 & 0.0372 & 0.0169 & 0.0370 \\ 
  MSI & 0.5008 & 0.9230 & 0.2242 & 0.0382 & 0.0199 & 0.0382 & 0.0353 & 0.0253 & 0.0317 & 0.0381 & 0.0202 & 0.0376 \\ 
  IPW & 0.5020 & 0.9232 & 0.2242 & 0.0506 & 0.0266 & 0.0534 & 0.0493 & 0.0251 & 0.0520 & 0.0499 & 0.0263 & 0.0525 \\ 
  SPE & 0.5012 & 0.9235 & 0.2245 & 0.0401 & 0.0255 & 0.0416 & 0.0399 & 0.0244 & 0.0403 & 0.0409 & 0.0253 & 0.0545 \\ 
  \hline 
     & \multicolumn{12}{c}{cut-point = (4,5)} \\
  True & 0.9347 & 0.2752 & 0.7752 &  &  &  &  &  &  &  &  &  \\ 
  FI & 0.9349 & 0.2758 & 0.7747 & 0.0176 & 0.0285 & 0.0372 & 0.0161 & 0.0246 & 0.0436 & 0.0174 & 0.0289 & 0.0374 \\ 
  MSI & 0.9348 & 0.2754 & 0.7743 & 0.0194 & 0.0326 & 0.0395 & 0.0179 & 0.0291 & 0.0455 & 0.0191 & 0.0328 & 0.0396 \\ 
  IPW & 0.9352 & 0.2754 & 0.7743 & 0.0299 & 0.0472 & 0.0464 & 0.0263 & 0.0466 & 0.0500 & 0.0284 & 0.0472 & 0.0463 \\ 
  SPE & 0.9353 & 0.2755 & 0.7744 & 0.0270 & 0.0407 & 0.0442 & 0.0249 & 0.0399 & 0.0433 & 0.0291 & 0.0404 & 0.0448 \\ 
  \hline 
     & \multicolumn{12}{c}{cut-point = (4,7)} \\
  True & 0.9347 & 0.4883 & 0.2248 &  &  &  &  &  &  &  &  &  \\ 
  FI & 0.9349 & 0.4872 & 0.2241 & 0.0176 & 0.0375 & 0.0377 & 0.0161 & 0.0347 & 0.0310 & 0.0174 & 0.0375 & 0.0370 \\ 
  MSI & 0.9348 & 0.4872 & 0.2242 & 0.0194 & 0.0396 & 0.0382 & 0.0179 & 0.0373 & 0.0317 & 0.0191 & 0.0398 & 0.0376 \\ 
  IPW & 0.9352 & 0.4876 & 0.2242 & 0.0299 & 0.0520 & 0.0534 & 0.0263 & 0.0511 & 0.0520 & 0.0284 & 0.0516 & 0.0525 \\ 
  SPE & 0.9353 & 0.4877 & 0.2245 & 0.0270 & 0.0462 & 0.0416 & 0.0249 & 0.0456 & 0.0403 & 0.0291 & 0.0463 & 0.0545 \\ 
  \hline 
     & \multicolumn{12}{c}{cut-point = (5,7)} \\
  True & 0.9883 & 0.2132 & 0.2248 &  &  &  &  &  &  &  &  &  \\ 
  FI & 0.9882 & 0.2114 & 0.2241 & 0.0051 & 0.0310 & 0.0377 & 0.0047 & 0.0274 & 0.0310 & 0.0054 & 0.0306 & 0.0370 \\ 
  MSI & 0.9882 & 0.2118 & 0.2242 & 0.0069 & 0.0330 & 0.0382 & 0.0058 & 0.0299 & 0.0317 & 0.0069 & 0.0329 & 0.0376 \\ 
  IPW & 0.9886 & 0.2121 & 0.2242 & 0.0137 & 0.0467 & 0.0534 & 0.0088 & 0.0441 & 0.0520 & 0.0130 & 0.0452 & 0.0525 \\ 
  SPE & 0.9886 & 0.2123 & 0.2245 & 0.0133 & 0.0398 & 0.0416 & 0.0097 & 0.0387 & 0.0403 & 0.0127 & 0.0398 & 0.0545 \\ 
   \hline
\end{tabular}
\end{scriptsize}
\end{center}
\end{sidewaystable}

%% the first corvariance matrix, n = 1000
\begin{sidewaystable}
\begin{center}
\caption{Simulation results from 5000 replications when both models for $\rho_{k}$ and $\pi$ are correctly specified (Study 1) and the first value of $\Lambda$ is considered. ``True'' denotes the true parameter value. Sample size = 1000.}
\label{scen1:1:1000}
\begin{scriptsize}
\begin{tabular}{r c c c c c c c c c c c c}
  \hline
 & TCF$_1$ & TCF$_2$ & TCF$_3$ & MC.sd$_1$ & MC.sd$_2$ & MC.sd$_3$ & asy.sd$_1$ & asy.sd$_2$ & asy.sd$_3$ & boot.sd$_1$ & boot.sd$_2$ & boot.sd$_3$ \\ 
 \hline 
    & \multicolumn{12}{c}{cut-point = (2,4)} \\
True & 0.5000 & 0.4347 & 0.9347 &  &  &  &  &  &  &  &  &  \\ 
  FI & 0.5001 & 0.4346 & 0.9348 & 0.0265 & 0.0235 & 0.0133 & 0.0242 & 0.0217 & 0.0254 & 0.0262 & 0.0238 & 0.0135 \\ 
  MSI & 0.5002 & 0.4349 & 0.9349 & 0.0273 & 0.0264 & 0.0157 & 0.0250 & 0.0250 & 0.0268 & 0.0269 & 0.0268 & 0.0160 \\ 
  IPW & 0.5006 & 0.4357 & 0.9349 & 0.0362 & 0.0357 & 0.0184 & 0.0352 & 0.0358 & 0.0198 & 0.0353 & 0.0360 & 0.0185 \\ 
  SPE & 0.5004 & 0.4353 & 0.9349 & 0.0287 & 0.0321 & 0.0180 & 0.0282 & 0.0320 & 0.0180 & 0.0283 & 0.0322 & 0.0182 \\ 
  \hline 
      & \multicolumn{12}{c}{cut-point = (2,5)} \\
  True & 0.5000 & 0.7099 & 0.7752 &  &  &  &  &  &  &  &  &  \\ 
  FI & 0.5001 & 0.7096 & 0.7758 & 0.0265 & 0.0232 & 0.0260 & 0.0242 & 0.0227 & 0.0308 & 0.0262 & 0.0232 & 0.0263 \\ 
  MSI & 0.5002 & 0.7095 & 0.7756 & 0.0273 & 0.0256 & 0.0276 & 0.0250 & 0.0251 & 0.0321 & 0.0269 & 0.0256 & 0.0279 \\ 
  IPW & 0.5006 & 0.7104 & 0.7756 & 0.0362 & 0.0349 & 0.0325 & 0.0352 & 0.0342 & 0.0354 & 0.0353 & 0.0345 & 0.0327 \\ 
  SPE & 0.5004 & 0.7100 & 0.7757 & 0.0287 & 0.0309 & 0.0307 & 0.0282 & 0.0303 & 0.0308 & 0.0283 & 0.0305 & 0.0310 \\ 
  \hline 
      & \multicolumn{12}{c}{cut-point = (2,7)} \\
  True & 0.5000 & 0.9230 & 0.2248 &  &  &  &  &  &  &  &  &  \\ 
  FI & 0.5001 & 0.9228 & 0.2250 & 0.0265 & 0.0117 & 0.0260 & 0.0242 & 0.0160 & 0.0220 & 0.0262 & 0.0119 & 0.0262 \\ 
  MSI & 0.5002 & 0.9230 & 0.2252 & 0.0273 & 0.0141 & 0.0265 & 0.0250 & 0.0178 & 0.0226 & 0.0269 & 0.0142 & 0.0266 \\ 
  IPW & 0.5006 & 0.9233 & 0.2258 & 0.0362 & 0.0187 & 0.0383 & 0.0352 & 0.0181 & 0.0374 & 0.0353 & 0.0186 & 0.0375 \\ 
  SPE & 0.5004 & 0.9235 & 0.2256 & 0.0287 & 0.0180 & 0.0286 & 0.0282 & 0.0176 & 0.0286 & 0.0283 & 0.0180 & 0.0291 \\ 
  \hline 
      & \multicolumn{12}{c}{cut-point = (4,5)} \\
  True & 0.9347 & 0.2752 & 0.7752 &  &  &  &  &  &  &  &  &  \\ 
  FI & 0.9346 & 0.2749 & 0.7758 & 0.0124 & 0.0203 & 0.0260 & 0.0115 & 0.0173 & 0.0308 & 0.0123 & 0.0203 & 0.0263 \\ 
  MSI & 0.9345 & 0.2746 & 0.7756 & 0.0137 & 0.0232 & 0.0276 & 0.0128 & 0.0205 & 0.0321 & 0.0136 & 0.0231 & 0.0279 \\ 
  IPW & 0.9346 & 0.2748 & 0.7756 & 0.0213 & 0.0337 & 0.0325 & 0.0196 & 0.0332 & 0.0354 & 0.0205 & 0.0335 & 0.0327 \\ 
  SPE & 0.9344 & 0.2747 & 0.7757 & 0.0190 & 0.0286 & 0.0307 & 0.0183 & 0.0283 & 0.0308 & 0.0187 & 0.0285 & 0.0310 \\  
  \hline 
      & \multicolumn{12}{c}{cut-point = (4,7)} \\
  True & 0.9347 & 0.4883 & 0.2248 &  &  &  &  &  &  &  &  &  \\ 
  FI & 0.9346 & 0.4882 & 0.2250 & 0.0124 & 0.0262 & 0.0260 & 0.0115 & 0.0245 & 0.0220 & 0.0123 & 0.0264 & 0.0262 \\ 
  MSI & 0.9345 & 0.4881 & 0.2252 & 0.0137 & 0.0279 & 0.0265 & 0.0128 & 0.0263 & 0.0226 & 0.0136 & 0.0280 & 0.0266 \\ 
  IPW & 0.9346 & 0.4876 & 0.2258 & 0.0213 & 0.0365 & 0.0383 & 0.0196 & 0.0364 & 0.0374 & 0.0205 & 0.0366 & 0.0375 \\ 
  SPE & 0.9344 & 0.4882 & 0.2256 & 0.0190 & 0.0325 & 0.0286 & 0.0183 & 0.0324 & 0.0286 & 0.0187 & 0.0326 & 0.0291 \\ 
  \hline 
      & \multicolumn{12}{c}{cut-point = (5,7)} \\
  True & 0.9883 & 0.2132 & 0.2248 &  &  &  &  &  &  &  &  &  \\ 
  FI & 0.9881 & 0.2132 & 0.2250 & 0.0036 & 0.0217 & 0.0260 & 0.0033 & 0.0194 & 0.0220 & 0.0037 & 0.0216 & 0.0262 \\ 
  MSI & 0.9881 & 0.2135 & 0.2252 & 0.0048 & 0.0234 & 0.0265 & 0.0044 & 0.0212 & 0.0226 & 0.0049 & 0.0232 & 0.0266 \\ 
  IPW & 0.9882 & 0.2129 & 0.2258 & 0.0100 & 0.0325 & 0.0383 & 0.0077 & 0.0317 & 0.0374 & 0.0097 & 0.0320 & 0.0375 \\ 
  SPE & 0.9880 & 0.2135 & 0.2256 & 0.0097 & 0.0282 & 0.0286 & 0.0080 & 0.0276 & 0.0286 & 0.0094 & 0.0278 & 0.0291 \\ 
   \hline
\end{tabular}
\end{scriptsize}
\end{center}
\end{sidewaystable}

%% the second covariance matrix, n = 500
\begin{sidewaystable}
\begin{center}
\caption{Simulation results from 5000 replications when both models for $\rho_{k}$ and $\pi$ are correctly specified (Study 1) and the second value of $\Lambda$ is considered. ``True'' denotes the true parameter value. Sample size = 500.}
\label{scen1:2:500}
\begin{scriptsize}
\begin{tabular}{r c c c c c c c c c c c c}
  \hline
 & TCF$_1$ & TCF$_2$ & TCF$_3$ & MC.sd$_1$ & MC.sd$_2$ & MC.sd$_3$ & asy.sd$_1$ & asy.sd$_2$ & asy.sd$_3$ & boot.sd$_1$ & boot.sd$_2$ & boot.sd$_3$ \\ 
 \hline 
       & \multicolumn{12}{c}{cut-point = (2,4)} \\
True & 0.5000 & 0.3970 & 0.8970 &  &  &  &  &  &  &  &  &  \\ 
  FI & 0.4999 & 0.3974 & 0.8965 & 0.0356 & 0.0294 & 0.0253 & 0.0326 & 0.0263 & 0.0355 & 0.0355 & 0.0298 & 0.0256 \\ 
  MSI & 0.4999 & 0.3975 & 0.8961 & 0.0368 & 0.0355 & 0.0291 & 0.0339 & 0.0326 & 0.0380 & 0.0367 & 0.0355 & 0.0291 \\ 
  IPW & 0.5000 & 0.3977 & 0.8962 & 0.0470 & 0.0492 & 0.0373 & 0.0460 & 0.0484 & 0.0369 & 0.0464 & 0.0487 & 0.0368 \\ 
  SPE & 0.5000 & 0.3976 & 0.8963 & 0.0402 & 0.0446 & 0.0363 & 0.0397 & 0.0438 & 0.0348 & 0.0400 & 0.0442 & 0.0356 \\ 
 \hline 
       & \multicolumn{12}{c}{cut-point = (2,5)} \\  
  True & 0.5000 & 0.6335 & 0.7365 &  &  &  &  &  &  &  &  &  \\ 
  FI & 0.4999 & 0.6342 & 0.7360 & 0.0356 & 0.0303 & 0.0410 & 0.0326 & 0.0287 & 0.0439 & 0.0355 & 0.0306 & 0.0409 \\ 
  MSI & 0.4999 & 0.6339 & 0.7358 & 0.0368 & 0.0356 & 0.0437 & 0.0339 & 0.0342 & 0.0463 & 0.0367 & 0.0357 & 0.0435 \\ 
  IPW & 0.5000 & 0.6336 & 0.7363 & 0.0470 & 0.0477 & 0.0528 & 0.0460 & 0.0471 & 0.0529 & 0.0464 & 0.0474 & 0.0514 \\ 
  SPE & 0.5000 & 0.6341 & 0.7362 & 0.0402 & 0.0440 & 0.0494 & 0.0397 & 0.0434 & 0.0479 & 0.0400 & 0.0437 & 0.0483 \\ 
  \hline 
        & \multicolumn{12}{c}{cut-point = (2,7)} \\
  True & 0.5000 & 0.8682 & 0.2635 &  &  &  &  &  &  &  &  &  \\ 
  FI & 0.4999 & 0.8677 & 0.2631 & 0.0356 & 0.0222 & 0.0388 & 0.0326 & 0.0233 & 0.0352 & 0.0355 & 0.0219 & 0.0391 \\ 
  MSI & 0.4999 & 0.8678 & 0.2633 & 0.0368 & 0.0263 & 0.0401 & 0.0339 & 0.0272 & 0.0370 & 0.0367 & 0.0261 & 0.0407 \\ 
  IPW & 0.5000 & 0.8677 & 0.2638 & 0.0470 & 0.0354 & 0.0477 & 0.0460 & 0.0341 & 0.0486 & 0.0464 & 0.0349 & 0.0484 \\ 
  SPE & 0.5000 & 0.8679 & 0.2635 & 0.0402 & 0.0336 & 0.0420 & 0.0397 & 0.0326 & 0.0420 & 0.0400 & 0.0331 & 0.0424 \\ 
 \hline 
       & \multicolumn{12}{c}{cut-point = (4,5)} \\  
  True & 0.8970 & 0.2365 & 0.7365 &  &  &  &  &  &  &  &  &  \\ 
  FI & 0.8972 & 0.2368 & 0.7360 & 0.0205 & 0.0257 & 0.0410 & 0.0195 & 0.0223 & 0.0439 & 0.0203 & 0.0258 & 0.0409 \\ 
  MSI & 0.8968 & 0.2364 & 0.7358 & 0.0229 & 0.0310 & 0.0437 & 0.0219 & 0.0279 & 0.0463 & 0.0226 & 0.0308 & 0.0435 \\ 
  IPW & 0.8969 & 0.2359 & 0.7363 & 0.0268 & 0.0421 & 0.0528 & 0.0261 & 0.0411 & 0.0529 & 0.0265 & 0.0415 & 0.0514 \\ 
  SPE & 0.8967 & 0.2365 & 0.7362 & 0.0260 & 0.0374 & 0.0494 & 0.0254 & 0.0370 & 0.0479 & 0.0257 & 0.0373 & 0.0483 \\ 
 \hline 
       & \multicolumn{12}{c}{cut-point = (4,7)} \\  
  True & 0.8970 & 0.4711 & 0.2635 &  &  &  &  &  &  &  &  &  \\ 
  FI & 0.8972 & 0.4703 & 0.2631 & 0.0205 & 0.0356 & 0.0388 & 0.0195 & 0.0328 & 0.0352 & 0.0203 & 0.0356 & 0.0391 \\ 
  MSI & 0.8968 & 0.4703 & 0.2633 & 0.0229 & 0.0398 & 0.0401 & 0.0219 & 0.0370 & 0.0370 & 0.0226 & 0.0394 & 0.0407 \\ 
  IPW & 0.8969 & 0.4699 & 0.2638 & 0.0268 & 0.0492 & 0.0477 & 0.0261 & 0.0483 & 0.0486 & 0.0265 & 0.0486 & 0.0484 \\ 
  SPE & 0.8967 & 0.4703 & 0.2635 & 0.0260 & 0.0454 & 0.0420 & 0.0254 & 0.0445 & 0.0420 & 0.0257 & 0.0449 & 0.0424 \\ 
 \hline 
       & \multicolumn{12}{c}{cut-point = (5,7)} \\  
  True & 0.9711 & 0.2347 & 0.2635 &  &  &  &  &  &  &  &  &  \\ 
  FI & 0.9710 & 0.2335 & 0.2631 & 0.0086 & 0.0283 & 0.0388 & 0.0084 & 0.0260 & 0.0352 & 0.0088 & 0.0284 & 0.0391 \\ 
  MSI & 0.9711 & 0.2339 & 0.2633 & 0.0116 & 0.0327 & 0.0401 & 0.0111 & 0.0304 & 0.0370 & 0.0117 & 0.0325 & 0.0407 \\ 
  IPW & 0.9711 & 0.2341 & 0.2638 & 0.0144 & 0.0402 & 0.0477 & 0.0136 & 0.0397 & 0.0486 & 0.0143 & 0.0400 & 0.0484 \\ 
  SPE & 0.9711 & 0.2339 & 0.2635 & 0.0142 & 0.0376 & 0.0420 & 0.0135 & 0.0370 & 0.0420 & 0.0141 & 0.0373 & 0.0424 \\ 
   \hline
\end{tabular}
\end{scriptsize}
\end{center}
\end{sidewaystable}

%% the second covariance matrix, n = 1000
\begin{sidewaystable}
\begin{center}
\caption{Simulation results from 5000 replications when both models for $\rho_{k}$ and $\pi$ are correctly specified (Study 1) and the second value of $\Lambda$ is considered. ``True'' denotes the true parameter value. Sample size = 1000.}
\label{scen1:2:1000}
\begin{scriptsize}
\begin{tabular}{r c c c c c c c c c c c c}
  \hline
  & TCF$_1$ & TCF$_2$ & TCF$_3$ & MC.sd$_1$ & MC.sd$_2$ & MC.sd$_3$ & asy.sd$_1$ & asy.sd$_2$ & asy.sd$_3$ & boot.sd$_1$ & boot.sd$_2$ & boot.sd$_3$ \\ 
 \hline 
       & \multicolumn{12}{c}{cut-point = (2,4)} \\
True & 0.5000 & 0.3970 & 0.8970 &  &  &  &  &  &  &  &  &  \\ 
  FI & 0.4997 & 0.3967 & 0.8966 & 0.0248 & 0.0208 & 0.0177 & 0.0230 & 0.0185 & 0.0250 & 0.0250 & 0.0209 & 0.0180 \\ 
  MSI & 0.4997 & 0.3965 & 0.8966 & 0.0257 & 0.0251 & 0.0202 & 0.0240 & 0.0230 & 0.0268 & 0.0259 & 0.0250 & 0.0205 \\ 
  IPW & 0.4994 & 0.3967 & 0.8967 & 0.0323 & 0.0349 & 0.0259 & 0.0327 & 0.0343 & 0.0263 & 0.0327 & 0.0344 & 0.0260 \\ 
  SPE & 0.4997 & 0.3966 & 0.8967 & 0.0279 & 0.0317 & 0.0251 & 0.0281 & 0.0311 & 0.0250 & 0.0282 & 0.0311 & 0.0252 \\ 
 \hline 
       & \multicolumn{12}{c}{cut-point = (2,5)} \\  
  True & 0.5000 & 0.6335 & 0.7365 &  &  &  &  &  &  &  &  &  \\ 
  FI & 0.4997 & 0.6330 & 0.7364 & 0.0248 & 0.0215 & 0.0286 & 0.0230 & 0.0203 & 0.0310 & 0.0250 & 0.0216 & 0.0288 \\ 
  MSI & 0.4997 & 0.6327 & 0.7361 & 0.0257 & 0.0253 & 0.0304 & 0.0240 & 0.0241 & 0.0327 & 0.0259 & 0.0252 & 0.0307 \\ 
  IPW & 0.4994 & 0.6326 & 0.7365 & 0.0323 & 0.0339 & 0.0360 & 0.0327 & 0.0335 & 0.0375 & 0.0327 & 0.0335 & 0.0363 \\ 
  SPE & 0.4997 & 0.6328 & 0.7362 & 0.0279 & 0.0314 & 0.0338 & 0.0281 & 0.0308 & 0.0340 & 0.0282 & 0.0309 & 0.0341 \\ 
 \hline 
       & \multicolumn{12}{c}{cut-point = (2,7)} \\   
  True & 0.5000 & 0.8682 & 0.2635 &  &  &  &  &  &  &  &  &  \\ 
  FI & 0.4997 & 0.8679 & 0.2640 & 0.0248 & 0.0153 & 0.0274 & 0.0230 & 0.0164 & 0.0249 & 0.0250 & 0.0154 & 0.0275 \\ 
  MSI & 0.4997 & 0.8680 & 0.2643 & 0.0257 & 0.0183 & 0.0286 & 0.0240 & 0.0192 & 0.0262 & 0.0259 & 0.0184 & 0.0287 \\ 
  IPW & 0.4994 & 0.8682 & 0.2645 & 0.0323 & 0.0248 & 0.0343 & 0.0327 & 0.0244 & 0.0345 & 0.0327 & 0.0246 & 0.0341 \\ 
  SPE & 0.4997 & 0.8682 & 0.2644 & 0.0279 & 0.0236 & 0.0299 & 0.0281 & 0.0232 & 0.0297 & 0.0282 & 0.0234 & 0.0298 \\ 
 \hline 
       & \multicolumn{12}{c}{cut-point = (4,5)} \\  
  True & 0.8970 & 0.2365 & 0.7365 &  &  &  &  &  &  &  &  &  \\ 
  FI & 0.8971 & 0.2363 & 0.7364 & 0.0144 & 0.0180 & 0.0286 & 0.0138 & 0.0157 & 0.0310 & 0.0143 & 0.0182 & 0.0288 \\ 
  MSI & 0.8971 & 0.2362 & 0.7361 & 0.0160 & 0.0217 & 0.0304 & 0.0155 & 0.0197 & 0.0327 & 0.0160 & 0.0217 & 0.0307 \\ 
  IPW & 0.8972 & 0.2359 & 0.7365 & 0.0188 & 0.0297 & 0.0360 & 0.0186 & 0.0291 & 0.0375 & 0.0187 & 0.0293 & 0.0363 \\ 
  SPE & 0.8972 & 0.2362 & 0.7362 & 0.0183 & 0.0264 & 0.0338 & 0.0181 & 0.0262 & 0.0340 & 0.0182 & 0.0262 & 0.0341 \\ 
 \hline 
       & \multicolumn{12}{c}{cut-point = (4,7)} \\  
  True & 0.8970 & 0.4711 & 0.2635 &  &  &  &  &  &  &  &  &  \\ 
  FI & 0.8971 & 0.4712 & 0.2640 & 0.0144 & 0.0252 & 0.0274 & 0.0138 & 0.0232 & 0.0249 & 0.0143 & 0.0250 & 0.0275 \\ 
  MSI & 0.8971 & 0.4715 & 0.2643 & 0.0160 & 0.0280 & 0.0286 & 0.0155 & 0.0261 & 0.0262 & 0.0160 & 0.0278 & 0.0287 \\ 
  IPW & 0.8972 & 0.4715 & 0.2645 & 0.0188 & 0.0348 & 0.0343 & 0.0186 & 0.0342 & 0.0345 & 0.0187 & 0.0343 & 0.0341 \\ 
  SPE & 0.8972 & 0.4717 & 0.2644 & 0.0183 & 0.0321 & 0.0299 & 0.0181 & 0.0316 & 0.0297 & 0.0182 & 0.0316 & 0.0298 \\ 
  \hline 
        & \multicolumn{12}{c}{cut-point = (5,7)} \\
  True & 0.9711 & 0.2347 & 0.2635 &  &  &  &  &  &  &  &  &  \\ 
  FI & 0.9709 & 0.2350 & 0.2640 & 0.0061 & 0.0201 & 0.0274 & 0.0060 & 0.0184 & 0.0249 & 0.0062 & 0.0200 & 0.0275 \\ 
  MSI & 0.9709 & 0.2353 & 0.2643 & 0.0082 & 0.0229 & 0.0286 & 0.0080 & 0.0216 & 0.0262 & 0.0082 & 0.0229 & 0.0287 \\ 
  IPW & 0.9709 & 0.2356 & 0.2645 & 0.0101 & 0.0285 & 0.0343 & 0.0099 & 0.0282 & 0.0345 & 0.0102 & 0.0283 & 0.0341 \\ 
  SPE & 0.9710 & 0.2354 & 0.2644 & 0.0100 & 0.0266 & 0.0299 & 0.0098 & 0.0263 & 0.0297 & 0.0100 & 0.0264 & 0.0298 \\ 
   \hline
\end{tabular}
\end{scriptsize}
\end{center}
\end{sidewaystable}

%% third covariance matrix, n = 500
\begin{sidewaystable}
\begin{center}
\caption{Simulation results from 5000 replications when both models for $\rho_{k}$ and $\pi$ are correctly specified (Study 1) and the third value of $\Lambda$ is considered. ``True'' denotes the true parameter value. Sample size = 500.}
\label{scen1:3:500}
\begin{scriptsize}
\begin{tabular}{r c c c c c c c c c c c c}
  \hline
  & TCF$_1$ & TCF$_2$ & TCF$_3$ & MC.sd$_1$ & MC.sd$_2$ & MC.sd$_3$ & asy.sd$_1$ & asy.sd$_2$ & asy.sd$_3$ & boot.sd$_1$ & boot.sd$_2$ & boot.sd$_3$ \\ 
  \hline 
        & \multicolumn{12}{c}{cut-point = (2,4)} \\
True & 0.5000 & 0.3031 & 0.8031 &  &  &  &  &  &  &  &  &  \\ 
  FI & 0.5001 & 0.3027 & 0.8034 & 0.0356 & 0.0240 & 0.0348 & 0.0320 & 0.0206 & 0.0373 & 0.0350 & 0.0242 & 0.0346 \\ 
  MSI & 0.5001 & 0.3031 & 0.8034 & 0.0375 & 0.0322 & 0.0384 & 0.0340 & 0.0291 & 0.0408 & 0.0367 & 0.0318 & 0.0383 \\ 
  IPW & 0.5004 & 0.3031 & 0.8037 & 0.0454 & 0.0444 & 0.0454 & 0.0438 & 0.0439 & 0.0451 & 0.0440 & 0.0441 & 0.0452 \\ 
  SPE & 0.5002 & 0.3032 & 0.8036 & 0.0410 & 0.0411 & 0.0441 & 0.0400 & 0.0406 & 0.0438 & 0.0401 & 0.0407 & 0.0440 \\ 
  \hline 
        & \multicolumn{12}{c}{cut-point = (2,5)} \\  
  True & 0.5000 & 0.4682 & 0.6651 &  &  &  &  &  &  &  &  &  \\ 
  FI & 0.5001 & 0.4681 & 0.6656 & 0.0356 & 0.0266 & 0.0438 & 0.0320 & 0.0237 & 0.0430 & 0.0350 & 0.0271 & 0.0428 \\ 
  MSI & 0.5001 & 0.4679 & 0.6654 & 0.0375 & 0.0348 & 0.0469 & 0.0340 & 0.0325 & 0.0461 & 0.0367 & 0.0350 & 0.0459 \\ 
  IPW & 0.5004 & 0.4676 & 0.6655 & 0.0454 & 0.0475 & 0.0538 & 0.0438 & 0.0474 & 0.0526 & 0.0440 & 0.0476 & 0.0524 \\ 
  SPE & 0.5002 & 0.4678 & 0.6654 & 0.0410 & 0.0440 & 0.0513 & 0.0400 & 0.0440 & 0.0500 & 0.0401 & 0.0442 & 0.0503 \\ 
  \hline 
        & \multicolumn{12}{c}{cut-point = (2,7)} \\   
  True & 0.5000 & 0.7027 & 0.3349 &  &  &  &  &  &  &  &  &  \\ 
  FI & 0.5001 & 0.7033 & 0.3346 & 0.0356 & 0.0268 & 0.0424 & 0.0320 & 0.0246 & 0.0383 & 0.0350 & 0.0267 & 0.0412 \\ 
  MSI & 0.5001 & 0.7033 & 0.3346 & 0.0375 & 0.0336 & 0.0455 & 0.0340 & 0.0318 & 0.0414 & 0.0367 & 0.0334 & 0.0441 \\ 
  IPW & 0.5004 & 0.7031 & 0.3352 & 0.0454 & 0.0439 & 0.0515 & 0.0438 & 0.0437 & 0.0505 & 0.0440 & 0.0440 & 0.0504 \\ 
  SPE & 0.5002 & 0.7034 & 0.3347 & 0.0410 & 0.0416 & 0.0481 & 0.0400 & 0.0413 & 0.0465 & 0.0401 & 0.0414 & 0.0468 \\ 
  \hline 
        & \multicolumn{12}{c}{cut-point = (4,5)} \\  
  True & 0.8031 & 0.1651 & 0.6651 &  &  &  &  &  &  &  &  &  \\ 
  FI & 0.8033 & 0.1654 & 0.6656 & 0.0278 & 0.0196 & 0.0438 & 0.0260 & 0.0166 & 0.0430 & 0.0274 & 0.0196 & 0.0428 \\ 
  MSI & 0.8030 & 0.1648 & 0.6654 & 0.0303 & 0.0256 & 0.0469 & 0.0284 & 0.0236 & 0.0461 & 0.0297 & 0.0259 & 0.0459 \\ 
  IPW & 0.8030 & 0.1645 & 0.6655 & 0.0344 & 0.0346 & 0.0538 & 0.0335 & 0.0346 & 0.0526 & 0.0337 & 0.0349 & 0.0524 \\ 
  SPE & 0.8030 & 0.1645 & 0.6654 & 0.0334 & 0.0317 & 0.0513 & 0.0325 & 0.0321 & 0.0500 & 0.0326 & 0.0322 & 0.0503 \\ 
  \hline 
        & \multicolumn{12}{c}{cut-point = (4,7)} \\   
  True & 0.8031 & 0.3996 & 0.3349 &  &  &  &  &  &  &  &  &  \\ 
  FI & 0.8033 & 0.4007 & 0.3346 & 0.0278 & 0.0300 & 0.0424 & 0.0260 & 0.0268 & 0.0383 & 0.0274 & 0.0299 & 0.0412 \\ 
  MSI & 0.8030 & 0.4002 & 0.3346 & 0.0303 & 0.0367 & 0.0455 & 0.0284 & 0.0339 & 0.0414 & 0.0297 & 0.0364 & 0.0441 \\ 
  IPW & 0.8030 & 0.4000 & 0.3352 & 0.0344 & 0.0458 & 0.0515 & 0.0335 & 0.0456 & 0.0505 & 0.0337 & 0.0458 & 0.0504 \\ 
  SPE & 0.8030 & 0.4002 & 0.3347 & 0.0334 & 0.0431 & 0.0481 & 0.0325 & 0.0429 & 0.0465 & 0.0326 & 0.0430 & 0.0468 \\ 
  \hline 
        & \multicolumn{12}{c}{cut-point = (5,7)} \\ 
  True & 0.8996 & 0.2345 & 0.3349 &  &  &  &  &  &  &  &  &  \\ 
  FI & 0.8996 & 0.2353 & 0.3346 & 0.0192 & 0.0245 & 0.0424 & 0.0182 & 0.0220 & 0.0383 & 0.0190 & 0.0248 & 0.0412 \\ 
  MSI & 0.8996 & 0.2354 & 0.3346 & 0.0221 & 0.0307 & 0.0455 & 0.0212 & 0.0288 & 0.0414 & 0.0219 & 0.0310 & 0.0441 \\ 
  IPW & 0.8997 & 0.2355 & 0.3352 & 0.0253 & 0.0384 & 0.0515 & 0.0249 & 0.0388 & 0.0505 & 0.0252 & 0.0391 & 0.0504 \\ 
  SPE & 0.8997 & 0.2356 & 0.3347 & 0.0249 & 0.0364 & 0.0481 & 0.0246 & 0.0366 & 0.0465 & 0.0247 & 0.0367 & 0.0468 \\ 
   \hline
\end{tabular}
\end{scriptsize}
\end{center}
\end{sidewaystable}

%% third covariance matrix, n = 1000
\begin{sidewaystable}
\begin{center}
\caption{Simulation results from 5000 replications when both models for $\rho_{k}$ and $\pi$ are correctly specified (Study 1) and the third value of $\Lambda$ is considered. ``True'' denotes the true parameter value. Sample size = 1000.}
\label{scen1:3:1000}
\begin{scriptsize}
\begin{tabular}{r c c c c c c c c c c c c}
  \hline
  & TCF$_1$ & TCF$_2$ & TCF$_3$ & MC.sd$_1$ & MC.sd$_2$ & MC.sd$_3$ & asy.sd$_1$ & asy.sd$_2$ & asy.sd$_3$ & boot.sd$_1$ & boot.sd$_2$ & boot.sd$_3$ \\
  \hline 
        & \multicolumn{12}{c}{cut-point = (2,4)} \\
True & 0.5000 & 0.3031 & 0.8031 &  &  &  &  &  &  &  &  &  \\ 
  FI & 0.5003 & 0.3030 & 0.8040 & 0.0242 & 0.0169 & 0.0243 & 0.0226 & 0.0145 & 0.0264 & 0.0247 & 0.0170 & 0.0243 \\ 
  MSI & 0.5001 & 0.3030 & 0.8038 & 0.0256 & 0.0222 & 0.0270 & 0.0240 & 0.0206 & 0.0288 & 0.0259 & 0.0224 & 0.0270 \\ 
  IPW & 0.5001 & 0.3032 & 0.8038 & 0.0310 & 0.0310 & 0.0320 & 0.0310 & 0.0311 & 0.0320 & 0.0311 & 0.0311 & 0.0319 \\ 
  SPE & 0.5001 & 0.3030 & 0.8040 & 0.0281 & 0.0285 & 0.0312 & 0.0283 & 0.0288 & 0.0310 & 0.0283 & 0.0287 & 0.0311 \\ 
  \hline 
        & \multicolumn{12}{c}{cut-point = (2,5)} \\   
  True & 0.5000 & 0.4682 & 0.6651 &  &  &  &  &  &  &  &  &  \\ 
  FI & 0.5003 & 0.4682 & 0.6663 & 0.0242 & 0.0193 & 0.0301 & 0.0226 & 0.0167 & 0.0304 & 0.0247 & 0.0191 & 0.0301 \\ 
  MSI & 0.5001 & 0.4681 & 0.6663 & 0.0256 & 0.0248 & 0.0320 & 0.0240 & 0.0230 & 0.0326 & 0.0259 & 0.0247 & 0.0324 \\ 
  IPW & 0.5001 & 0.4683 & 0.6664 & 0.0310 & 0.0337 & 0.0368 & 0.0310 & 0.0336 & 0.0373 & 0.0311 & 0.0336 & 0.0370 \\ 
  SPE & 0.5001 & 0.4682 & 0.6665 & 0.0281 & 0.0311 & 0.0350 & 0.0283 & 0.0312 & 0.0355 & 0.0283 & 0.0312 & 0.0355 \\ 
  \hline 
        & \multicolumn{12}{c}{cut-point = (2,7)} \\  
  True & 0.5000 & 0.7027 & 0.3349 &  &  &  &  &  &  &  &  &  \\ 
  FI & 0.5003 & 0.7028 & 0.3359 & 0.0242 & 0.0188 & 0.0289 & 0.0226 & 0.0173 & 0.0271 & 0.0247 & 0.0188 & 0.0290 \\ 
  MSI & 0.5001 & 0.7025 & 0.3359 & 0.0256 & 0.0236 & 0.0307 & 0.0240 & 0.0225 & 0.0293 & 0.0259 & 0.0236 & 0.0311 \\ 
  IPW & 0.5001 & 0.7023 & 0.3360 & 0.0310 & 0.0311 & 0.0350 & 0.0310 & 0.0310 & 0.0358 & 0.0311 & 0.0311 & 0.0356 \\ 
  SPE & 0.5001 & 0.7024 & 0.3358 & 0.0281 & 0.0292 & 0.0324 & 0.0283 & 0.0293 & 0.0329 & 0.0283 & 0.0293 & 0.0330 \\ 
  \hline 
        & \multicolumn{12}{c}{cut-point = (4,5)} \\  
  True & 0.8031 & 0.1651 & 0.6651 &  &  &  &  &  &  &  &  &  \\ 
  FI & 0.8034 & 0.1652 & 0.6663 & 0.0193 & 0.0139 & 0.0301 & 0.0184 & 0.0117 & 0.0304 & 0.0193 & 0.0138 & 0.0301 \\ 
  MSI & 0.8032 & 0.1652 & 0.6663 & 0.0211 & 0.0184 & 0.0320 & 0.0201 & 0.0167 & 0.0326 & 0.0209 & 0.0183 & 0.0324 \\ 
  IPW & 0.8034 & 0.1651 & 0.6664 & 0.0241 & 0.0248 & 0.0368 & 0.0237 & 0.0246 & 0.0373 & 0.0237 & 0.0247 & 0.0370 \\ 
  SPE & 0.8032 & 0.1653 & 0.6665 & 0.0233 & 0.0229 & 0.0350 & 0.0230 & 0.0228 & 0.0355 & 0.0230 & 0.0228 & 0.0355 \\ 
  \hline 
        & \multicolumn{12}{c}{cut-point = (4,7)} \\  
  True & 0.8031 & 0.3996 & 0.3349 &  &  &  &  &  &  &  &  &  \\ 
  FI & 0.8034 & 0.3998 & 0.3359 & 0.0193 & 0.0207 & 0.0289 & 0.0184 & 0.0189 & 0.0271 & 0.0193 & 0.0210 & 0.0290 \\ 
  MSI & 0.8032 & 0.3995 & 0.3359 & 0.0211 & 0.0253 & 0.0307 & 0.0201 & 0.0240 & 0.0293 & 0.0209 & 0.0256 & 0.0311 \\ 
  IPW & 0.8034 & 0.3991 & 0.3360 & 0.0241 & 0.0319 & 0.0350 & 0.0237 & 0.0323 & 0.0358 & 0.0237 & 0.0323 & 0.0356 \\ 
  SPE & 0.8032 & 0.3994 & 0.3358 & 0.0233 & 0.0299 & 0.0324 & 0.0230 & 0.0303 & 0.0329 & 0.0230 & 0.0304 & 0.0330 \\ 
  \hline 
        & \multicolumn{12}{c}{cut-point = (5,7)} \\   
  True & 0.8996 & 0.2345 & 0.3349 &  &  &  &  &  &  &  &  &  \\ 
  FI & 0.8998 & 0.2346 & 0.3359 & 0.0134 & 0.0172 & 0.0289 & 0.0129 & 0.0155 & 0.0271 & 0.0134 & 0.0174 & 0.0290 \\ 
  MSI & 0.8997 & 0.2343 & 0.3359 & 0.0157 & 0.0216 & 0.0307 & 0.0150 & 0.0204 & 0.0293 & 0.0155 & 0.0218 & 0.0311 \\ 
  IPW & 0.8998 & 0.2340 & 0.3360 & 0.0180 & 0.0273 & 0.0350 & 0.0177 & 0.0274 & 0.0358 & 0.0177 & 0.0275 & 0.0356 \\ 
  SPE & 0.8997 & 0.2342 & 0.3358 & 0.0178 & 0.0256 & 0.0324 & 0.0174 & 0.0258 & 0.0329 & 0.0174 & 0.0259 & 0.0330 \\ 
   \hline
\end{tabular}
\end{scriptsize}
\end{center}
\end{sidewaystable}

\begin{sidewaystable}
\begin{center}
\caption{Simulation results from 5000 replications  when the  model for the verification process is misspecified (Study 2) and the first value of $\Lambda$ is used. ``True'' indicates the true parameter value. Sample size = 1000.} 
\label{scen2:1:1000}
\begin{scriptsize}
\begin{tabular}{r c c c c c c c c c c c c}
  \hline
  & TCF$_1$ & TCF$_2$ & TCF$_3$ & MC.sd$_1$ & MC.sd$_2$ & MC.sd$_3$ & asy.sd$_1$ & asy.sd$_2$ & asy.sd$_3$ & boot.sd$_1$ & boot.sd$_2$ & boot.sd$_3$ \\
  \hline 
       & \multicolumn{12}{c}{cut-point = (2,4)} \\
True & 0.5000 & 0.4347 & 0.9347 &  &  &  &  &  &  &  &  &  \\ 
  FI & 0.5011 & 0.4345 & 0.9351 & 0.0277 & 0.0238 & 0.0152 & 0.0239 & 0.0203 & 0.0262 & 0.0275 & 0.0241 & 0.0153 \\ 
  MSI & 0.5011 & 0.4344 & 0.9351 & 0.0280 & 0.0259 & 0.0169 & 0.0243 & 0.0226 & 0.0271 & 0.0279 & 0.0260 & 0.0169 \\ 
  IPW & \textbf{0.5822} & \textbf{0.4436} & 0.9375 & 0.0381 & 0.0407 & 0.0213 & 0.0380 & 0.0400 & 0.0255 & 0.0381 & 0.0401 & 0.0212 \\ 
  SPE & 0.5011 & 0.4345 & 0.9352 & 0.0304 & 0.0334 & 0.0218 & 0.0304 & 0.0330 & 0.0214 & 0.0305 & 0.0331 & 0.0216 \\ 
  \hline 
       & \multicolumn{12}{c}{cut-point = (2,5)} \\  
  True & 0.5000 & 0.7099 & 0.7752 &  &  &  &  &  &  &  &  &  \\ 
  FI & 0.5011 & 0.7105 & 0.7765 & 0.0277 & 0.0227 & 0.0297 & 0.0239 & 0.0223 & 0.0334 & 0.0275 & 0.0228 & 0.0298 \\ 
  MSI & 0.5011 & 0.7101 & 0.7762 & 0.0280 & 0.0245 & 0.0305 & 0.0243 & 0.0241 & 0.0343 & 0.0279 & 0.0245 & 0.0309 \\ 
  IPW & \textbf{0.5822} & \textbf{0.6815} & \textbf{0.8046} & 0.0381 & 0.0376 & 0.0325 & 0.0380 & 0.0370 & 0.0381 & 0.0381 & 0.0371 & 0.0327 \\ 
  SPE & 0.5011 & 0.7099 & 0.7760 & 0.0304 & 0.0309 & 0.0328 & 0.0304 & 0.0306 & 0.0330 & 0.0305 & 0.0307 & 0.0331 \\ 
  \hline 
       & \multicolumn{12}{c}{cut-point = (2,7)} \\  
  True & 0.5000 & 0.9230 & 0.2248 &  &  &  &  &  &  &  &  &  \\ 
  FI & 0.5011 & 0.9233 & 0.2256 & 0.0277 & 0.0144 & 0.0270 & 0.0239 & 0.0193 & 0.0250 & 0.0275 & 0.0143 & 0.0270 \\ 
  MSI & 0.5011 & 0.9234 & 0.2258 & 0.0280 & 0.0161 & 0.0275 & 0.0243 & 0.0204 & 0.0256 & 0.0279 & 0.0158 & 0.0275 \\ 
  IPW & \textbf{0.5822} & \textbf{0.9009} & 0.2306 & 0.0381 & 0.0276 & 0.0306 & 0.0380 & 0.0268 & 0.0316 & 0.0381 & 0.0270 & 0.0308 \\ 
  SPE & 0.5011 & 0.9234 & 0.2258 & 0.0304 & 0.0225 & 0.0279 & 0.0304 & 0.0218 & 0.0280 & 0.0305 & 0.0220 & 0.0281 \\ 
  \hline 
       & \multicolumn{12}{c}{cut-point = (4,5)} \\   
  True & 0.9347 & 0.2752 & 0.7752 &  &  &  &  &  &  &  &  &  \\ 
  FI & 0.9352 & 0.2760 & 0.7765 & 0.0135 & 0.0218 & 0.0297 & 0.0127 & 0.0168 & 0.0334 & 0.0135 & 0.0215 & 0.0298 \\ 
  MSI & 0.9352 & 0.2757 & 0.7762 & 0.0143 & 0.0237 & 0.0305 & 0.0135 & 0.0191 & 0.0343 & 0.0143 & 0.0233 & 0.0309 \\ 
  IPW & \textbf{0.9540} & \textbf{0.2379} & \textbf{0.8046} & 0.0139 & 0.0335 & 0.0325 & 0.0138 & 0.0330 & 0.0381 & 0.0139 & 0.0331 & 0.0327 \\ 
  SPE & 0.9352 & 0.2754 & 0.7760 & 0.0161 & 0.0279 & 0.0328 & 0.0160 & 0.0275 & 0.0330 & 0.0161 & 0.0275 & 0.0331 \\ 
  \hline 
       & \multicolumn{12}{c}{cut-point = (4,7)} \\   
  True & 0.9347 & 0.4883 & 0.2248 &  &  &  &  &  &  &  &  &  \\ 
  FI & 0.9352 & 0.4888 & 0.2256 & 0.0135 & 0.0290 & 0.0270 & 0.0127 & 0.0259 & 0.0250 & 0.0135 & 0.0287 & 0.0270 \\ 
  MSI & 0.9352 & 0.4889 & 0.2258 & 0.0143 & 0.0302 & 0.0275 & 0.0135 & 0.0273 & 0.0256 & 0.0143 & 0.0300 & 0.0275 \\ 
  IPW & \textbf{0.9540} & \textbf{0.4574} & 0.2306 & 0.0139 & 0.0391 & 0.0306 & 0.0138 & 0.0387 & 0.0316 & 0.0139 & 0.0388 & 0.0308 \\ 
  SPE & 0.9352 & 0.4890 & 0.2258 & 0.0161 & 0.0328 & 0.0279 & 0.0160 & 0.0327 & 0.0280 & 0.0161 & 0.0328 & 0.0281 \\
  \hline 
       & \multicolumn{12}{c}{cut-point = (5,7)} \\  
  True & 0.9883 & 0.2132 & 0.2248 &  &  &  &  &  &  &  &  &  \\ 
  FI & 0.9883 & 0.2128 & 0.2256 & 0.0040 & 0.0216 & 0.0270 & 0.0038 & 0.0190 & 0.0250 & 0.0040 & 0.0215 & 0.0270 \\ 
  MSI & 0.9884 & 0.2133 & 0.2258 & 0.0050 & 0.0231 & 0.0275 & 0.0046 & 0.0208 & 0.0256 & 0.0050 & 0.0231 & 0.0275 \\ 
  IPW & \textbf{0.9912} & 0.2195 & 0.2306 & 0.0060 & 0.0305 & 0.0306 & 0.0054 & 0.0301 & 0.0316 & 0.0059 & 0.0302 & 0.0308 \\ 
  SPE & 0.9885 & 0.2135 & 0.2258 & 0.0065 & 0.0256 & 0.0279 & 0.0060 & 0.0256 & 0.0280 & 0.0064 & 0.0257 & 0.0281 \\ 
   \hline
\end{tabular}
\end{scriptsize}
\end{center}
\end{sidewaystable}
 
\begin{sidewaystable}
\begin{center}
\caption{Simulation results from 5000 replications when the  model for the verification process is misspecified (Study 2) and the third value of $\Lambda$ is used. ``True'' indicates the true parameter value. Sample size = 1000.}
\label{scen2:3:1000}
\begin{scriptsize}
\begin{tabular}{r c c c c c c c c c c c c}
  \hline
  & TCF$_1$ & TCF$_2$ & TCF$_3$ & MC.sd$_1$ & MC.sd$_2$ & MC.sd$_3$ & asy.sd$_1$ & asy.sd$_2$ & asy.sd$_3$ & boot.sd$_1$ & boot.sd$_2$ & boot.sd$_3$ \\
   \hline 
        & \multicolumn{12}{c}{cut-point = (2,4)} \\
True & 0.5000 & 0.3031 & 0.8031 &  &  &  &  &  &  &  &  &  \\ 
  FI & 0.4998 & 0.3026 & 0.8043 & 0.0257 & 0.0172 & 0.0280 & 0.0221 & 0.0124 & 0.0293 & 0.0259 & 0.0171 & 0.0278 \\ 
  MSI & 0.4999 & 0.3027 & 0.8044 & 0.0264 & 0.0204 & 0.0297 & 0.0230 & 0.0166 & 0.0308 & 0.0267 & 0.0204 & 0.0293 \\ 
  IPW & \textbf{0.6267} & \textbf{0.2614} & \textbf{0.8259} & 0.0345 & 0.0371 & 0.0371 & 0.0346 & 0.0364 & 0.0372 & 0.0348 & 0.0366 & 0.0365 \\ 
  SPE & 0.5000 & 0.3031 & 0.8047 & 0.0322 & 0.0323 & 0.0361 & 0.0324 & 0.0321 & 0.0352 & 0.0326 & 0.0322 & 0.0354 \\ 
   \hline 
        & \multicolumn{12}{c}{cut-point = (2,5)} \\  
  True & 0.5000 & 0.4682 & 0.6651 &  &  &  &  &  &  &  &  &  \\ 
  FI & 0.4998 & 0.4681 & 0.6667 & 0.0257 & 0.0192 & 0.0341 & 0.0221 & 0.0151 & 0.0342 & 0.0259 & 0.0192 & 0.0343 \\ 
  MSI & 0.4999 & 0.4681 & 0.6664 & 0.0264 & 0.0227 & 0.0354 & 0.0230 & 0.0195 & 0.0357 & 0.0267 & 0.0229 & 0.0358 \\ 
  IPW & \textbf{0.6267} & \textbf{0.3884} & \textbf{0.7253} & 0.0345 & 0.0403 & 0.0396 & 0.0346 & 0.0400 & 0.0413 & 0.0348 & 0.0402 & 0.0401 \\ 
  SPE & 0.5000 & 0.4684 & 0.6665 & 0.0322 & 0.0352 & 0.0389 & 0.0324 & 0.0353 & 0.0391 & 0.0326 & 0.0355 & 0.0393 \\ 
   \hline 
        & \multicolumn{12}{c}{cut-point = (2,7)} \\   
  True & 0.5000 & 0.7027 & 0.3349 &  &  &  &  &  &  &  &  &  \\ 
  FI & 0.4998 & 0.7035 & 0.3360 & 0.0257 & 0.0201 & 0.0318 & 0.0221 & 0.0184 & 0.0311 & 0.0259 & 0.0203 & 0.0320 \\ 
  MSI & 0.4999 & 0.7035 & 0.3360 & 0.0264 & 0.0237 & 0.0337 & 0.0230 & 0.0224 & 0.0331 & 0.0267 & 0.0240 & 0.0339 \\ 
  IPW & \textbf{0.6267} & \textbf{0.6157} & \textbf{0.4102} & 0.0345 & 0.0417 & 0.0386 & 0.0346 & 0.0416 & 0.0398 & 0.0348 & 0.0417 & 0.0386 \\ 
  SPE & 0.5000 & 0.7038 & 0.3360 & 0.0322 & 0.0360 & 0.0350 & 0.0324 & 0.0361 & 0.0350 & 0.0326 & 0.0364 & 0.0352 \\ 
   \hline 
        & \multicolumn{12}{c}{cut-point = (4,5)} \\   
  True & 0.8031 & 0.1651 & 0.6651 &  &  &  &  &  &  &  &  &  \\ 
  FI & 0.8032 & 0.1655 & 0.6667 & 0.0207 & 0.0139 & 0.0341 & 0.0189 & 0.0099 & 0.0342 & 0.0207 & 0.0141 & 0.0343 \\ 
  MSI & 0.8031 & 0.1654 & 0.6664 & 0.0217 & 0.0165 & 0.0354 & 0.0200 & 0.0135 & 0.0357 & 0.0216 & 0.0169 & 0.0358 \\ 
  IPW & \textbf{0.8512} & \textbf{0.1270} & \textbf{0.7253} & 0.0217 & 0.0245 & 0.0396 & 0.0215 & 0.0251 & 0.0413 & 0.0215 & 0.0253 & 0.0401 \\ 
  SPE & 0.8030 & 0.1653 & 0.6665 & 0.0239 & 0.0225 & 0.0389 & 0.0237 & 0.0228 & 0.0391 & 0.0238 & 0.0229 & 0.0393 \\ 
   \hline 
        & \multicolumn{12}{c}{cut-point = (4,7)} \\   
  True & 0.8031 & 0.3996 & 0.3349 &  &  &  &  &  &  &  &  &  \\ 
  FI & 0.8032 & 0.4009 & 0.3360 & 0.0207 & 0.0226 & 0.0318 & 0.0189 & 0.0194 & 0.0311 & 0.0207 & 0.0227 & 0.0320 \\ 
  MSI & 0.8031 & 0.4008 & 0.3360 & 0.0217 & 0.0261 & 0.0337 & 0.0200 & 0.0234 & 0.0331 & 0.0216 & 0.0262 & 0.0339 \\ 
  IPW & \textbf{0.8512} & \textbf{0.3544} & \textbf{0.4102} & 0.0217 & 0.0358 & 0.0386 & 0.0215 & 0.0362 & 0.0398 & 0.0215 & 0.0363 & 0.0386 \\ 
  SPE & 0.8030 & 0.4008 & 0.3360 & 0.0239 & 0.0326 & 0.0350 & 0.0237 & 0.0325 & 0.0350 & 0.0238 & 0.0327 & 0.0352 \\ 
   \hline 
        & \multicolumn{12}{c}{cut-point = (5,7)} \\  
  True & 0.8996 & 0.2345 & 0.3349 &  &  &  &  &  &  &  &  &  \\ 
  FI & 0.8997 & 0.2354 & 0.3360 & 0.0144 & 0.0183 & 0.0318 & 0.0135 & 0.0156 & 0.0311 & 0.0144 & 0.0184 & 0.0320 \\ 
  MSI & 0.8995 & 0.2354 & 0.3360 & 0.0158 & 0.0223 & 0.0337 & 0.0149 & 0.0197 & 0.0331 & 0.0157 & 0.0220 & 0.0339 \\ 
  IPW & \textbf{0.9149} & \textbf{0.2274} & \textbf{0.4102} & 0.0163 & 0.0303 & 0.0386 & 0.0160 & 0.0299 & 0.0398 & 0.0161 & 0.0301 & 0.0386 \\ 
  SPE & 0.8995 & 0.2355 & 0.3360 & 0.0175 & 0.0273 & 0.0350 & 0.0173 & 0.0268 & 0.0350 & 0.0174 & 0.0269 & 0.0352 \\ 
   \hline
\end{tabular}
\end{scriptsize}
\end{center}
\end{sidewaystable} 

\section{Simulation results of VUS estimators}\label{app:simu:vus}
In this section, we give some simulation results concerning the estimators of the VUS presented in Subsection \ref{s:vus}.

The disease  $D$ is generated by a trinomial random vector $(D_{1},D_{2},D_{3})$, such that $D_{k}$ is a Bernoulli random variable with mean $\theta_k$, $k = 1,2,3$. We set $\theta_1 = 0.4, \theta_2 = 0.35$ and $\theta_3 = 0.25$. The pairs $T, A$ are generated from the following conditional models
\[
T,A |D_{k} \sim \mathcal{N}_2 \left(\mu_k, \Lambda\right), \qquad k = 1,2,3,
\]
where $\mu_k = k(\mu_T,\mu_A)^\top.$ 
We consider three values of $\Lambda$,
\[
\left(\begin{array}{l l}
1.2 & 1 \\
1 & 1
\end{array}\right) , \qquad
\left(\begin{array}{c c}
1.75 & 0.1 \\
0.1 & 2.5
\end{array}\right) , \qquad
\left(\begin{array}{c c}
5.5 & 3 \\
3 & 2.5
\end{array}\right)
.
\]
The true VUS value is equal to 0.9472 for the first value of $\Lambda$ and $(\mu_T,\mu_A) = (3,2);$  is equal to 0.7175 for the second value of $\Lambda$  and $(\mu_T,\mu_A) = (2,1);$ is equal to 0.4778 for the third value of $\Lambda$  and $(\mu_T,\mu_A) = (2,1).$ We simulate the verification status $V$ by using the following model
\[
\mathrm{logit}\left\{\Pro(V = 1|T,A)\right\} = \delta_0 + \delta_1 T + \delta_2 A.
\]
The  parameters $(\delta_0,\delta_1,\delta_2)$ are fixed equal to $(1,-2.87,4.06)$ when the first value of $\Lambda$ is considered, and equal to  $(1,-2.2,4)$  otherwise. These choices give rise  to a verification rate of about $0.52$. 
Under our data--generating setting, the disease process follows  a multinomial logistic model. 
We consider two sample sizes, i.e., $n = 200$ and $n = 500$. Each simulation experiment was based on 1000 replications.

FI, MSI, IPW and SPE estimates of VUS are computed under correct working models for both the disease and the verification processes. Tables \ref{tab:result:vus1}--\ref{tab:result:vus3} show Monte Carlo means, Monte Carlo standard deviations (MC.sd), the square roots of the variances estimated via asymptotic results (Asy.sd) and bootstrap standard deviations (Boot.sd) of $\hat{\mu}$.

\begin{table}[h]
\caption{Simulation results for VUS estimators, $\mu = 0.9472.$}
\label{tab:result:vus1}
\begin{center}
\begin{tabular}{c l r r r r}
\toprule 
Sample size & ~ &  Mean &  MC.sd  &  Asy.sd  &  Boot.sd \\
\midrule
\multirow{4}{*}{$n = 200$} & FI   & 0.9471 & 0.0251 & 0.0219 & 0.0256 \\
& MSI  & 0.9466 & 0.0252 & 0.0222 & 0.0258 \\
& IPW  & 0.9498 & 0.0377 & 0.0261 & 0.0271 \\
& SPE  & 0.9461 & 0.0323 & 0.0274 & 0.0315 \\
\midrule
\multirow{4}{*}{$n = 500$} & FI   & 0.9470 & 0.0144 & 0.0143 & 0.0149 \\
& MSI  & 0.9468 & 0.0144 & 0.0144 & 0.0150 \\
& IPW  & 0.9480 & 0.0244 & 0.0192 & 0.0192 \\
& SPE  & 0.9467 & 0.0228 & 0.0181 & 0.0224 \\
\bottomrule
\end{tabular}
\end{center}
\end{table}

\begin{table}[h]
\caption{Simulation results for VUS estimators, $\mu = 0.7175.$}
\label{tab:result:vus2}
\begin{center}
\begin{tabular}{c l r r r r}
\toprule 
Sample size & ~ &  Mean &  MC.sd  &  Asy.sd  &  Boot.sd  \\
\midrule
%\multirow{5}{*}{$n = 200$} &   &   &  &  &  \\ 
\multirow{4}{*}{$n = 200$} & FI   & 0.7185 & 0.0549 & 0.0559 & 0.0566 \\
& MSI  & 0.7165 & 0.0552 & 0.0571 & 0.0577 \\
& IPW  & 0.7261 & 0.0981 & 0.1197 & 0.0754  \\
& SPE  & 0.7155 & 0.1021 & 0.0981 & 0.1106  \\
\midrule
\multirow{4}{*}{$n = 500$} & FI & 0.7183 & 0.0357 & 0.0356 & 0.0357 \\
& MSI  & 0.7176 & 0.0358 & 0.0360 & 0.0361 \\
& IPW  & 0.7272 & 0.0814 & 0.0549 & 0.0564 \\
& SPE  & 0.7184 & 0.0813 & 0.0698 & 0.0864  \\
\bottomrule
\end{tabular}
\end{center}
\end{table}

\begin{table}[h]
\caption{Simulation results for VUS estimators, $\mu = 0.4778.$}
\label{tab:result:vus3}
\begin{center}
\begin{tabular}{c l r r r r}
\toprule
Sample size & ~ &  Mean &  MC.sd  &  Asy.sd  &  Boot.sd  \\
\midrule
\multirow{4}{*}{$n = 200$} & FI   & 0.4788 & 0.0575 & 0.0558 & 0.0574  \\
& MSI  & 0.4775 & 0.0584 & 0.0576 & 0.0589 \\
& IPW  & 0.4760 & 0.1054 & 0.0767 & 0.0876 \\
& SPE  & 0.4815 & 0.1121 & 0.1472 & 0.1418 \\
\midrule
\multirow{4}{*}{$n = 500$} & FI   & 0.4782 & 0.0360 & 0.0350 & 0.0354 \\
& MSI  & 0.4779 & 0.0364 & 0.0358 & 0.0361 \\
& IPW  & 0.4804 & 0.0792 & 0.0608 & 0.0640 \\
& SPE  & 0.4868 & 0.0943 & 0.1101 & 0.0995 \\
\bottomrule
\end{tabular}
\end{center}
\end{table}

\section{Some figures related to the first illustration}
\label{app:fg:ex1}
In this section, we provide some extra plots related to the analysis of the first dataset used in the main paper.
In particular, in Figure~\ref{fg:full:ex1} we present the estimate of the ROC surface for the test CA125 based
on the full data set.
\begin{figure}[h]
\begin{center}
\includegraphics[width=0.44\linewidth]{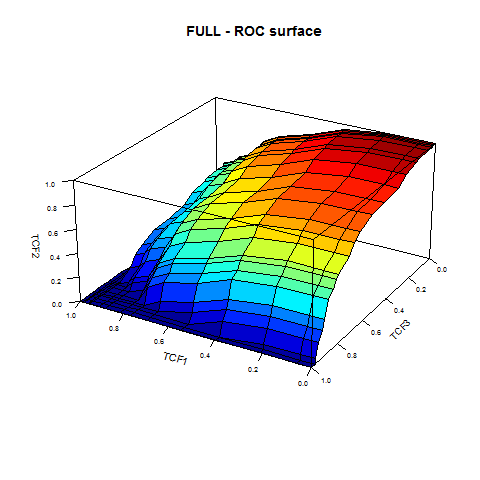}
\caption{Estimated ROC surface for CA125 assessing the classification into three class of EOC: benign disease, early stage (I and II) and late stage (I and II). This surface is estimated by using full data.}
\label{fg:full:ex1}
\end{center}
\end{figure}
Figure~\ref{fg:roc_threshold_projection} and Figure~\ref{fg:roc_logistic_projection} present  
the projections of the estimated ROC surfaces
to the planes defined by $\TCF_1$ versus $\TCF_2$, $\TCF_1$ versus $\TCF_3$ and $\TCF_2$ versus $\TCF_3$, i.e., the ROC curves between  classes 1 and 2, classes 1 and 3, classes 2 and 3. For the IPW and SPE methods, to estimate the verification process, we make use, firstly, of a correctly specified model, i.e., a linear threshold regression model (Figure~\ref{fg:roc_threshold_projection}) and, then, of a misspecified model, i.e., a logistic model (Figure~\ref{fg:roc_logistic_projection}).
\begin{figure}[h]
\begin{center}
\includegraphics[width=0.46\linewidth]{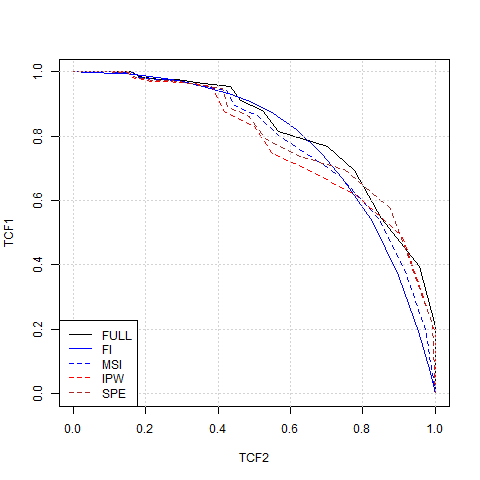}\\[-1.0em]
\begin{tabular}{@{}c@{}c@{}}
  \multicolumn{2}{c}{ }\\[-0.2em]
  \includegraphics[width=0.46\linewidth]{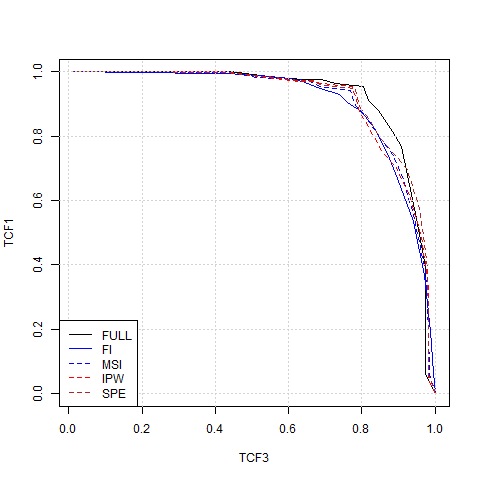}&
  \includegraphics[width=0.46\linewidth]{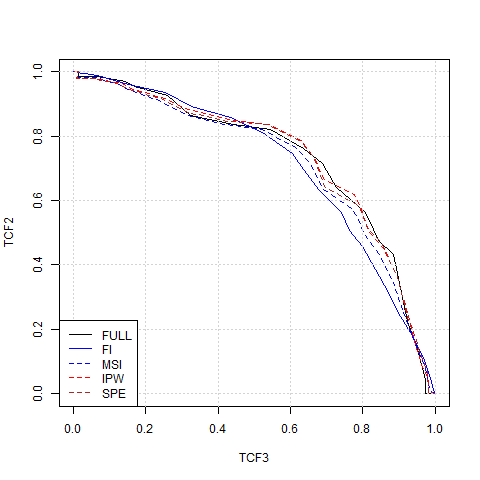}\\[-0.5em]
\end{tabular}
\caption{Two dimensional ROC curve projections. A threshold model is used to estimate the verification process.}
\label{fg:roc_threshold_projection}
\end{center}
\end{figure}
\begin{figure}[h]
\begin{center}
\includegraphics[width=0.46\linewidth]{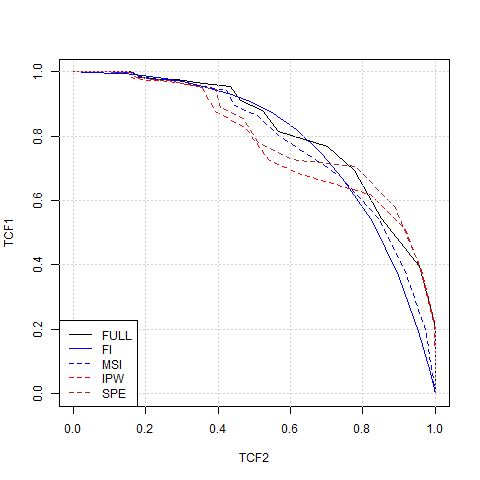}\\[-1.0em]
\begin{tabular}{@{}c@{}c@{}}
  \multicolumn{2}{c}{ }\\[-0.2em]
  \includegraphics[width=0.46\linewidth]{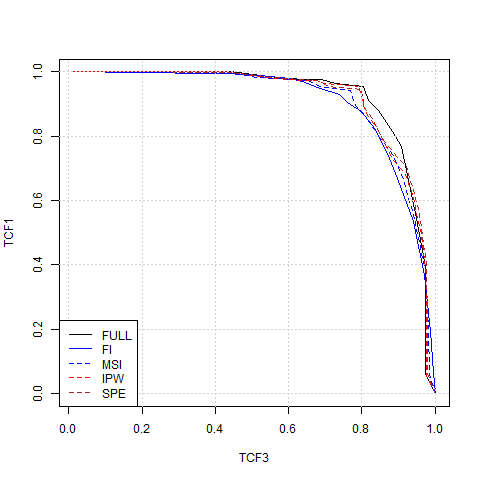}&
  \includegraphics[width=0.46\linewidth]{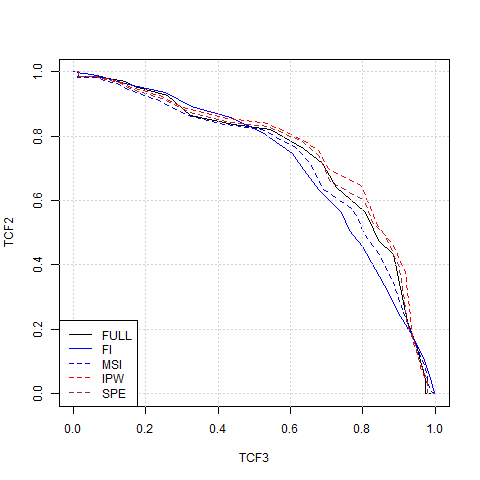}\\[-0.5em]
\end{tabular}
\caption{Two dimensional ROC curve projections. A logistic model is used to estimate the verification process.}
\label{fg:roc_logistic_projection}
\end{center}
\end{figure}
Finally, as an example, Figure~\ref{fg:conf_region}  plots confidence regions for the pair $(\TCF_1(c_1), \TCF_2(c_1,+\infty))$ at 
three   values of $c_1$, when the MSI approach is used.
An approximated 95\% elliptical confidence region is obtained in a standard way as the set of points
\begin{align}
R_{12,\mathrm{MSI}} &= \Bigg\{ \begin{pmatrix}
\TCF_{1}(c_1) \\ \TCF_{2}(c_1,+\infty)
\end{pmatrix}: 
\begin{pmatrix}
\ud\TCF_{1}(c_1) \\ \ud\TCF_{2}(c_1,+\infty)
\end{pmatrix}^\top \hat\Sigma_{12}^{-1}
\begin{pmatrix}
\ud\TCF_{1}(c_1) \\ \ud\TCF_{2}(c_1,+\infty)
\end{pmatrix} \nonumber \\
& \qquad \qquad \qquad \qquad \qquad \qquad \le \chi^2_{0.95,2}; c_1 \in \mathbb{R}
\Bigg\}, \nonumber
\end{align}
where
\[
\begin{pmatrix}
\ud\TCF_{1}(c_1) \\ \ud\TCF_{2}(c_1,+\infty)
\end{pmatrix} = \begin{pmatrix}
\TCF_{1}(c_1) \\ \TCF_{2}(c_1,+\infty)
\end{pmatrix} - 
\begin{pmatrix}
\widehat{\TCF}_{1,\mathrm{MSI}}(c_1) \\ \widehat{\TCF}_{2,\mathrm{MSI}}(c_1,+\infty)
\end{pmatrix},
\]
the quantity $\hat\Sigma_{12}$ is the estimated asymptotic covariance matrix of $\Big(\widehat{\TCF}_{1,\mathrm{MSI}}(c_1),$ $ \widehat{\TCF}_{2,\mathrm{MSI}}(c_1,+\infty)\Big)$ and $\chi^2_{0.95,2}$ is the 95--th quantile of a Chi--square distribution with $2$ degree of freedom. In the plot, the black solid line represents the full data estimated ROC curve, whereas the blue dashed line is the bias--corrected estimated ROC curve. 
\begin{figure}[h]
\begin{center}
\includegraphics[width=0.46\linewidth]{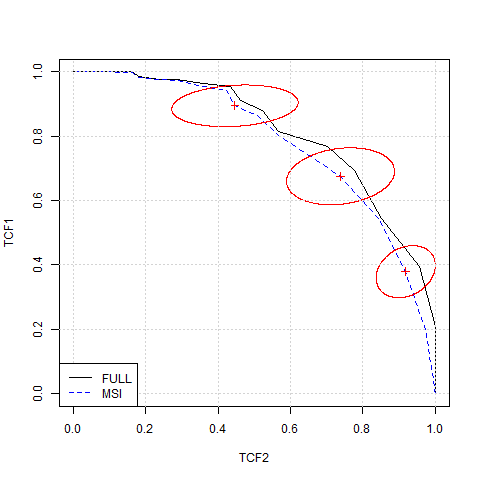}
\caption{ROC curve between  classes 1 and 2 estimated using the MSI approach
along with approximate 0.95 confidence regions for $c_1=-0.237, -0.399$ and 1.672, respectively.}
\label{fg:conf_region}
\end{center}
\end{figure}

\end{document}